\documentclass[11pt,a4paper]{article}
\usepackage[a4paper]{geometry}
\usepackage{amssymb,latexsym,amsmath,amsfonts,amsthm}
\usepackage{graphicx}
\usepackage{epstopdf}
\usepackage{tikz}
\usepackage{pgflibraryshapes}
\usetikzlibrary{arrows,decorations.markings}
\usepackage{subfigure}
\usepackage{overpic}
\usepackage{fullpage}
\usepackage{comment,verbatim}
\usepackage{hyperref}
\usepackage[title,titletoc]{appendix}
\usepackage{float}
\usepackage{authblk}

\newtheorem{thm}{Theorem}[section]

\newtheorem{pro}[thm]{Proposition}

\theoremstyle{remark}
\newtheorem{rem}[thm]{Remark}

\numberwithin{equation}{section}

\def\Im {\mathop{\rm Im}\nolimits}
\def\Re {\mathop{\rm Re}\nolimits}
\def\z {\zeta}

\begin{document}

\title{ Clarkson-McLeod solutions of the fourth Painlev\'e equation and the parabolic cylinder-kernel determinant}

\author[a,c]{Jun Xia}
\author[b]{Shuai-Xia Xu}
\author[c]{Yu-Qiu Zhao}

\affil[a]{ Department of Mathematics, Jiaying University, MeiZhou 514015, China.~Email:~xiaj7@mail2.sysu.edu.cn}
\affil[b]{Institut Franco-Chinois de l'Energie Nucl\'{e}aire, Sun Yat-sen University, GuangZhou 510275, China.~Email:~xushx3@mail.sysu.edu.cn}
\affil[c]{ Department of Mathematics, Sun Yat-sen University, GuangZhou 510275, China.~Email:~stszyq@mail.sysu.edu.cn}

\date{}
\maketitle

\begin{abstract}
  The Clarkson-McLeod solutions of the fourth Painlev\'e equation behave like
  $\kappa D_{\alpha-\frac{1}{2}}^2(\sqrt{2}x)$ as $x\rightarrow +\infty$,  where $\kappa$ is some real constant and $D_{\alpha-\frac{1}{2}}(x)$ is the parabolic cylinder function. Using the Deift-Zhou nonlinear steepest descent method, we derive the asymptotic behaviors for this class of solutions as $x\to-\infty$. This completes a proof of Clarkson and McLeod's conjecture on the asymptotics of this family of solutions. The total integrals of the Clarkson-McLeod solutions  and the  asymptotic approximations of the $\sigma$-form of  this family of solutions are also derived. Furthermore, we find a determinantal  representation of   the $\sigma$-form of  the Clarkson-McLeod solutions  via an integrable operator with the parabolic cylinder kernel.

\end{abstract}
\tableofcontents

\noindent


\section{Introduction and statement of results}
This paper is concerned with the fourth Painlev\'e (PIV) equation
\begin{equation}\label{PIV}
\frac{\mathrm{d}^2q}{\mathrm{d}x^2}
=\frac{1}{2q}\left(\frac{\mathrm{d}q}{\mathrm{d}x}\right)^{2}
+\frac{3}{2}q^{3}+4xq^{2}+2(x^{2}-2\alpha)q,
\end{equation}
where 
  the parameter
$\alpha\in\mathbb{R}$. We are interested in the real solutions of \eqref{PIV} 
satisfying the asymptotic condition
\begin{equation}\label{Boundcondit}
q(x)\rightarrow 0\quad \mathrm{as}\quad x\rightarrow +\infty.
\end{equation}

A result due to Bassom \emph{et al.} \cite{BCHM} shows that any real solution of the PIV equation \eqref{PIV}, fulfilling the boundary condition \eqref{Boundcondit}, possesses the following asymptotic behavior
\begin{equation}\label{qasy}
q(x)\sim \kappa D_{\alpha-\frac{1}{2}}^{2}(\sqrt{2}x)\quad \mathrm{as}\quad x\to+\infty
\end{equation}
for some constant $\kappa\in\mathbb{R}$. Here, $D_{\mu}(x)$ denotes the parabolic cylinder function,
which is the solution of Weber's equation 
 (cf. \cite[Chapter 12]{NIST})
\begin{align*}
\frac{\mathrm{d}^2D_{\mu}(x)}{\mathrm{d}x^2}=\left(\frac{1}{4}x^2
-\mu-\frac{1}{2}\right)D_{\mu}(x),
\end{align*}
uniquely characterized by the asymptotic property
\begin{equation}\label{eq:DAsy}
D_{\mu}(x)=x^{\mu}e^{-\frac{1}{4}x^2}\left(1+O\left(x^{-2}\right)\right),\quad \mathrm{as}\quad x\to+\infty.
\end{equation}
Conversely, 
for any constant $\kappa\in\mathbb{R}$, there exists a unique solution to the PIV equation \eqref{PIV} asymptotic to $\kappa D_{\alpha-\frac{1}{2}}^{2}(\sqrt{2}x) $
as $x\to +\infty$. 
For convenience,      we   denote such  solutions  by $q(x)=q(x;\alpha,\kappa)$.

Concerning the asymptotic behaviors of $q(x;\alpha,\kappa)$ as $x\rightarrow-\infty$, there has been the following conjecture.
\subsection*{Conjecture (Clarkson-McLeod \cite{CM})}
There exists a constant $\kappa^*>0$ such that:
\begin{description}
\item{(a)} When $0<\kappa<\kappa^{*}$ and as $x\rightarrow-\infty$, we have
\begin{equation}\label{qAsy1}
q(x;\alpha,\kappa)\sim c_n\,2^{\alpha-\frac{1}{2}} x^{2 \alpha-1} e^{-x^{2}}
\end{equation}
if $\alpha+\frac{1}{2}=n+1\in \mathbb{N}$ with $\mathbb{N}$ being the set of all positive integers; and
\begin{equation}\label{qAsy2}
q(x;\alpha,\kappa)\sim -\frac{2 x}{3}+\frac{4d_1}{\sqrt{3}} \sin \left(\frac{x^{2}}{\sqrt{3}}-\frac{4 d_1^{2}}{\sqrt{3}}\ln(\sqrt{2}|x|)+d_2+O\left(x^{-2}\right)\right)+O\left(x^{-1}\right)
\end{equation}
if $\alpha-\frac{1}{2}\notin \mathbb{Z}$, where the constants $c_n$, $d_1$, $d_2$ are dependent on $\kappa$.
\item{(b)} When $\kappa=\kappa^{*}$, $q(x;\alpha,\kappa)$ is asymptotic to $-2x$ as $x\rightarrow-\infty$.
\item{(c)} When $\kappa>\kappa^{*}$, $q(x;\alpha,\kappa)$ has a pole at some point $x$ on the real axis. 
\end{description}
The solutions $q(x;\alpha,\kappa)$ in this class are now called  the \emph{Clarkson-McLeod solutions} of the PIV equation \eqref{PIV}.
In the case $\alpha+\frac{1}{2}\in\mathbb{N}$ of part (a), 
the asymptotic formula \eqref{qAsy1} has been proven in \cite{BCHM,CM}. Therein, the values of $c_n$ and $\kappa^{*}$ were explicitly evaluated by
\begin{equation}\label{kappastarinteger}
c_n=\frac{\kappa}{1-\sqrt{\pi}\,n!\,\kappa},\qquad
\kappa^*=\frac{1}{\sqrt{\pi}\,n!}.
\end{equation}
While $\alpha-\frac{1}{2}\notin \mathbb{Z}$, the value of $\kappa^*$ in \cite{CM} was conjectured to be
\begin{equation}\label{kappastar}
\kappa^*=\frac{1}{\sqrt{\pi}\, \Gamma\left(\alpha+\frac{1}{2}\right)}.
\end{equation}
For $\alpha+\frac{1}{2}\in\mathbb{N}$, each  solution  $q(x)=q(x;\alpha,\kappa)$  of the PIV equation \eqref{PIV} can be explicitly expressed in terms of the classical special functions. The  behavior of the solution as $x\to-\infty$ can then be determined in a straightforward manner. For example, when $\alpha=\frac{1}{2}$, the exact solution is given by (see \cite[Equation (4.13)]{BCHM})
\begin{equation}\label{specialsolu}
q\left(x; {1}/{2},\kappa\right)=\frac{2\kappa\,e^{-x^2}}{2-\kappa\sqrt{\pi}\,\mathrm{erfc}(x)},
\quad x\in \mathbb{R},\end{equation}
where $\mathrm{erfc}(x)=\frac{2}{\sqrt{\pi}}\int^{\infty}_{x}e^{-t^2}dt$ is the complementary error function; cf. \cite[Chapter 7]{NIST}.
It is readily verified that $\kappa^*=1/\sqrt{\pi}$. The approximation \eqref{qAsy1} is valid since $\mathrm{erfc}(x)$ is strictly  monotone decreasing on $\mathbb{R}$ with range $(0, 2)$. Similarly, it holds that $q\left(x; {1}/{2},1/\sqrt\pi \right)\sim -2x$ as $x\to-\infty$. For $\kappa\sqrt \pi >1$,  the solution has a real pole which is the zero of the denominator in
\eqref{specialsolu}.
For general $\alpha+\frac{1}{2}\in\mathbb{N}$, the solutions and their asymptotic approximations as given in \eqref{qAsy1} can be obtained by using \eqref{specialsolu} and the B\"{a}cklund transformation; see \cite[Equation (5.9)]{BCHM}, or \cite[Equation (3.24)]{BCH}.


The solutions  $q(x; \alpha, \kappa)$ with positive half-integers $\alpha$
  find prominent applications in the theory of orthogonal polynomials with the discontinuous Hermite weight \cite{Chen}, and in random matrix theory \cite{FW,TW}.
It is well known that for the Gaussian unitary ensemble of $n\times n$  Hermitian matrices,
the probability of having no eigenvalues in the interval $(x,+\infty)$ can be expressed in terms of the Fredholm determinant
\begin{equation}
\det\left(\mathbf{I}-K_{n,x}\right).
\end{equation}
Here, $K_{n,x}$ is the integrable operator acting on $L^2(0,+\infty)$ with the classical Hermite kernel
 \begin{equation}\label{eq:OPkernel}
K_{n,x}(\lambda,\mu)=e^{-\frac{(x+\lambda)^2+(\mu+x)^2}{2}} \gamma_{n-1}^{2}\frac{\pi_n(\lambda+x)\pi_{n-1}(\mu+x)-\pi_{n-1}(\lambda+x)\pi_{n}(\mu+x)}{\lambda-\mu}, \end{equation}
where $\pi_n(\lambda)$ is the $n$-th monic Hermite polynomial
 determined by the orthogonal relation
  \begin{equation}\label{eq:OP} \int_{\mathbb{R}}\pi_{n}(x)\pi_{m}(x)e^{-x^2}dx=\gamma_{n}^{-2}\delta_{n,m},
  \end{equation}
with the normalization constant
  \begin{equation}\label{eq:gamman}\gamma_{n-1}^{2}=\frac{2^{n-1}}{\sqrt{\pi}\,\Gamma(n)};
  \end{equation}
  see \cite[Table 18.3.1]{NIST}.
 It is found by Tracy and Widom in \cite{TW}  that
\begin{equation}\label{eq:detH}
\frac{d}{dx}\ln \det(\mathbf{I}-K_{n,x})=\sigma_{n}(x)
\end{equation}
where $\sigma_{n}(x)$ is the unique solution of the $\sigma$-form of the PIV equation \cite{JM} with the parameter $\nu=n$
\begin{equation}\label{eq:sPIV}
\left(\sigma_{\nu}^{\prime \prime}\right)^{2}+4\left(\sigma_{\nu}^{\prime}\right)^{2}\left(\sigma_{\nu}^{\prime}+2 \nu\right)-4\left(x \sigma_{\nu}^{\prime}-\sigma_{\nu}\right)^{2}=0,
\end{equation}
characterized by the boundary condition 
\begin{equation}\label{eq:sPIVAsy}
\sigma_{\nu}(x)\sim \frac{2^{\nu-1} x^{2 \nu-2} e^{-x^{2}}}{\sqrt{\pi}(\nu-1) !},\quad x\to+\infty.
\end{equation}
It is worth mentioning  that   $\sigma_{\nu}(x)$  is closely related to the Hamiltonian for the PIV equation \eqref{PIV}; see Remark \ref{rem:Hq}  below and  \cite{JM}.

For general $\alpha-\frac{1}{2}\notin \mathbb{Z}$,
the asymptotic formula \eqref{qAsy2}  in part (a) was proved by Abdullayev \cite{Ab} using the integral equation method. The connection formulas, that is, the explicit expressions of the parameters $d_1$ and $d_2$ in \eqref{qAsy2}, in terms of  the parameter $\kappa$, were derived later by Its and Kapaev \cite{IK}, and by Wong and Zhang \cite{WZ}, using respectively   the isomonodromy method and the  uniform asymptotic  approach.
Moreover, according to the numerical investigations performed
in \cite{BCH1}, $q(x;\alpha,\kappa)$ might blow up at finite $x$ if $\alpha<-1/2$. While $\alpha>-1/2$, the same numerical results allow us to expect the absence of the real poles of $q(x;\alpha,\kappa)$.
The asymptotic behavior  of $q(x;\alpha,\kappa)$ and the  connection formulas in part (c) were recently derived by the current authors in \cite{XXZ}.
While, to the best of our knowledge, the asymptotic result in part (b) of the  Clarkson-McLeod conjecture  has not been  confirmed.
It is also desirable to know whether there exists a determinantal representation of the  $\sigma$-form of the Clarkson-McLeod solutions similar to \eqref{eq:detH}, for general $\alpha$ not being a half-integer.



In the present paper, we derive the asymptotic approximations and the connection formulas for the Clarkson-McLeod solutions as $x\to-\infty$ by using the Deift-Zhou nonlinear steepest descent method \cite{Deft,DZ,DZ1}. Particularly, we prove  part (b) and  revisit   part (a), (c)  of the  Clarkson-McLeod conjecture.
We also show that the $\sigma$-form of the Clarkson-McLeod solutions of the PIV equation \eqref{PIV} with general parameter $\alpha$ can be represented  by the Fredholm determinant of an integrable operator whose kernel is expressed in terms of the classical parabolic cylinder functions, thus generalizing the result \eqref{eq:detH} of Tracy and Widom.
Furthermore, the asymptotics of the Hamiltonian of the Clarkson-McLeod solutions and the evaluations of total integrals of the Clarkson-McLeod solutions are also obtained.

\subsection{Statement of results}
\subsubsection*{Asymptotics of the Clarkson-McLeod solutions}
Our first result is the following complete description of the asymptotic behaviors  of the Clarkson-McLeod solutions to the PIV equation \eqref{PIV} when the parameter $\alpha\in\mathbb{R}$ with $\alpha+\frac{1}{2}\notin\mathbb{N}$.
\begin{thm}\label{thm1}
Let $\kappa\neq 0$ be a given real number and $\kappa^*$ be the constant defined by \eqref{kappastar}. For any $\alpha\in\mathbb{R}$ with $\alpha+\frac{1}{2}\notin\mathbb{N}$, there exists a unique real solution $q(x;\alpha,\kappa)$ to the PIV equation \eqref{PIV} satisfying the following asymptotic behavior
\begin{equation}\label{qAsy}
q(x;\alpha,\kappa)=\kappa\, 2^{\alpha-\frac{1}{2}}x^{2\alpha-1}e^{-x^2}\left(1+O\left(x^{-2}\right)\right)\quad \mathrm{as}\quad x\to+\infty.
\end{equation}

For $\alpha\in\mathbb{R}$ with $\alpha-\frac{1}{2}\notin\mathbb{Z}$, the solution  $q(x;\alpha,\kappa)$ possesses the following asymptotic behaviors as $x\rightarrow-\infty$.
\begin{description}
\item{(1)} If $\kappa(\kappa-\kappa^*)<0$, then
\begin{equation}\label{q1}
q(x;\alpha,\kappa)=-\frac{2}{3}x+\frac{2\sqrt{6}\,b_1}{3}\sin \left(\frac{x^{2}}{\sqrt{3}}-\frac{b_1^{2}}{\sqrt{3}}\ln(2\sqrt{3} x^{2})+\psi_1\right)+O\left(\frac{1}{x}\right).
\end{equation}
\item{(2)} If $\kappa=\kappa^*$, then
\begin{equation}\label{q2}
q(x;\alpha,\kappa)=-2x+O\left(\frac{1}{x}\right).
\end{equation}
\item{(3)} If $\kappa(\kappa-\kappa^*)>0$, then
\begin{equation}\label{q3}
q(x;\alpha,\kappa)=-\frac{2}{3}x+\frac{2x}{2\cos\left(\frac{x^{2}}{\sqrt{3}}
-\frac{b_2}{\sqrt{3}}\ln \left(2 \sqrt{3} x^{2}\right)+\psi_2\right)+1}+O\left(\frac{1}{x}\right).
\end{equation}
\end{description}
The error term in the asymptotic expansion \eqref{q3} is uniform for $x$ bounded away from the singularities appearing on the right-hand side of the asymptotic expansion. Moreover, the corresponding connection formulas are respectively given by
\begin{equation}\label{connectionformula1}
\left\{
\begin{aligned}
b_1^{2}&=-\frac{\sqrt{3}}{2\pi} \ln(1-\left|\rho\right|^{2}),\quad b_1\geq0, \\
\psi_1&=-\frac{ \pi}{4}-\frac{2 \pi}{3}\alpha-\arg \Gamma\left(-\frac{b_1^{2}}{\sqrt{3}}i \right)-\arg \rho,
\end{aligned}\right.
\end{equation}
and
\begin{equation}\label{connectionformula2}
\left\{
\begin{aligned}
b_2&=-\frac{\sqrt{3}}{2\pi} \ln(|\rho|^{2}-1),\\
\psi_2&=-\frac{2\pi}{3}\alpha
-\arg\Gamma\left(-\frac{b_2}{\sqrt{3}}i+\frac{1}{2}\right)-\arg \rho,
\end{aligned}\right.
\end{equation}
where
\begin{equation}\label{rhoandkappa}
\rho=1-\frac{2\pi^{\frac{3}{2}}}{e^{\pi i\alpha}\Gamma\left(\frac{1}{2}-\alpha\right)}\kappa.
\end{equation}

When $\frac{1}{2}-\alpha\in \mathbb{N}$, it holds $\kappa^*=0$, and $q(x;\alpha,\kappa)$ possesses the   asymptotic behavior \eqref{q3}
for $\kappa\in \mathbb{R}\setminus\{0\}$
as $x\rightarrow-\infty$.
\end{thm}

%
%

\begin{rem}\label{rem-1.2} 
For all $\alpha\in\mathbb{R}$, the remaining case $\kappa=0$ corresponds to the trivial solution $q(x;\alpha,0)=0$. As aforementioned,
the asymptotic formula \eqref{q1} and connection formulas \eqref{connectionformula1} have been derived  in \cite{IK,WZ}. In this paper, we give alternative proofs of these formulas by performing asymptotic analysis of the Riemann-Hilbert (RH, for short) problem for PIV equation. Furthermore, we accomplish the case $b_1=0$ in \eqref{q1} and \eqref{connectionformula1} which is not covered in \cite[Equation (1.10)]{IK} and \cite[Equation (1.5)]{WZ}. We also show that \eqref{q3} is true for negative half-integer $\alpha$, a case   not considered in our previous work \cite{XXZ}.
We also provide a novel   proof of \eqref{q2}. Theorem \ref{thm1}, along with  
the case with positive half-integer $\alpha$ solved in terms of the special functions in \cite{BCHM,CM}, fully confirms  the conjecture of Clarkson and McLeod. Minor extension allows  $\kappa^*$  to be negative or zero.
\end{rem}\vskip .3cm

The asymptotic analysis of the  RH problem for the PIV equation enables us to derive simultaneously the asymptotics of the  Hamiltonian associated with the Clarkson-McLeod solutions and the total integrals of Clarkson-McLeod solutions along the real axis, which are of independent interests.

\begin{thm}\label{thm2}
Under the conditions in Theorem \ref{thm1}, for any $\alpha\in\mathbb{R}$ with $\alpha+\frac{1}{2}\notin\mathbb{N}$, the Hamiltonian $\mathcal{H}(x;\alpha,\kappa)$ of the Clarkson-McLeod solutions, defined by \eqref{eq:Hq} below,
 has the following asymptotic behavior as $x\rightarrow +\infty$:
\begin{equation}\label{Hasymp+infty}
\mathcal{H}(x;\alpha,\kappa)=-\kappa\,2^{\alpha-\frac{1}{2}}
x^{2\alpha-1}e^{-x^2}\left(1+O\left(x^{-2}\right)\right).
\end{equation}

While for $\alpha\in\mathbb{R}$ with $\alpha-\frac{1}{2}\notin\mathbb{Z}$,  $\mathcal{H}(x;\alpha,\kappa)$ satisfies the following asymptotic behaviors as $x\rightarrow-\infty$.
\begin{description}
\item{(1)} If $\kappa(\kappa-\kappa^*)<0$, then
\begin{equation}\label{H1}
\mathcal{H}(x;\alpha,\kappa)=-\frac{8x^3}{27}+\frac{4}{3}
\left(\alpha+b_1^2\right)x
-\frac{2\sqrt{2}\,b_1}{3}\cos
\left(\frac{x^{2}}{\sqrt{3}}-\frac{b_1^{2}}{\sqrt{3}}\ln (2 \sqrt{3}x^{2} )+\psi_1\right)+O\left(\frac{1}{x}\right),
\end{equation}
\item{(2)} If $\kappa=\kappa^*$, then
\begin{equation}\label{H2}
\mathcal{H}(x;\alpha,\kappa)=4\alpha x+O\left(\frac{1}{x}\right),
\end{equation}
\item{(3)} If $\kappa(\kappa-\kappa^*)>0$, then
\begin{equation}\label{H3}
\mathcal{H}(x;\alpha,\kappa)=-\frac{8x^3}{27}+\frac{4}{3}(\alpha+b_2)x
+\frac{\frac{4}{\sqrt{3}}x\sin\left(\frac{x^{2}}{\sqrt{3}}-\frac{b_2}{\sqrt{3}}\ln \left(2 \sqrt{3} x^{2}\right)+\psi_2\right)}
{2\cos\left(\frac{x^{2}}{\sqrt{3}}-\frac{b_2}{\sqrt{3}}\ln \left(2 \sqrt{3} x^{2}\right)+\psi_2\right)+1}+O\left(\frac{1}{x}\right).
\end{equation}
\end{description}
The parameters $b_1$, $\psi_1$ and $b_2$, $\psi_2$ are the same as in \eqref{connectionformula1} and \eqref{connectionformula2}, respectively. The error term in the asymptotic expansion \eqref{H3} is uniform for $x$ bounded away from the singularities appearing on the right-hand side of the asymptotic expansion.

Moreover, when $\frac{1}{2}-\alpha\in \mathbb{N}$, for any non-vanishing
$\kappa\in \mathbb{R}$, $\mathcal{H}(x;\alpha,\kappa)$ possesses the same asymptotic behavior \eqref{H3} as $x\rightarrow-\infty$.
\end{thm}


Our next result is the evaluation of the total integrals of the Clarkson-McLeod solutions $q(x;\alpha,\kappa)$. Similar results for the integrals of the Painlev\'e II transcendents have been derived in \cite{BBD,BBDI,DXZ,Kok,Mil}.


\begin{thm}\label{thm3}
Let $\mathrm{P.V.}$ denote the Cauchy principal value. For $\alpha\in\mathbb{R}$ with $\alpha-\frac{1}{2}\notin\mathbb{Z}$, we evaluate the total integrals of $q(x;\alpha,\kappa)$ as follows.
\begin{description}
\item{(1)} If $\kappa(\kappa-\kappa^{*})<0$, take $c < 0 < d$ such that all real poles of $q(x;\alpha,\kappa)$ lie in the interval $(c,d)$,
\begin{align}\label{integral-1}
&\exp\Bigg\{\int^{c}_{-\infty}\left(q(t;\alpha,\kappa)+\frac{2t}{3}
-\frac{2\alpha}{t}\right)dt+\mathrm{P.V.}\int_{c}^{d}q(t;\alpha,\kappa)\,dt
+\int^{+\infty}_{d}q(t;\alpha,\kappa)\,dt\Bigg\}\nonumber\\
&\qquad =\frac{(-1)^{N_{+}-N_{-}}\sqrt{\pi}e^{\frac{c^2}{3}}|c|^{-2\alpha}3^{2\alpha}(1-\rho)e^{\pi i\alpha}}{2^{\frac{1}{2}-\alpha}\Gamma(\frac{1}{2}-\alpha)(1-|\rho|^2)^{\frac{2}{3}}},
\end{align}
where $\rho$ is related to $\kappa$ by \eqref{rhoandkappa} and $N_{\pm}$ denote  the numbers of real poles
of $q(x;\alpha,\kappa)$ in the interval $(c,d)$ of residues $\pm1$, respectively.
\item{(2)} If $\kappa=\kappa^{*}$, take $c< 0 < d$ such that all real poles of $q(x;\alpha,\kappa)$ lie in the interval $(c,d)$,
\begin{align}\label{integral-2}
&\exp\Bigg\{\int^{c}_{-\infty}\left(q(t;\alpha,\kappa)+2t
+\frac{2\alpha}{t}\right)dt
+\mathrm{P.V.}\int_{c}^{d}q(t;\alpha,\kappa)\,dt
+\int^{+\infty}_{d}q(t;\alpha,\kappa)\,dt\Bigg\}\nonumber\\
&\qquad =\frac{(-1)^{N_{+}-N_{-}} \sqrt{\pi}e^{-c^2}
|c|^{-2\alpha}2^{\frac{1}{2}+\alpha}}{\Gamma(\frac{1}{2}+\alpha)},
\end{align}
where $N_{\pm}$ denote  the numbers of real poles
of $q(x;\alpha,\kappa)$ of residues $\pm1$ in the interval $(c,d)$.

\end{description}
\end{thm}

\begin{rem}
According to the numerical analysis performed in \cite{BCH1}, it is expected that $q(x;\alpha,\kappa)$ is pole free on the real axis when
$0<\kappa\leq\kappa^*$. However, to the best of our knowledge, there is no rigorous proof of the numerical evidence.
The Cauchy principal values in \eqref{integral-1} and \eqref{integral-2} may be removed if one can prove that $q(x;\alpha,\kappa)$
is pole free on the real axis in  these cases.
\end{rem}

\subsubsection*{Determinantal representation of the Clarkson-McLeod solutions}
Let $\gamma K_{\nu,x}$ be the integrable operator acting on $L^2(0,+\infty)$ with the parabolic cylinder kernel
 \begin{equation}\label{eq:PCKernel}
\gamma  K_{\nu, x}(\lambda,\mu)=\gamma \frac{D_{\nu}(\sqrt{2}(\lambda+x))D_{\nu-1}(\sqrt{2}(\mu+x))
-D_{\nu-1}(\sqrt{2}(\lambda+x))D_{\nu}(\sqrt{2}(\mu+x))}{\lambda-\mu},
\end{equation}
where $D_{\nu}$ is the parabolic cylinder function with the parameter $\nu \in \mathbb{R}$ and $\gamma$ is a real parameter. Then, the Fredholm determinant of the operator is related to the $\sigma$-form of the Clarkson-McLeod solutions as stated in the following theorem.

 \begin{thm}\label{thm:IntRep}
Suppose that $\nu\in\mathbb{R}$, $\gamma\in \mathbb{R}$ and $\gamma K_{\nu,x}$ is the integrable operator acting on $L^2(0,+\infty)$ with the kernel \eqref{eq:PCKernel}. We have
\begin{equation}\label{eq:IntRep}
\frac{d}{dx} \ln \det(\mathbf{I}-\gamma K_{\nu, x})=\sigma_{\nu}(x;\gamma),\end{equation}
where $\sigma_{\nu}(x;\gamma)$ is the unique solution of  the $\sigma$-form \eqref{eq:sPIV} of the PIV equation,
determined by the boundary condition near positive infinity
\begin{equation}\label{eq:SigmaAsy}
\sigma_{\nu}(x;\gamma)\sim \sqrt{2}\gamma D_{\nu-1}^2(\sqrt{2}x).
\end{equation}
\end{thm}\vskip .5cm

\begin{rem}\label{rem:Hq}
It is seen from \cite[Equations (C.34)-(C.36)]{JM} and \eqref{eq:Hq} below that the  $\sigma$-form of the PIV equation is related to $q(x;\alpha,\kappa)$ and the Hamiltonian $\mathcal{H}(x;\alpha,\kappa)$  by
\begin{equation}\label{eq:sigmaH}
\sigma_{\alpha+\frac{1}{2}}(x;\gamma)=\frac{1}{2}\left(q(x;\alpha,\kappa)-\mathcal{H}(x;\alpha,\kappa)\right),
\end{equation}
and conversely
\begin{equation}\label{eq:qsigma}q(x;\alpha,\kappa)= -\frac{\sigma''_{\alpha+\frac{1}{2}}(x;\gamma)+2x\sigma_{\alpha+\frac{1}{2}}'(x;\gamma)-2\sigma_{\alpha+\frac{1}{2}}(x;\gamma)}{2\sigma_{\alpha+\frac{1}{2}}'(x;\gamma)+4\alpha+2}.\end{equation}
It follows from the approximation \eqref{eq:SigmaAsy} and the relation \eqref{eq:qsigma} that $q(x;\alpha,\kappa)$ actually satisfies the boundary condition \eqref{qAsy},
and the parameters $\kappa$ and $\gamma$ are related by
\begin{equation}\label{eq:sigmaKappa}
\kappa=\sqrt{2}\gamma.
\end{equation}
In view of  \eqref{eq:sigmaH},  the asymptotic behavior  near negative infinity  of the solution of the   $\sigma$-form \eqref{eq:sPIV} of the PIV equation subject to  the boundary condition \eqref{eq:SigmaAsy} can be obtained by using  Theorems \ref{thm1} and \ref{thm2}.
\end{rem}

\begin{rem}
For $\nu=n\in\mathbb{N}$, the parabolic cylinder function is reduced  to the Hermite polynomial  \cite[(12.7.2)]{NIST}
  \begin{equation}\label{eq:Hermite}
D_{n}(\sqrt{2}x)=e^{-\frac{1}{2}x^2} 2^{\frac{n}{2}} \pi_{n}(x),
\end{equation}
where $\pi_{n}(x)$ is the monic  Hermite  polynomial defined through
\eqref{eq:OP}. Therefore, for given $\nu\in\mathbb{N}$, the kernel $\gamma^{*} K_{\nu, x}(\lambda,\mu)$ with
    \begin{equation}\label{eq:gammaSta} \gamma^{*}=\frac{1}{\sqrt{2\pi}\Gamma(\nu)}\end{equation}
     is reduced to  the classical Hermite kernel defined in \eqref{eq:OPkernel}.
  Applying Theorem \ref{thm:IntRep}, we recover  \eqref{eq:detH}-\eqref{eq:sPIVAsy}, which was obtained first by Tracy and Widom.
   \end{rem}


%
The rest of the present paper is arranged as follows. In Section \ref{sec:RHPforPIV}, we state the RH problem for the PIV equation \eqref{PIV} and express the solutions of the PIV equation and the associated Hamiltonian  in terms of  the solution to this RH problem. Subsequently, in Sections \ref{Asymptotic-infty1}-\ref{Asymptotic-infty3}, we apply the Deift-Zhou nonlinear steepest descent method to the mentioned RH problem as $x\to-\infty$,
 for the parameter $\kappa$ respectively in three different regimes.  Using the asymptotic analysis of the RH problems we performed, Theorems \ref{thm1}-\ref{thm3}  are then proved in Section \ref{proof1}. The final Section \ref{sec:proof of det} is devoted to the proof of Theorem \ref{thm:IntRep}. For the convenience of the reader, we collect in the Appendix three local parametrix models used in the asymptotic analysis of the RH problems.

\section{Riemann-Hilbert problem for the  Painlev\'e IV  equation}\label{sec:RHPforPIV}

In this section, we review the RH problem for the PIV equation \eqref{PIV},  a detailed description can be found in  \cite[Section 2]{IK} and \cite[Chapter 5.1]{FIKN}.

Denote  $\Sigma=\bigcup^{8}_{k=1}\gamma_{k}$, where $\gamma_{k}=\{\xi\in \mathbb{C}\mid\arg\xi=k\pi/4\}$; see Figure \ref{PIVj}. Then, the $2\times2$ matrix-valued function $\Psi(\xi,x)$ solves the following RH problem.
\subsection*{RH problem for $\Psi(\xi,x)$}
\begin{description}
\item{(1)} $\Psi(\xi,x)$ is analytic for all $\xi\in\mathbb{C}\setminus \Sigma$.

\item{(2)} $\Psi(\xi,x)$ satisfies the jump relations
$$
\Psi_{+}(\xi,x)=\Psi_{-}(\xi,x)\left\{
\begin{aligned}
&S_{k},\quad &\xi&\in\gamma_{k},\ k=1,\cdots,7,\\
&S_{8}e^{-2\pi i\alpha\sigma_{3}},\quad &\xi&\in\gamma_{8},
\end{aligned}
\right.
$$
where the Stokes matrices
\begin{equation}\label{eq:Sk}
S_{2i-1}=
\begin{pmatrix}
1 & s_{2i-1}\\
0 & 1
\end{pmatrix},\quad
S_{2i}=
\begin{pmatrix}
1 & 0\\
s_{2i} & 1
\end{pmatrix}, \quad \ i=1,2,3,4.
\end{equation}
The Stokes multipliers $s_k$ satisfy the following restrictions
\begin{equation}\label{srela1}
s_{k+4}=-s_{k}e^{(-1)^{k}2\pi i\alpha},\quad k=0,1,2,3,4,
\end{equation}
and
\begin{equation}\label{srela2}
[(1+s_{3}s_{4})(1+s_{1}s_{2})+s_{1}s_{4}]e^{-i\pi\alpha}-(1+s_{2}s_{3})e^{i\pi \alpha}=-2i\sin(\pi\alpha).
\end{equation}

\item{(3)} As $\xi\to\infty$, $\Psi(\xi,x)$ satisfies the following asymptotic condition
\begin{equation}\label{Asyatinfty}
\Psi(\xi,x)=\Psi^{(\infty)}(\xi,x)
e^{\Theta(\xi,x)\sigma_{3}},
\end{equation}
where
\begin{equation}\label{Psiinfty}
\Psi^{(\infty)}(\xi,x)=\mathbf{I}+\frac{\Psi_{1}}{\xi}+\frac{\Psi_{2}}{\xi^2}
+O\left(\frac{1}{\xi^{3}}\right),\quad \Theta(\xi,x)=\frac{1}{8}\xi^{4}+\frac{x}{2}\xi^{2}+\alpha\ln\xi,
\end{equation}
with the branch of $\ln\xi$  chosen so that $\arg\xi\in(0,2\pi)$.


\item{(4)} As $\xi\to0$, $\Psi(\xi,x)$ has the asymptotic behavior of the form
\begin{equation}\label{Asyatzero}
\Psi(\xi,x)=\left\{\begin{aligned}
&\Psi^{(0)}(\xi,x)\xi^{\alpha\sigma_3}E, && \mathrm{for}\quad \alpha-\frac{1}{2}\notin \mathbb{Z},\\
&\Psi^{(0)}(\xi,x)\xi^{\alpha\sigma_3}\begin{pmatrix} 1 & \frac{s_0}{\pi i}\ln\xi \\ 0 & 1\end{pmatrix}E,&& \mathrm{for}\quad \alpha+\frac{1}{2}\in \mathbb{N},\\
&\Psi^{(0)}(\xi,x)\xi^{\alpha\sigma_3}\begin{pmatrix} 1 & 0 \\ -\frac{s_0}{\pi i}\ln\xi & 1\end{pmatrix}E,&& \mathrm{for}\quad \frac{1}{2}-\alpha\in \mathbb{N},
\end{aligned}\right.
\end{equation}
where $\Psi^{(0)}(\xi,x)$ is  analytic in the neighborhood of $\xi=0$.
The functions $\xi^{\alpha}$ and $\ln\xi$ take principal values. The connection matrix $E$ is given by
\begin{equation}\label{connectionmatrix}
E=E_0S_{0 }\cdots S_{k-1},\quad\xi\in\Omega_k,\quad k=1,\cdots,8,
\end{equation}
where $S_0:=\mathbf{I}$ and the regions $\Omega_k=\left\{\xi\in\mathbb{C}\mid\arg\xi\in(\frac{(k-1)\pi}{4},\frac{k\pi}{4})\right\}$ are depicted in Figure \ref{PIVj}.
Moreover, the connection matrix $E_0$ takes the form
\begin{equation}\label{E0}
E_0=\left\{\begin{aligned}
&\begin{pmatrix} 1 & 0 \\ \frac{s_0e^{2\pi i\alpha}}{e^{2\pi i\alpha}+1} & 1
\end{pmatrix}, && \mathrm{for}\quad \alpha-\frac{1}{2}\notin\mathbb{Z},\\
&\begin{pmatrix} p_1 & -1 \\ 1 & 0
\end{pmatrix}, && \mathrm{for}\quad \alpha+\frac{1}{2}\in\mathbb{N},\\
&\begin{pmatrix} 1 & 0 \\ p_2 & 1
\end{pmatrix}, && \mathrm{for}\quad \frac{1}{2}-\alpha\in\mathbb{N},
\end{aligned}\right.
\end{equation}
where $p_1$, $p_2$ are two arbitrary nonzero constants.
\end{description}

In the RH formulation, we  denote  the Pauli matrices by $\sigma_k$, $k=1,2,3$
\begin{equation}\label{Pauli}
\sigma_1=\begin{pmatrix}0 & 1\\ 1 & 0\end{pmatrix},\quad \sigma_2=\begin{pmatrix}0 & -i\\ i & 0\end{pmatrix},\quad \sigma_3=\begin{pmatrix}1 & 0\\ 0 & -1\end{pmatrix}.
\end{equation}
\begin{figure}[t]
  \centering
  \includegraphics[width=7cm,height=7cm]{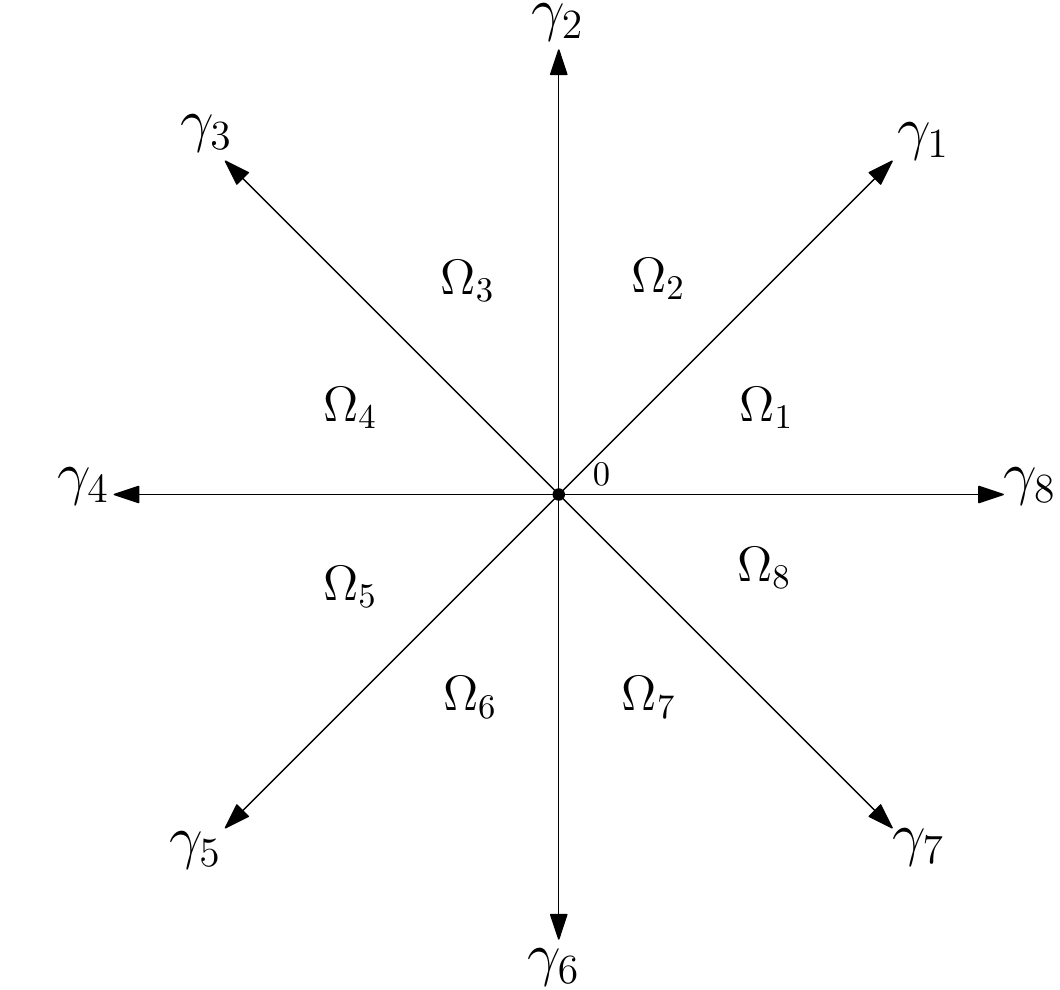}\\
  \caption{The jump contour $\Sigma$ and the regions $\Omega_k$}\label{PIVj}
\end{figure}

A significant  fact  is that
the  PIV transcendents  and the associated Hamiltonian can be expressed in terms of  the solution of the RH problem for $\Psi(\xi,x)$.

\begin{pro}\label{eq:PIVRHP}
The solution  of the PIV equation \eqref{PIV} and the associated Hamiltonian  are  related to  the solution to above RH problem for $\Psi(\xi,x)$ by
\begin{equation}\label{qsolu1}
q(x)=(\Psi_{1})_{12}(\Psi_{1})_{21},
\end{equation}
and
\begin{equation}\label{tau}
\mathcal{H}(x)
=\left(\Psi_2\right)_{11}-\left(\Psi_2\right)_{22},
\end{equation}
where $\Psi_1$ and $\Psi_2$ are given in  \eqref{Psiinfty}.
Let  $\Psi^{(0)}(\xi,x)$  be defined as \eqref{Asyatzero}, we have
\begin{equation}\label{Fequation}
F(x):=\Psi^{(0)}(0,x)=\exp\bigg[\int^x q(x)dx\; \sigma_3\bigg].
\end{equation}
\end{pro}
\begin{proof}
The relations  \eqref{qsolu1} and \eqref{Fequation} have been   derived in   \cite[Equations (2.11) and (2.40)]{IK}.
According to Jimbo, Miwa and Ueno \cite{JMU}, it follows from \eqref{Asyatinfty} and \eqref{Psiinfty} that
\begin{equation}\label{tau-Res}
\mathcal{H}(x)=-\mathop{\mathrm{Res}}\limits_{z=\infty}
\mathrm{Tr}\left(\Psi^{(\infty)}(\xi,x)^{-1}
\frac{\mathrm{d}\Psi^{(\infty)}(\xi,x)}{\mathrm{d}\xi}
\frac{\mathrm{d}\Theta(\xi,x)}{\mathrm{d}x}\sigma_3\right)
=\left(\Psi_2\right)_{11}-\left(\Psi_2\right)_{22}.
\end{equation}
\end{proof}

We mention that using the Lax pair  for $\Psi(\xi,x)$  \cite[Equations (2.1)-(2.2)]{IK}, the Hamiltonian $\mathcal{H}(x)$ can be expressed in terms of $q$
as follows
\begin{equation}\label{eq:Hq}
\mathcal{H}(x)=\frac{q^3}{4}+xq^2+(x^2-2\alpha)q-\frac{(q')^2}{4q}.
\end{equation}
Let $p=\frac{q'}{q}$, then
\begin{equation}\label{eq:H}
H(x):=-2\mathcal{H}(x)=-\frac{q^3}{2}-2xq^2-2(x^2-2\alpha)q+\frac{1}{2}p^2q.
\end{equation}
The PIV equation \eqref{PIV} can be derived by eliminating $p$ from the  Hamilton equations
\begin{equation}\label{eq:Hequ}
\frac{dq}{dx}=\frac{\partial H}{\partial p},\quad \frac{dp}{dx}=-\frac{\partial H}{\partial q}.
\end{equation}
Applying \eqref{eq:Hq}-\eqref{eq:Hequ}, we also find the relation
\begin{equation}\label{eq:dH}
\frac{d}{dx}\mathcal{H}(x)=\frac{\partial }{\partial x}\mathcal{H}(x)=q^2+2xq.
\end{equation}

It was observed by Its and Kapaev \cite[Equation (2.23)]{IK} that, for any real solution of the PIV equation \eqref{PIV}, the Stokes multipliers  satisfy  the conditions
\begin{equation}\label{assum1}
\overline{s}_0=s_0,\quad \overline{s}_1=-e^{2\pi i\alpha}s_{3}.
\end{equation}
When $\alpha\in\mathbb{R}$ with $\alpha-\frac{1}{2}\notin\mathbb{Z}$, Its and Kapaev proved that for solutions of the PIV equation \eqref{PIV} determined by the asymptotic behavior \eqref{qAsy} with real parameter $\kappa$, the Stokes multipliers fulfill the following conditions (see \cite[Theorem 3.1]{IK})
\begin{equation}\label{assum2}
s_2=0,\quad s_1+s_3=0,\quad s_*\neq1, \quad (1-s_*)e^{\pi i\alpha}\in\mathbb{R},
\end{equation}
where $s_*$ is constituted by the Stokes multipliers $s_0$, $s_1$  through
\begin{equation}\label{sstar}
s_*=s_0s_1+1.
\end{equation}
The connection formula between
the Stokes multipliers
and the parameter $\kappa$ in the asymptotic behavior \eqref{qAsy}
was also derived therein
\begin{equation}\label{kapparep}
\kappa=-\frac{(s_*-1)e^{\pi i\alpha}}{2\pi^{\frac{3}{2}}}\Gamma
\left(\frac{1}{2}-\alpha\right).
\end{equation}
Actually, the relation \eqref{kapparep} is also true for  $\frac{1}{2}-\alpha\in \mathbb{N}$.

For the case $\alpha-\frac{1}{2}\notin\mathbb{Z}$, combining  \eqref{kappastar} with \eqref{kapparep}, we find the following correspondence between the conditions on the  parameter $\kappa$ and the Stokes multiplier $s_*$ as shown in Table
\ref{table}.
\begin{table}[H]
\renewcommand\arraystretch{1}
\centering
\begin{tabular}{|c|c|}
\hline
   $\kappa$ &  $s_*$ \\
\hline
  $\kappa(\kappa-\kappa^*)<0$  &   $|s_*|<1$    \\
\hline
  $\kappa=\kappa^*$  &  $|s_*|=1,\ s_*\neq1$  \\
\hline
$\kappa(\kappa-\kappa^*)>0$ &  $|s_*|>1$  \\
\hline
\end{tabular}
\caption{The correspondence between $\kappa$ and $s_*$}\label{table}
\end{table}
We point out that the exceptional case $s_*=1$ in the table is equivalent to $\kappa=0$ and the solution determined by \eqref{qAsy} is trivial; see Remark \ref{rem-1.2}.  While, for the case $\frac{1}{2}-\alpha\in\mathbb{N}$, it follows from  \eqref{assum1}, \eqref{sstar} and \eqref{kapparep} that we always have
$$|s_*|^2=1+\frac{4\pi^3\kappa^2}{\Gamma(\frac{1}{2}-\alpha)^2}>1.$$
For $\alpha+\frac{1}{2}\in\mathbb{N}$, as mentioned in the introduction,  the solution of the PIV equation \eqref{PIV} determined by  \eqref{qAsy} can be expressed in terms of the complementary error function.
These special solutions correspond to the specified values of Stokes multipliers (see \cite{Kap98})
\begin{equation}\label{SpecialsoluStokes}
s_0=s_2=0,\quad s_1+s_3=0.
\end{equation}

\section{RH analysis as $x\to-\infty$ with $0\leq |s_*|<1$}\label{Asymptotic-infty1}
In this section, we start by carrying out the Deift-Zhou nonlinear steepest descent analysis of the RH problem for $\Psi$ as $x\to-\infty$ under the condition $0< |s_*|<1$. Then, the reduced case $|s_*|=0$ is considered at the end of this section.

Assume now that $x<0$. We start with the first transformation
\begin{equation}\label{rescaling}
\Phi(z)=(-x)^{-\frac{\alpha}{2}\sigma_3}
\Psi\left((-x)^{\frac{1}{2}}z,x\right).
\end{equation}
Immediately, $\Phi(z)$ solves the following RH problem.

\subsection*{RH problem for $\Phi(z)$}
\begin{description}
\item{(1)} $\Phi(z)$ is analytic for all $z\in\mathbb{C}\setminus\widetilde\Sigma$, where $\widetilde\Sigma=\Sigma\setminus(\gamma_2\cup\gamma_6)$; see Figure \ref{PIVj}.

\item{(2)} $\Phi(z)$ satisfies the same jump conditions as $\Psi(\xi,x)$ on $\widetilde\Sigma$.
\item{(3)} As $z\to\infty$,
\begin{equation}\label{Asyatinfty1}
\Phi(z)=\left(\mathbf{I}+\frac{\Phi_{1}}{z}+\frac{\Phi_{2}}{z^2}+O(z^{-3})\right)
z^{\alpha\sigma_3}e^{x^2(\frac{1}{8}z^{4}-\frac{1}{2}z^{2})\sigma_{3}},
\end{equation}
where the branch of $z^{\alpha}$ is chosen such that $\arg z\in(0,2\pi)$.
\item{(4)} $\Phi(z)$ has the same asymptotic behaviors as $\Psi(\xi,x)$ as $z\to0$; see \eqref{Asyatzero} and \eqref{connectionmatrix}.
\end{description}

Simultaneously, by \eqref{qsolu1}, \eqref{tau} and \eqref{rescaling}, we get
\begin{equation}\label{solu2}
q(x;\alpha,\kappa)=-x(\Phi_{1})_{12}(\Phi_{1})_{21},
\end{equation}
and
\begin{equation}\label{tau1}
\mathcal{H}(x;\alpha,\kappa)=-x\left[(\Phi_2)_{11}-(\Phi_2)_{22}\right],
\end{equation}
where $\Phi_1=\Phi_1(x)$ and $\Phi_2=\Phi_2(x)$ are the coefficients in \eqref{Asyatinfty1}.

Introduce the following $g$-function
\begin{equation}\label{g-function}
g(z)=\frac{1}{8}z\left(z^{2}-\frac{8}{3}\right)^{\frac{3}{2}},
\end{equation}
where $\arg\left(z\pm\textstyle\sqrt{ {8}/{3}}\right)\in(-\pi,\pi)$. By straightforward computation,
\begin{equation}\label{gatinfty}
g(z)=\frac{1}{8}z^4-\frac{1}{2}z^2+\frac{1}{3}+\frac{4}{27}z^{-2}+O(z^{-4})
\end{equation}
as $z\rightarrow\infty$. Moreover, $g(z)$ has four saddle points
\begin{equation*}
z_{1,\pm}=\pm\sqrt{\frac{2}{3}},\quad z_{2,\pm}=\pm\sqrt{\frac{8}{3}}.
\end{equation*}

We then make the second transformation
\begin{equation}\label{U(z)}
\mathbf{U}(z)=e^{\frac{x^2}{3}\sigma_{3}}\Phi(z)e^{-x^2g(z)\sigma_{3}}.
\end{equation}
It is direct to check that $\mathbf{U}(z)$ satisfies the following RH problem.

\subsection*{RH problem for $\mathbf{U}(z)$}
\begin{description}
\item{(1)} $\mathbf{U}(z)$ is analytic for $z\in\mathbb{C}\setminus\widetilde\Sigma$.

\item{(2)} We have the jump relations $\mathbf{U}_{+}(z)=\mathbf{U}_{-}(z)J_\mathbf{U}(z)$, where
$$J_\mathbf{U}(z)=\left\{
\begin{aligned}
&\begin{pmatrix}1 & s_ke^{2x^2g(z)}\\ 0 & 1 \end{pmatrix},\quad &z&\in\gamma_{k},\ k=1,3,5,7,\\
&\begin{pmatrix}e^{x^2(g_-(z)-g_+(z))} & 0\\ -s_0e^{2\pi i\alpha}e^{-x^2(g_+(z)+g_-(z))} & e^{x^2(g_+(z)-g_-(z))} \end{pmatrix},\quad &z&\in\gamma_{4},\\
&\begin{pmatrix}e^{-2\pi i\alpha}e^{x^2(g_-(z)-g_+(z))} & 0\\ s_0e^{2\pi i\alpha}e^{-x^2(g_+(z)+g_-(z))} & e^{2\pi i\alpha}e^{x^2(g_+(z)-g_-(z))} \end{pmatrix},\quad &z&\in\gamma_{8}.
\end{aligned}
\right.
$$
\item{(3)} $\mathbf{U}(z)$ satisfies the asymptotic condition
$$\mathbf{U}(z)=\left(\mathbf{I}+O(z^{-1})\right)z^{\alpha\sigma_3},\quad \mathrm{as}\quad z\rightarrow\infty,$$
where $\arg z\in(0,2\pi)$.
\item{(4)} $\mathbf{U}(z)$ satisfies the following asymptotic behavior as $z\to0$
\begin{equation}\label{Uatzero}
\mathbf{U}(z)=\mathbf{U}_{0}(z)z^{\alpha\sigma_3}E_0S_1e^{-x^2g(z)\sigma_3},\qquad \arg z\in(\pi/4, {\pi}/{2}),
\end{equation}
where $\mathbf{U}_{0}(z)$ is analytic in the neighborhood of $z=0$. The asymptotic behaviors of $\mathbf{U}(z)$ in other regions are determined by \eqref{Uatzero} and the jump relations satisfied by $\mathbf{U}(z)$.
\end{description}

\begin{figure}[t]
  \centering
  \includegraphics[width=13cm,height=8cm]{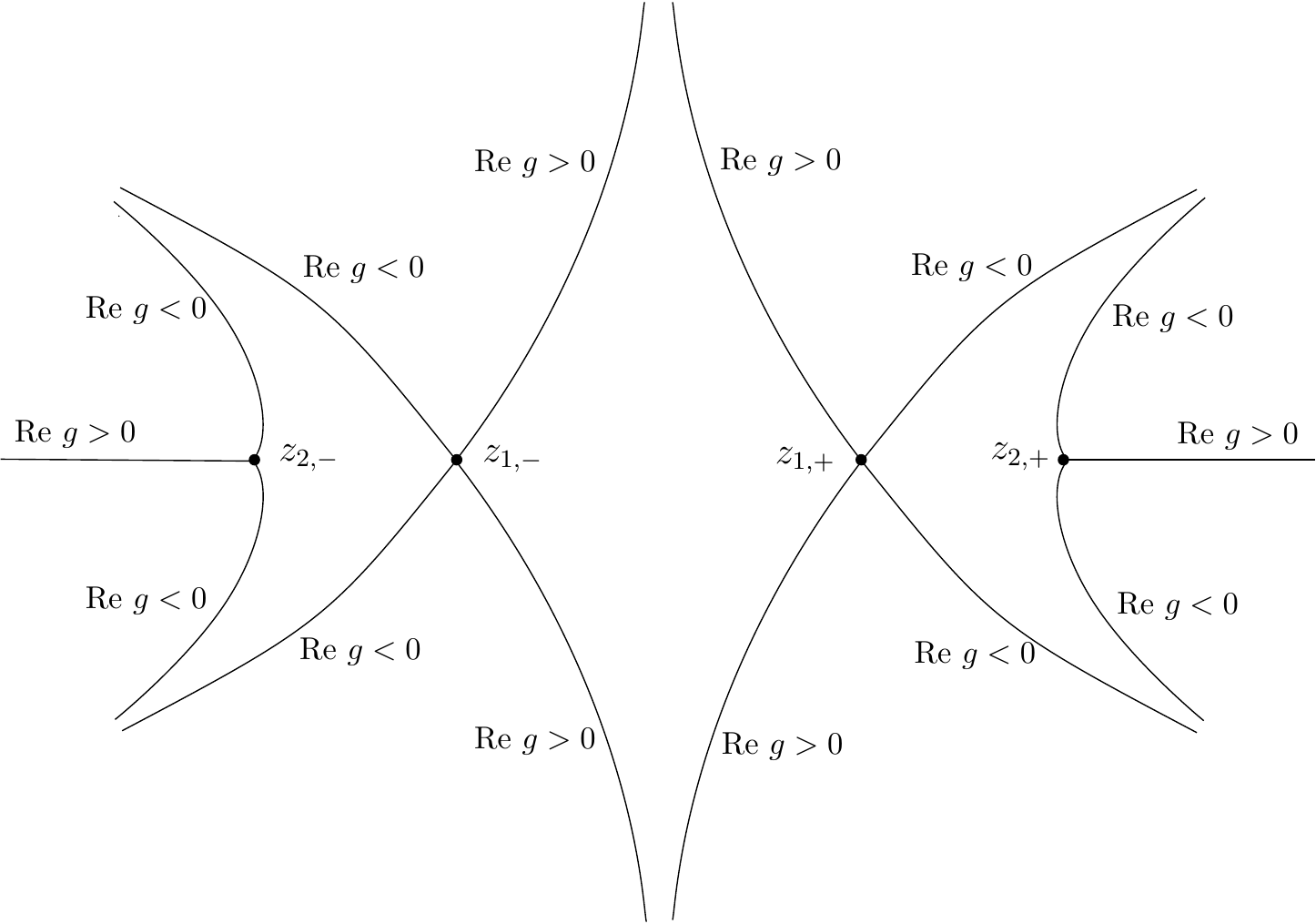}\\
  \caption{The anti-Stokes curves of the exponent $g(z)$}\label{ASC}
\end{figure}

\begin{figure}[H]
  \centering
  \includegraphics[width=13cm,height=7cm]{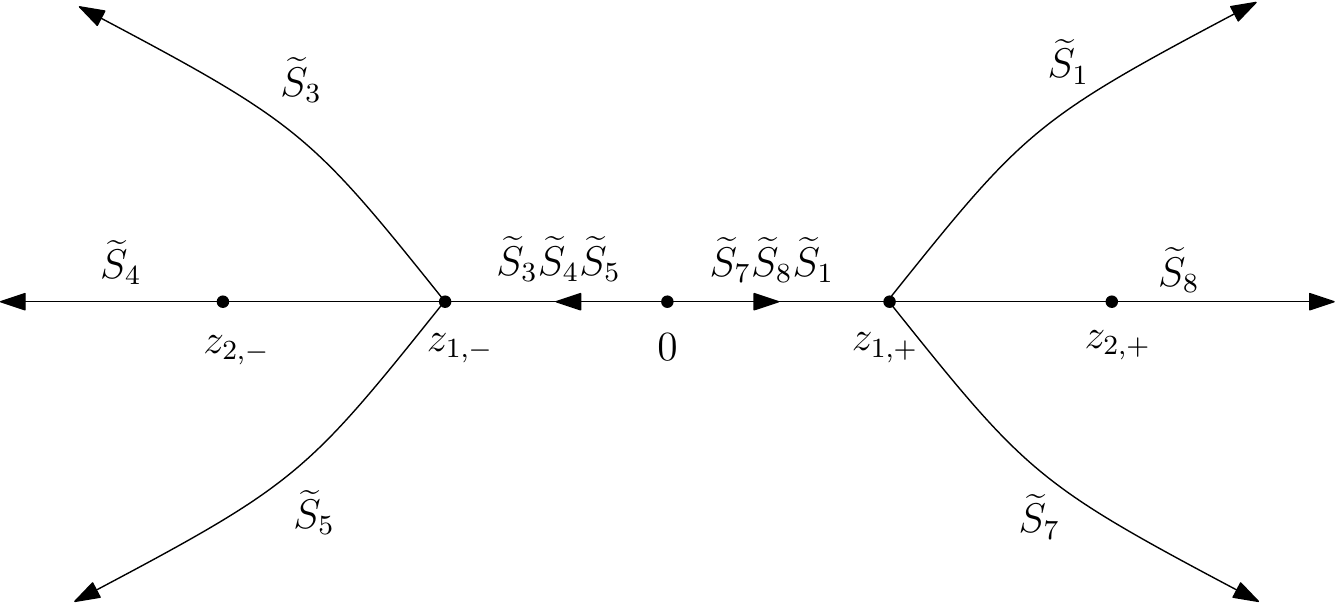}\\
  \caption{The first deformation of the jump contour} \label{Deforma1}
\end{figure}

\subsection{Deformations of the jump contour}\label{sec:RHT}

Next, we  transform the RH problem for $\mathbf{U}(z)$ to a RH problem formulated on the anti-Stokes curves of $g(z)$, as illustrated in Figure \ref{ASC}.

To this aim,  we first notice that the RH problem for $\mathbf{U}(z)$ is equivalent to the one posed on the contour shown in Figure \ref{Deforma1}, where we use the notations $\widetilde{S}_k$ to denote the
analytically extended  jump matrices $J_{\mathbf{U}}(z)$.

Secondly, since the jump matrices on $(z_{2,-},z_{2,+})$ are now oscillating for large $|x|$, we should deform the segment $(z_{2,-},z_{2,+})$ to the anti-Stokes curves of $g(z)$. Thus, we introduce the third transformation $\mathbf{U}\rightarrow \mathbf{T}$. This transformation is based on the following factorizations of the jump matrices on $(z_{2,-},z_{2,+})$.
\begin{align}\nonumber
\widetilde{S}_{4}&=\begin{pmatrix}e^{x^2(g_-(z)-g_+(z))} & 0\\ -s_0e^{2\pi i\alpha} & e^{x^2(g_+(z)-g_-(z))}\end{pmatrix}\\ \nonumber
&=\begin{pmatrix}1 & -s_0^{-1}e^{-2\pi i\alpha}e^{2x^2g_-(z)} \\ 0 & 1\end{pmatrix}
\begin{pmatrix}0 & s_0^{-1}e^{-2\pi i\alpha} \\ -s_0e^{2\pi i\alpha} & 0\end{pmatrix}
\begin{pmatrix}1 & -s^{-1}_0e^{-2\pi i\alpha}e^{2x^2g_+(z)} \\ 0 & 1\end{pmatrix}\\
&=:\widetilde{S}_{U_1}\widetilde{S}_{P_-}\widetilde{S}_{U_2}, \label{decom1}
\end{align}
\begin{align}\nonumber
(\widetilde{S}_3\widetilde{S}_4\widetilde{S}_5)^{-1}&=\begin{pmatrix}s_*
e^{x^2(g_-(z)-g_+(z))} & s_1(e^{-2\pi i\alpha}+s_*)\\ s_0e^{2\pi i\alpha} & \overline{s}_*e^{x^2(g_+(z)-g_-(z))}\end{pmatrix}\\ \nonumber
&=\begin{pmatrix}1 & 0 \\ \frac{\overline{s}_*e^{-2x^2g_-(z)}}{s_1(e^{-2\pi i\alpha}+s_*)} & 1\end{pmatrix}
\begin{pmatrix}0 &  (|s_*|^2-1)s_0^{-1}e^{-2\pi i\alpha}\\ s_0e^{2\pi i\alpha}(1-|s_*|^2)^{-1} & 0\end{pmatrix}
\\ \nonumber
  &~~~~\times \begin{pmatrix}1 & 0 \\ \frac{s_*e^{-2x^2g_+(z)}}{s_1(e^{-2\pi i\alpha}+s_*)} & 1\end{pmatrix}\\ \label{decom3}
&=:\widetilde{S}_{L_1}\widetilde{S}_P\widetilde{S}_{L_2},
\end{align}
\begin{align} \nonumber
\widetilde{S}_7\widetilde{S}_8\widetilde{S}_1&=\begin{pmatrix}\overline{s}_*
e^{-2\pi i\alpha}e^{x^2(g_-(z)-g_+(z))} & s_1(e^{-2\pi i\alpha}+s_*) \\
s_0e^{2\pi i\alpha} & s_*e^{2\pi i\alpha}e^{x^2(g_+(z)-g_-(z))}\end{pmatrix}\\ \nonumber
&=\begin{pmatrix}1 & 0 \\ \frac{s_*e^{2\pi i\alpha}e^{-2x^2g_-(z)}}{s_1(e^{-2\pi i\alpha}+s_*)} & 1\end{pmatrix}
\begin{pmatrix}0 &  (|s_*|^2-1)s_0^{-1}e^{-2\pi i\alpha}\\ s_0e^{2\pi i\alpha}(1-|s_*|^2)^{-1} & 0\end{pmatrix}
\\ \nonumber
  &~~~~\times
\begin{pmatrix}1 & 0 \\ \frac{\overline{s}_*e^{-2\pi i\alpha}e^{-2x^2g_+(z)}}{s_1(e^{-2\pi i\alpha}+s_*)} & 1\end{pmatrix}\\ \label{decom4}
&=:\widetilde{S}_{L_3}\widetilde{S}_P\widetilde{S}_{L_4},
\end{align}
and
\begin{align} \nonumber
\widetilde{S}_{8}&=\begin{pmatrix}e^{-2\pi i\alpha}e^{x^2(g_-(z)-g_+(z))} & 0\\ s_0e^{2\pi i\alpha} & e^{2\pi i\alpha}e^{x^2(g_+(z)-g_-(z))} \end{pmatrix}\\ \nonumber
&=\begin{pmatrix}1 & s_0^{-1}e^{-4\pi i\alpha}e^{2x^2g_-(z)} \\ 0 & 1\end{pmatrix}
\begin{pmatrix}0 & -s_0^{-1}e^{-2\pi i\alpha} \\ s_0e^{2\pi i\alpha} & 0\end{pmatrix}
\begin{pmatrix}1 & s^{-1}_0e^{2x^2g_+(z)} \\ 0 & 1\end{pmatrix}\\
&=:\widetilde{S}_{U_3}\widetilde{S}_{P_+}\widetilde{S}_{U_4}.\label{decom2}
\end{align}
In the above factorizations, we have utilized the property
$$g_+(z)+g_-(z)=0, \quad \mathrm{for} \quad z\in (z_{2,-},z_{2,+}), $$
 and the complex conjugate relation
\begin{equation}\label{sstarbar}
\overline{s}_*=s_0s_1e^{2\alpha\pi i}+1.
\end{equation}

Now,
we define the third transformation $\mathbf{U}\mapsto \mathbf{T}$ as
\begin{equation}\label{T(z)}
\mathbf{T}(z)=\left\{\begin{aligned}
&\mathbf{U}(z), &\mathrm{for}\ &z\ \mathrm{outside\ the\ four\ lens\ regions},\\
&\mathbf{U}(z)\widetilde{S}_{U_2}, &\mathrm{for}\ &z\ \mathrm{in\ the\ upper\ part\ of\ the\ first\ lens\ region},\\
&\mathbf{U}(z)\widetilde{S}_{U_1}^{-1}, &\mathrm{for}\ &z\ \mathrm{in\ the\ lower\ part\ of\ the\ first\ lens\ region},\\
&\mathbf{U}(z)\widetilde{S}_{L_2}^{-1}, &\mathrm{for}\ &z\ \mathrm{in\ the\ upper\ part\ of\ the\ second\ lens\ region},\\
&\mathbf{U}(z)\widetilde{S}_{L_1}, &\mathrm{for}\ &z\ \mathrm{in\ the\ lower\ part\ of\ the\ second\ lens\ region},\\
&\mathbf{U}(z)\widetilde{S}_{L_4}^{-1}, &\mathrm{for}\ &z\ \mathrm{in\ the\ upper\ part\ of\ the\ third\ lens\ region},\\
&\mathbf{U}(z)\widetilde{S}_{L_3}, &\mathrm{for}\ &z\ \mathrm{in\ the\ lower\ part\ of\ the\ third\ lens\ region}\\
&\mathbf{U}(z)\widetilde{S}_{U_4}^{-1}, &\mathrm{for}\ &z\ \mathrm{in\ the\ upper\ part\ of\ the\ fourth\ lens\ region},\\
&\mathbf{U}(z)\widetilde{S}_{U_3}, &\mathrm{for}\ &z\ \mathrm{in\ the\ lower\ part\ of\ the\ fourth\ lens\ region},
\end{aligned}
\right.
\end{equation}
where the lens regions are illustrated in  Figure \ref{Deforma2} and the same notations are used to denote the analytic extensions of the corresponding jump matrices.

\begin{figure}[t]
  \centering
  \includegraphics[width=13cm,height=7cm]{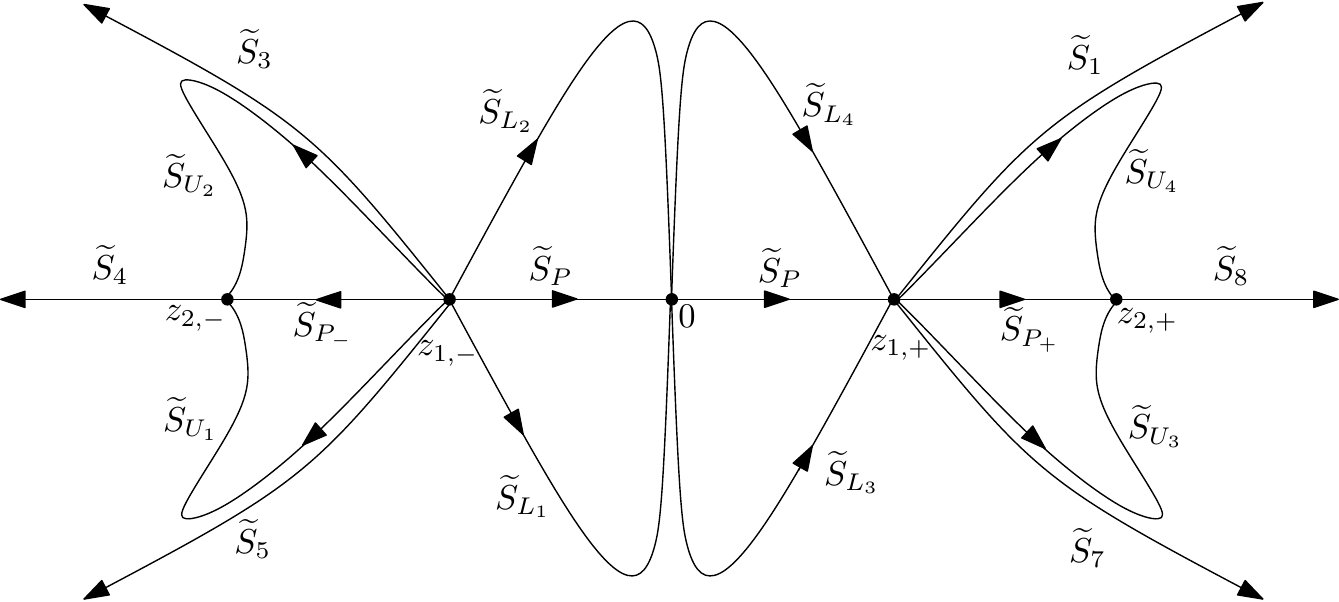}\\
  \caption{The second deformation of the jump contour} \label{Deforma2}
\end{figure}

Immediately, $\mathbf{T}(z)$ solves a RH problem whose jump contour and jump matrices are shown in Figure \ref{Deforma2}. 
We proceed to
blow up the four lens in Figure \ref{Deforma2}. Consequently, we obtain the following equivalent RH problem for $\mathbf{T}(z)$.

\begin{figure}[t]
  \centering
  \includegraphics[width=13cm,height=8cm]{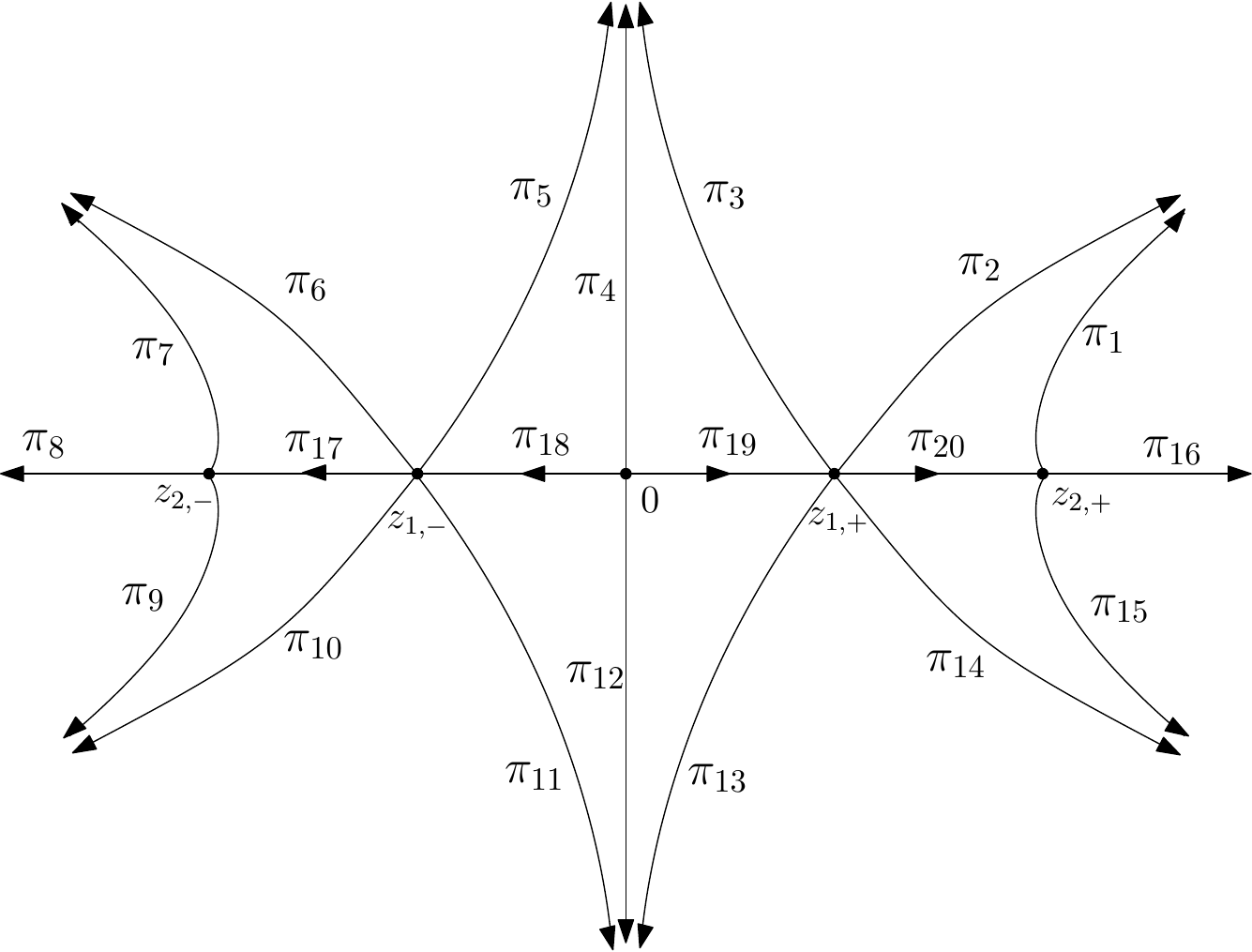}\\
  \caption{The jump contour $\Sigma_\mathbf{T}$ } \label{Deforma3}
\end{figure}

\subsection*{RH problem for $\mathbf{T}(z)$}
\begin{description}
  \item{(1)} $\mathbf{T}(z)$ is analytic for $z\in \mathbb{C}\setminus\Sigma_{\mathbf{T}}$, where $\Sigma_\mathbf{T}$ is shown in Figure \ref{Deforma3}.
  \item{(2)} We have $\mathbf{T}_+(z)=\mathbf{T}_-(z)J_k(z)$ for $z\in\pi_k$, $k=1,\cdots,20$, where
  \begin{align*}
  &J_1(z)=\begin{pmatrix}1 & -s_0^{-1}e^{2x^2g(z)}\\ 0 & 1 \end{pmatrix},
   &&
  J_2(z)=\begin{pmatrix}1 & s_*s_0^{-1}e^{2x^2g(z)}\\ 0 & 1 \end{pmatrix},\\
  &J_3(z)=\begin{pmatrix}1 & 0\\ -\frac{\overline{s}_*e^{-2\pi i\alpha}e^{-2x^2g(z)}}
  {s_1(e^{-2\pi i\alpha}+s_*)} & 1 \end{pmatrix}, &&
  J_4(z)=\begin{pmatrix}1 & 0\\ \frac{(e^{-2\pi i\alpha}-1)e^{-2x^2g(z)}}
  {s_1(e^{-2\pi i\alpha}+s_*)} & 1 \end{pmatrix},\\
  &J_5(z)=\begin{pmatrix}1 & 0\\ \frac{s_*e^{-2x^2g(z)}}
  {s_1(e^{-2\pi i\alpha}+s_*)} & 1 \end{pmatrix}, &&
  J_6(z)=\begin{pmatrix}1 & -\overline{s}_*s_0^{-1}e^{-2\pi i\alpha}
  e^{2x^2g(z)} \\ 0 & 1 \end{pmatrix},\\
  &J_7(z)=J_9(z)=\begin{pmatrix}1 &
  \frac{e^{2x^2g(z)}}{s_0e^{2\pi i\alpha}} \\ 0 & 1 \end{pmatrix},&&
  J_8(z)=\begin{pmatrix}1 & 0\\ -s_0e^{2\pi i\alpha}e^{-2x^2g(z)} & 1 \end{pmatrix},\\
  &J_{10}(z)=\begin{pmatrix}1 & -\frac{s_*
  e^{2x^2g(z)}}{s_0e^{2\pi i\alpha}}\\ 0 & 1 \end{pmatrix}
  ,&&
  J_{11}(z)=\begin{pmatrix}1 & 0\\ \frac{\overline{s}_*e^{-2x^2g(z)}}
  {s_1(e^{-2\pi i\alpha}+s_*)} & 1 \end{pmatrix},\\
  &J_{12}(z)=\begin{pmatrix}1 & 0\\ \frac{(e^{2\pi i\alpha}-1)e^{-2x^2g(z)}}
  {s_1(e^{-2\pi i\alpha}+s_*)} & 1 \end{pmatrix},&&
   J_{13}(z)=\begin{pmatrix}1 & 0\\ -\frac{s_*e^{2\pi i\alpha}e^{-2x^2g(z)}}
  {s_1(e^{-2\pi i\alpha}+s_*)} & 1 \end{pmatrix},\\
  &J_{14}(z)=\begin{pmatrix}1 & \overline{s}_*s_0^{-1}e^{-4\pi i\alpha}e^{2x^2g(z)} \\ 0 & 1 \end{pmatrix},
  &&
  J_{15}(z)=\begin{pmatrix}1 & -s_0^{-1}e^{-4\pi i\alpha}e^{2x^2g(z)} \\ 0 & 1 \end{pmatrix},\\
  &J_{16}(z)=\begin{pmatrix}e^{-2\pi i\alpha} & 0\\ s_0e^{2\pi i\alpha}e^{-2x^2g(z)} & e^{2\pi i\alpha} \end{pmatrix},&&
  J_{17}(z)=J_{20}(z)^{-1}=\begin{pmatrix}0 & s_0^{-1}e^{-2\pi i\alpha}\\ -s_0e^{2\pi i\alpha} & 0 \end{pmatrix},
  \end{align*}
  and
  \begin{equation*}
  J_{19}(z)=J_{18}(z)^{-1}=\begin{pmatrix}0 & -s_0^{-1}e^{-2\pi i\alpha}(1-|s_*|^2) \\ s_0e^{2\pi i\alpha}(1-|s_*|^2)^{-1} & 0\end{pmatrix}.
  \end{equation*}
  \item{(3)} As $z\rightarrow \infty$, $\mathbf{T}(z)=\left(\mathbf{I}+O(z^{-1})\right)z^{\alpha\sigma_3}$, where $\arg z\in(0,2\pi)$.
  \item{(4)} As $z\to 0$, $\mathbf{T}(z)$ has the following asymptotic behavior
 \begin{equation}\label{T0}
\mathbf{T}(z)=\mathbf{T}_0(z)z^{\alpha\sigma_{3}}
E_0S_1e^{-x^2g(z)\sigma_{3}}J_3(z),\quad \arg z\in (\pi/4, {\pi}/{2}),
\end{equation}
where $\mathbf{T}_0(z)$ is analytic in the neighborhood of $z=0$. The behaviors of $\mathbf{T}(z)$ in other regions  are determined by \eqref{T0} and  the jump relations satisfied by $\mathbf{T}(z)$.
\end{description}

Taking into account the lower and upper triangular structures of the jump matrices, the properties of  $\Re g(z)$ on the anti-Stokes curves (see Figure \ref{ASC}), and the fact that $\Re g(z)>0$ on the imaginary axis, we see that the jump matrices on $\pi_k$, $k=1,\cdots,15$ of the RH problem for $\mathbf{T}(z)$ tend to the identity matrix exponentially fast when $x\rightarrow -\infty$. Therefore, it is expected that the
dominant contribution to the asymptotics of $\mathbf{T}(z)$ as  $x\rightarrow -\infty$ will come from the remaining jumps along the line $(z_{2,-},+\infty)$ and the neighborhoods of the origin $z=0$ and the four saddle points $z_{1,\pm}=\pm\sqrt{2/3}$, $z_{2,\pm}=\pm\sqrt{8/3}$. We are in a position  to construct the global parametrix  solving the constant jumps along $(z_{2,-},+\infty)$ and local parametrices near the origin and the four saddle points.

\subsection{Global parametrix }
We need to solve the following RH problem.
\subsection*{RH problem for $\mathbf{P}^{(\infty)}(z)$}
\begin{description}
  \item{(1)} $\mathbf{P}^{(\infty)}(z)$ is analytic for $z\in \mathbb{C}\setminus [z_{2,-},+\infty)$.

  \item{(2)} $\mathbf{P}^{(\infty)}(z)$ satisfies the following jump conditions
  \begin{equation}\label{Pinftyjump}
  \mathbf{P}^{(\infty)}_{+}(z)=\mathbf{P}^{(\infty)}_{-}(z)J^{(\infty)}(z),
  \end{equation}
  where
  $$
  J^{(\infty)}(z)=\left\{\begin{aligned}
  &\begin{pmatrix}0 & -e^{-2\pi i\alpha}s^{-1}_0 \\ e^{2\pi i\alpha} s_0& 0\end{pmatrix},& z&\in(z_{2,-},z_{1,-}), \\
  &\begin{pmatrix}0 & -e^{-2\pi i\alpha} s_0^{-1}(1-|s_*|^2)\\ e^{2\pi i\alpha} s_0(1-|s_*|^2)^{-1} & 0\end{pmatrix}, & z&\in(z_{1,-},z_{1,+}), \\
  &\begin{pmatrix}0 & -e^{-2\pi i\alpha} s^{-1}_0\\ e^{2\pi i\alpha} s_0& 0 \end{pmatrix}, & z&\in(z_{1,+},z_{2,+}),\\
  &\begin{pmatrix}e^{-2\pi i\alpha} & 0 \\ 0 & e^{2\pi i\alpha} \end{pmatrix}, & z&\in(z_{2,+},+\infty).
  \end{aligned}\right.
  $$
  \item{(3)}  As $z\rightarrow \infty$, we have
  \begin{equation}\label{Pinftyinfty}
  \mathbf{P}^{(\infty)}(z)=\left(\mathbf{I}+O(z^{-1})\right)z^{\alpha\sigma_3},
  \end{equation}
  where the branch cut of $z^{\alpha}$ is taken along $\mathbb{R}_+$ so that $\arg z\in(0,2\pi)$.
\end{description}

The solution of the above RH problem can be constructed as follows
\begin{equation}\label{Pinfty}
\mathbf{P}^{(\infty)}(z)=s_0^{-\frac{\sigma_3}{2}}f^{-\sigma_3}_{\infty}\mathbf{X}(z)
f(z)^{\sigma_3}
s_0^{\frac{\sigma_3}{2}},
\end{equation}
where $f(z)$ is the Szeg\H{o} function
\begin{equation}\label{f(z)}
f(z)= \left(\frac{3}{8}\right)^{\frac{\alpha}{2}}
\left(z+\sqrt{z^2-\frac{8}{3}}\right)^{\alpha}
\left(\frac{\left((2-\sqrt{3})z-i\sqrt{z^2-\frac{8}{3}}\right)^2-\frac{8}{3}}
{\left((2-\sqrt{3})z+i\sqrt{z^2-\frac{8}{3}}\right)^2-\frac{8}{3}}\right)^{{\beta}},
\end{equation}
$\mathbf{X}(z)$ is given by
\begin{equation}\label{X(z)}
\mathbf{X}(z)=\begin{pmatrix}1 & 1 \\ -i & i\end{pmatrix} \omega(z)^{\sigma_{3}}
\begin{pmatrix} 1 & 1 \\ -i & i\end{pmatrix}^{-1},\quad \omega(z)=\left(\frac{z-\sqrt{\frac{8}{3}}}
{z+\sqrt{\frac{8}{3}}}\right)^{\frac{1}{4}},
\end{equation}
and $f_{\infty}$ is defined by
\begin{equation}\label{finfty}
f_{\infty}=2^{-\frac{\alpha}{2}}3^{\frac{\alpha}{2}}e^{\frac{\pi i{\beta}}{3}}.
\end{equation}
The branches of the multi-valued functions in \eqref{Pinfty}-\eqref{X(z)} are chosen by specifying
$$
\arg z\in(-\pi,\pi),\quad
\arg\left(z\pm\textstyle\sqrt{\frac{8}{3}}\right)\in(-\pi,\pi),\qquad
\arg \left(\textstyle z+\sqrt{z^2-\frac{8}{3}}\right)\in(0,2\pi), 
$$
and
$$
\arg\left(\frac{\left((2-\sqrt{3})z-i\sqrt{z^2-\frac{8}{3}}\right)^2-\frac{8}{3}}
{\left((2-\sqrt{3})z+i\sqrt{z^2-\frac{8}{3}}\right)^2-\frac{8}{3}}\right)
\in(-\pi,\pi).
$$
The constant ${\beta}$ in \eqref{f(z)} and \eqref{finfty} is given by
\begin{equation}\label{nu}
{\beta}=-\frac{1}{2\pi i}\ln(1-|s_*|^2),
\end{equation}
where the logarithm takes the principal value. In the case $0<|s_*|<1$, it is easily seen that
\begin{equation}\label{nuimag}
{\beta}\neq0\quad \mathrm{and}\quad{\beta}\in i\mathbb{R}.
\end{equation}
From \eqref{f(z)}  and \eqref{X(z)}, as $z\rightarrow\infty$, straightforward   computation yields
\begin{align}\label{X(z)atinfty}
\mathbf{X}(z)&=\mathbf{I}+\sqrt{\frac{2}{3}}\sigma_2z^{-1}
+\frac{\mathbf{I}}{3}z^{-2}+O(z^{-3}),\\
f(z)&=2^{-\frac{\alpha}{2}}3^{\frac{\alpha}{2}}
e^{\frac{\pi i{\beta}}{3}}\left[1-\frac{2}{3}\left(\alpha+\sqrt{3}i{\beta}\right)z^{-2}
+O(z^{-4})\right]z^{\alpha}.\label{fatinfty}
\end{align}

\subsection{Local parametrices near the saddle points $z_{2,\pm}$}\label{sec:AiParametrix}
We will construct two parametrices $\mathbf{P}^{(2,\pm)}(z)$ satisfying the same jumps as $\mathbf{T}(z)$ on $\Sigma_\mathbf{T}$ (see Figure \ref{Deforma3}), respectively  in the neighborhoods $U(z_{2,\pm},\delta)=\{z\in\mathbb{C}\mid|z-z_{2,\pm}|<\delta\}$ of the saddle points $z_{2,\pm}$, and matching with $\mathbf{P}^{(\infty)}(z)$ on the boundaries $\partial U(z_{2,\pm},\delta)=\{z\in \mathbb{C}\mid|z-z_{2,\pm}|=\delta\}$.
\subsection*{RH problem for $\mathbf{P}^{(2,+)}(z)$}
\begin{description}
\item{(1)} $\mathbf{P}^{(2,+)}(z)$ is analytic for $z\in U(z_{2,+},\delta)\setminus\Sigma_\mathbf{T}$.
\item{(2)} $\mathbf{P}^{(2,+)}(z)$ satisfies the same jumps as $\mathbf{T}(z)$ on $U(z_{2,+},\delta)\cap\Sigma_\mathbf{T}$.
\item{(3)} On the circular boundary $\partial U(z_{2,+},\delta)$, it holds that
\begin{equation}\label{P2+matching}
\mathbf{P}^{(2,+)}(z)=\left(\mathbf{I}+O(|x|^{-2})\right)\mathbf{P}^{(\infty)}(z),\quad \mathrm{as}\quad x\rightarrow -\infty.
\end{equation}
\end{description}
To   solve the above RH problem, we define a conformal mapping
\begin{equation}\label{varphi1}
\varphi_1(z)=\left(\frac{3}{2}g(z)\right)^{\frac{2}{3}},
\end{equation}
where the branch is specified by the asymptotic condition
\begin{equation}\label{varphi1beha}
\varphi_1(z)=2^{\frac{5}{6}}3^{-\frac{1}{6}}
\textstyle\left(z-\sqrt{\frac{8}{3}}\right)
\left(1+o(1)\right),\quad \mathrm{as} \quad z\rightarrow \sqrt{\frac{8}{3}}.
\end{equation}
Then, the solution to the above RH problem can be explicitly constructed in terms of the Airy function:
\begin{equation}\label{P2+}
\mathbf{P}^{(2,+)}(z)=\mathbf{E}^{(2,+)}(z)\Phi^{\mathrm{(Ai)}}
\left(|x|^{\frac{4}{3}}\varphi_1(z)\right)
\left(s_0e^{2\pi i\alpha}\right)^{-\frac{\sigma_3}{2}}\sigma_{1}e^{\mp\pi i\alpha\sigma_3}e^{-x^2g(z)\sigma_3},\quad \pm\Im z>0,
\end{equation}
where $\Phi^{\mathrm{(Ai)}}$ is the Airy model parametrix (see Appendix \ref{AP}), and $\mathbf{E}^{(2,+)}(z)$ is given by
\begin{equation}\label{E2+}
\mathbf{E}^{(2,+)}(z)=\mathbf{P}^{(\infty)}(z)e^{\pm\pi i\alpha\sigma_3}\sigma_{1}\left(s_0e^{2\pi i\alpha}\right)^{\frac{\sigma_3}{2}}
\frac{1}{\sqrt{2}} \begin{pmatrix}1 & -i\\ -i &1\end{pmatrix}|x|^{\frac{\sigma_3}{3}}\varphi_1(z)^{\frac{\sigma_3}{4}},\ \pm\Im z>0.
\end{equation}
Here, the branch of $\varphi_1(z)^{\frac{1}{4}}$ is chosen such that $\arg\varphi_1(z)\in(-\pi,\pi)$. This means that
\begin{equation}\label{etajump}
\left(\varphi_1(z)^{\frac{\sigma_3}{4}}\right)_+=
\left(\varphi_1(z)^{\frac{\sigma_3}{4}}\right)_-e^{\frac{\pi i}{2}\sigma_3},\quad z\in (z_{2,-},z_{2,+}).
\end{equation}
From \eqref{Pinftyjump} and \eqref{etajump}, it is readily verified that $\mathbf{E}^{(2,+)}(z)$ is analytic in the neighborhood $U(z_{2,+},\delta)$.
Finally, combining \eqref{P2+}, \eqref{E2+} with the asymptotic behavior \eqref{AiryAsyatinfty} gives us the matching condition \eqref{P2+matching}.

\subsection*{RH problem for $\mathbf{P}^{(2,-)}(z)$}
\begin{description}
\item{(1)} $\mathbf{P}^{(2,-)}(z)$ is analytic for $z\in U(z_{2,-},\delta)\setminus\Sigma_\mathbf{T}$.
\item{(2)} On $U(z_{2,-},\delta)\cap\Sigma_\mathbf{T}$, $\mathbf{P}^{(2,-)}(z)$ shares the same jumps as $\mathbf{T}(z)$.
\item{(3)} On the boundary $\partial U(z_{2,-},\delta)$, we have
\begin{equation}\label{P2-matching}
\mathbf{P}^{(2,-)}(z)=\left(\mathbf{I}+O(|x|^{-2})\right)\mathbf{P}^{(\infty)}(z),\quad \mathrm{as}\quad x\rightarrow -\infty.
\end{equation}
\end{description}
By symmetry, the solution to the above RH problem is given by
\begin{equation}\label{P2-}
\mathbf{P}^{(2,-)}(z)=\mathbf{E}^{(2,-)}(z)\Phi^{(\mathrm{Ai})}
\left(|x|^{\frac{4}{3}}\varphi_1(-z)\right)
\left(s_0e^{2\pi i\alpha}\right)^{-\frac{\sigma_3}{2}}\sigma_2e^{-x^2g(z)\sigma_3},
\end{equation}
where $\varphi_1(z)$ is the conformal mapping defined by \eqref{varphi1}  and $\mathbf{E}^{(2,-)}(z)$ is given by
\begin{equation}\label{E2-}
\mathbf{E}^{(2,-)}(z)=\mathbf{P}^{(\infty)}(z)\sigma_2\left(s_0e^{2\pi i\alpha}\right)^{\frac{\sigma_3}{2}}
\frac{1}{\sqrt{2}} \begin{pmatrix}1 & -i\\ -i &1\end{pmatrix}|x|^{\frac{\sigma_3}{3}}\varphi_1(-z)^{\frac{\sigma_3}{4}}.
\end{equation}

\subsection{Local parametrices near the saddle points $z_{1,\pm}$}
In this subsection, we seek two parametrices $\mathbf{P}^{(1,\pm)}(z)$ satisfying the same jumps  as $\mathbf{T}(z)$ on  $\Sigma_\mathbf{T}$ (see Figure \ref{Deforma3}) in the neighborhoods $U(z_{1,\pm},\delta)=\{z\in\mathbb{C}\mid|z-z_{1,\pm}|<\delta\}$ of the saddle points $z_{1,\pm}$, and matching with $\mathbf{P}^{(\infty)}(z)$ on the boundaries $\partial U(z_{1,\pm},\delta)=\{z\in \mathbb{C}\mid|z-z_{1,\pm}|=\delta\}$.

\subsection*{RH problem for $\mathbf{P}^{(1,+)}(z)$}
\begin{description}
\item{(1)} $\mathbf{P}^{(1,+)}(z)$ is analytic for $z\in U(z_{1,+},\delta)\setminus\Sigma_\mathbf{T}$.
\item{(2)} $\mathbf{P}^{(1,+)}(z)$ shares the same jumps as $\mathbf{T}(z)$ on $U(z_{1,+},\delta)\cap\Sigma_\mathbf{T}$.
\item{(3)} On the circle $\partial U(z_{1,+},\delta)$, $\mathbf{P}^{(1,+)}(z)$ satisfies
\begin{equation}\label{P1+matching}
\mathbf{P}^{(1,+)}(z)=\left(\mathbf{I}+O(|x|^{-1})\right)\mathbf{P}^{(\infty)}(z),\quad \mathrm{as}\quad x\rightarrow -\infty.
\end{equation}
\end{description}
Let us define the conformal mapping
\begin{equation}\label{varphi2}
\varphi_2(z)=\left\{\begin{aligned}
&2\sqrt{-\textstyle\frac{\sqrt{3}i}{6}-g(z)},\quad& \Im z>0,\\
&2\sqrt{-\textstyle\frac{\sqrt{3}i}{6}+g(z)},\quad& \Im z<0,
\end{aligned}\right.
\end{equation}
where the branches of the square roots are specified by choosing
\begin{equation}\label{varphi2beha}
\varphi_2(z)=e^{-\frac{\pi i}{4}}2\cdot3^{-\frac{1}{4}}\textstyle \left(z-\sqrt{\frac{2}{3}}\right)
\left(1+o(1)\right),\quad \mathrm{as}\quad z\rightarrow\sqrt{\frac{2}{3}}.
\end{equation}
Let $\Phi^{\mathrm{(PC)}}$ be the parabolic cylinder model parametrix with parameter ${\beta}$ given by \eqref{nu}; see Appendix \ref{PCP}. The solution to the above RH problem can be constructed as follows:
\begin{align}\label{P1+}
\mathbf{P}^{(1,+)}(z)=\mathbf{E}^{(1,+)}(z)\Phi^{\mathrm{(PC)}}\left(|x|\varphi_2(z)\right)
\left(\frac{s_*}{h_0}\right)^{\frac{\sigma_3}{2}}
\left\{\begin{aligned}
&s_0^{-\frac{\sigma_3}{2}}(-\sigma_{1})e^{-x^2g(z)\sigma_3},& \Im z>0,\\
&s_0^{\frac{\sigma_3}{2}}\sigma_3e^{2\pi i(\alpha+{\beta})\sigma_3}
e^{-x^2g(z)\sigma_3},& \Im z<0,
\end{aligned}\right.
\end{align}
where $h_{0}$ is the Stokes multiplier in \eqref{h0} and $\mathbf{E}^{(1,+)}(z)$ is given by
\begin{equation}\label{E1+}
\mathbf{E}^{(1,+)}(z)=\mathbf{W}^{(+)}(z)\left(\frac{s_*}{h_0}\right)^{-\frac{\sigma_3}{2}}
|x|^{{\beta}\sigma_3}e^{i\frac{\sqrt{3}}{6}x^2
\sigma_{3}}2^{-\frac{\sigma_3}{2}}\begin{pmatrix}|x|\varphi_2(z) & 1 \\ 1 & 0\end{pmatrix}
\end{equation}
with
\begin{equation}\label{W+}
\mathbf{W}^{(+)}(z)=\mathbf{P}^{(\infty)}(z)\left\{\begin{aligned}
&(-\sigma_1)s_0^{\frac{1}{2}\sigma_3}\varphi_2(z)^{{\beta}\sigma_3},
&\Im z&>0,\\
&e^{-2\pi i(\alpha+{\beta})\sigma_3}\sigma_3s_0^{-\frac{1}{2}\sigma_3}
\varphi_2(z)^{{\beta}\sigma_3}, & \Im z&<0.
\end{aligned}\right.
\end{equation}
The branch of $\varphi_2(z)^{{\beta}}$ is chosen by specifying $\arg\varphi_2(z)\in(0,2\pi)$, this implies the jump relations
\begin{equation}\label{varphijump}
\left\{\begin{aligned}
&\left(\varphi_2(z)^{{\beta}}\right)_+=
\left(\varphi_2(z)^{{\beta}}\right)_-e^{2\pi i{\beta}}, \quad &z&\in(z_{1,+},z_{2,+}),\\
&\left(\varphi_2(z)^{{\beta}}\right)_+
=\left(\varphi_2(z)^{{\beta}}\right)_-, \quad &z&\in(z_{2,-},z_{1,+}).
\end{aligned}\right.
\end{equation}
From the jump conditions \eqref{Pinftyjump} and \eqref{varphijump}, it is readily seen that $\mathbf{W}^{(+)}(z)$ is analytic  in the neighborhood $U(z_{1,+},\delta)$. In particular, it follows from \eqref{Pinfty}, \eqref{f(z)}, \eqref{X(z)}, \eqref{varphi2beha} and \eqref{W+} that
\begin{equation}\label{W+atz+}
\mathbf{W}^{(+)}(\textstyle\sqrt{\frac{2}{3}})=-s_0^{-\frac{\sigma_3}{2}}
f_{\infty}^{-\sigma_3}
\begin{pmatrix}1 & 1 \\ -i & i\end{pmatrix} 3^{-\frac{\sigma_{3}}{4}}
e^{\frac{\pi i}{4}\sigma_3}\begin{pmatrix} 1 & 1 \\
-i & i\end{pmatrix}^{-1}2^{-\frac{{\beta}}{2}\sigma_3}
3^{-\frac{{\beta}}{4}\sigma_3}e^{\frac{\pi i{\beta}}{4}\sigma_3}
e^{\frac{\pi i\alpha}{3}\sigma_3}\sigma_1 .
\end{equation}
Finally, a combination of \eqref{Pinfty}, \eqref{nuimag}, \eqref{PCAsyatinfty} and \eqref{P1+} gives
\begin{align}\label{matchingcondition+}
&\mathbf{P}^{(1,+)}(z)\mathbf{P}^{(\infty)}(z)^{-1}\nonumber\\
&=\mathbf{W}^{(+)}(z)\begin{pmatrix}
1+\frac{{\beta}({\beta}+1)}{2|x|^2\varphi_2^2}+O\left(|x|^{-4}\right) & \frac{e^{i\frac{\sqrt{3}}{3}x^2}|x|^{2{\beta}}}{h_0^{-1}s_*}
\left(\frac{{\beta}}{|x|\varphi_2}
+O\left(|x|^{-3}\right)\right) \\ \frac{h_0^{-1}s_*}{e^{i\frac{\sqrt{3}}{3}x^2}|x|^{2{\beta}}}
\left(\frac{1}{|x|\varphi_2}+O\left(|x|^{-3}\right)\right) & 1-\frac{{\beta}({\beta}-1)}{2|x|^2\varphi_2^2}+O\left(|x|^{-4}\right)
\end{pmatrix}\mathbf{W}^{(+)}(z)^{-1}\nonumber\\
&=\mathbf{I}+\mathbf{G}^{(+)}(z)|x|^{-1}+O\left(|x|^{-2}\right),
\end{align}
where
\begin{equation}\label{G+}
\mathbf{G}^{(+)}(z)=\mathbf{W}^{(+)}(z)\begin{pmatrix}0 & \frac{\beta h_0e^{i\frac{\sqrt{3}}{3}x^2}|x|^{2\beta}}
{s_*\varphi_2(z)} \\ \frac{s_*|x|^{-2\beta}}
{h_0e^{ i\frac{\sqrt{3}}{3}x^2}\varphi_2(z)} & 0 \end{pmatrix}\mathbf{W}^{(+)}(z)^{-1}.
\end{equation}

\subsection*{RH problem for $\mathbf{P}^{(1,-)}(z)$}
\begin{description}
\item{(1)} $\mathbf{P}^{(1,-)}(z)$ is analytic for all $z\in U(z_{1,-},\delta)\setminus\Sigma_\mathbf{T}$.
\item{(2)} $\mathbf{P}^{(1,-)}(z)$ satisfies the same jump conditions as $\mathbf{T}(z)$ on $U(z_{1,-},\delta)\cap\Sigma_\mathbf{T}$.
\item{(3)} $\mathbf{P}^{(1,-)}(z)$ fulfills the matching condition on the boundary $\partial U(z_{1,-},\delta)$:
\begin{equation}\label{P1-matching}
\mathbf{P}^{(1,-)}(z)=\left(\mathbf{I}+O(|x|^{-1})\right)\mathbf{P}^{(\infty)}(z),\quad \mathrm{as}\quad x\rightarrow -\infty.
\end{equation}
\end{description}
Similarly, we  introduce the conformal mapping
\begin{equation}\label{varphi3}
\varphi_3(z)=\left\{\begin{aligned}
&2\sqrt{-\textstyle\frac{\sqrt{3}i}{6}+g(z)},\quad& \Im z>0,\\
&2\sqrt{-\textstyle\frac{\sqrt{3}i}{6}-g(z)},\quad& \Im z<0,
\end{aligned}\right.
\end{equation}
where the branches of square roots are taken so that
\begin{equation}\label{varphi3beha}
\varphi_3(z)=e^{\frac{3\pi i}{4}}2\cdot3^{-\frac{1}{4}}\textstyle \left(z+\sqrt{\frac{2}{3}}\right)
\left(1+o(1)\right),\quad \mathrm{as}\quad z\rightarrow-\sqrt{\frac{2}{3}}.
\end{equation}
The solution to the above RH problem can be also constructed in terms of the parabolic cylinder functions as follows:
\begin{align}\label{P1-}
&\mathbf{P}^{(1,-)}(z)=\mathbf{E}^{(1,-)}(z)\Phi^{\mathrm{(PC)}}\left(|x|\varphi_3(z)\right)
\left(\frac{s_*}{h_0}\right)^{\frac{\sigma_3}{2}}\left\{\begin{aligned}
&e^{2\pi i{\beta}\sigma_3}\left(s_0e^{2\pi i\alpha}\right)^{\frac{\sigma_3}{2}}e^{-x^2g(z)\sigma_3},& \Im z>0,\\
&\left(s_0e^{2\pi i\alpha}\right)^{-\frac{\sigma_3}{2}}
\sigma_{3}\sigma_{1}e^{-x^2g(z)\sigma_3},
& \Im z<0,\end{aligned}\right.
\end{align}
where $\Phi^{\mathrm{(PC)}}$ is the parabolic cylinder model parametrix, $h_{0}$ is given in \eqref{h0}, ${\beta}$ is defined by \eqref{nu} and $\mathbf{E}^{(1,-)}(z)$ is given by
\begin{equation}\label{E1-}
\mathbf{E}^{(1,-)}(z)=\mathbf{W}^{(-)}(z)\left(\frac{s_*}{h_0}\right)^{-\frac{\sigma_3}{2}}
|x|^{{\beta}\sigma_3}e^{\frac{i\sqrt{3}x^2}{6}
\sigma_{3}}2^{-\frac{\sigma_3}{2}}\begin{pmatrix}|x|\varphi_3(z) & 1 \\ 1 & 0\end{pmatrix}
\end{equation}
with
\begin{equation}\label{W-}
\mathbf{W}^{(-)}(z)=\mathbf{P}^{(\infty)}(z)\left\{\begin{aligned}
&\left(s_0e^{2\pi i\alpha}\right)^{-\frac{\sigma_3}{2}}
e^{-2\pi i{\beta}\sigma_3}\varphi_3(z)^{{\beta}\sigma_3},\ &\Im z&>0,\\
&\sigma_1\sigma_3\left(s_0e^{2\pi i\alpha}\right)^{\frac{\sigma_3}{2}}\varphi_3(z)^{{\beta}\sigma_3}, \ &\Im z&<0.
\end{aligned}\right.
\end{equation}
We point out that the branch of the function $\varphi_3(z)^{{\beta}}$ is chosen such that $\arg\varphi_3(z)\in(-\pi,\pi)$. This implies that
\begin{equation}\label{zetajump}
\left\{\begin{aligned}
&\left(\varphi_3(z)^{{\beta}}\right)_+=\left(\varphi_3(z)^{{\beta}}\right)_-e^{2\pi i{\beta}}, \quad &z&\in(z_{2,-},z_{1,-}),\\
&\left(\varphi_3(z)^{{\beta}}\right)_+=\left(\varphi_3(z)^{{\beta}}\right)_-, \quad &z&\in(z_{1,-},z_{2,+}).
\end{aligned}\right.
\end{equation}
Using the jump relations \eqref{Pinftyjump} and \eqref{zetajump}, it is straightforward to check that $\mathbf{W}^{(-)}(z)$ is analytic in the neighborhood $U(z_{1,-},\delta)$. Particularly, by substituting \eqref{Pinfty}, \eqref{f(z)}, \eqref{X(z)} and \eqref{varphi3beha} into \eqref{W-}, we have
\begin{equation}\label{W-atz-}
\mathbf{W}^{(-)}(-\textstyle\sqrt{\frac{2}{3}})=s_0^{-\frac{\sigma_3}{2}}
f_{\infty}^{-\sigma_3}\begin{pmatrix}1 & 1 \\ -i & i\end{pmatrix} 3^{\frac{1}{4}\sigma_{3}}e^{\frac{\pi i}{4}\sigma_3}\begin{pmatrix} 1 & 1 \\
-i & i\end{pmatrix}^{-1}2^{\frac{{\beta}}{2}\sigma_3}
3^{\frac{{\beta}}{4}\sigma_3}
e^{-\frac{\pi i\alpha}{3}\sigma_3}e^{-\frac{\pi i{\beta}}{4}\sigma_3}.
\end{equation}
Furthermore, using \eqref{Pinfty}, \eqref{nuimag}, \eqref{P1-} and the asymptotic behavior \eqref{PCAsyatinfty}, we obtain the matching condition
\begin{align}\label{matchingcondition-}
&\mathbf{P}^{(1,-)}(z)\mathbf{P}^{(\infty)}(z)^{-1}\nonumber \\
&=\mathbf{W}^{(-)}(z)\begin{pmatrix}
1+\frac{{\beta}({\beta}+1)}{2|x|^2\varphi_3^2}+O\left(|x|^{-4}\right) & \frac{e^{i\frac{\sqrt{3}}{3}x^2}|x|^{2{\beta}}}{h_0^{-1}s_*}
\left(\frac{{\beta}}{|x|\varphi_3}
+O\left(|x|^{-3}\right)\right)\\ \frac{h_0^{-1}s_*}{e^{i\frac{\sqrt{3}}{3}x^2}|x|^{2{\beta}}}
\left(\frac{1}{|x|\varphi_3}+O\left(|x|^{-3}\right)\right) & 1-\frac{{\beta}({\beta}-1)}{2|x|^2\varphi_3^2}+O\left(|x|^{-4}\right)
\end{pmatrix}\mathbf{W}^{(-)}(z)^{-1}\nonumber\\
&=\mathbf{I}+\mathbf{G}^{(-)}(z)|x|^{-1}+O\left(|x|^{-2}\right)
\end{align}
where
\begin{equation}\label{G-}
\mathbf{G}^{(-)}(z)=\mathbf{W}^{(-)}(z)\begin{pmatrix}0 & \frac{{\beta} h_0e^{i\frac{\sqrt{3}}{3}x^2}|x|^{2{\beta}}}
{s_*\varphi_3(z)}\\ \frac{s_*|x|^{-2{\beta}}}{h_0e^{i\frac{\sqrt{3}}{3}x^2}\varphi_3(z)} & 0 \end{pmatrix}\mathbf{W}^{(-)}(z)^{-1}.
\end{equation}

\subsection{Local parametrix near the origin}
In this subsection, we intend to find a function $\mathbf{P}^{(0)}(z)$ having the same jumps  as $\mathbf{T}(z)$ on $\Sigma_\mathbf{T}$ (see Figure \ref{Deforma3}) in the neighborhood $U(0,\delta)=\{z\in \mathbb{C}\mid|z|<\delta\}$ of the origin, and matching with $\mathbf{P}^{(\infty)}(z)$ on the boundary  $\partial U(0,\delta)=\{z\in \mathbb{C}\mid|z|=\delta\}$.

\subsection*{RH problem for $\mathbf{P}^{(0)}(z)$}
\begin{description}
\item{(1)} $\mathbf{P}^{(0)}(z)$ is analytic for $z\in U(0,\delta)\setminus\Sigma_\mathbf{T}$.
\item{(2)} $\mathbf{P}^{(0)}(z)$ satisfies the same jump conditions as $\mathbf{T}(z)$ on $U(0,\delta)\cap\Sigma_\mathbf{T}$.
\item{(3)} On the boundary $\partial U(0,\delta)$, we have
\begin{equation}\label{P0matching}
\mathbf{P}^{(0)}(z)=\left(\mathbf{I}+O(|x|^{-2})\right)\mathbf{P}^{(\infty)}(z),\quad \mathrm{as}\quad x\rightarrow -\infty.
\end{equation}
\item{(4)} $\mathbf{P}^{(0)}(z)$ has the same asymptotic behavior as $\mathbf{T}(z)$ near the origin; see \eqref{T0}.
\end{description}

To begin with, we construct a conformal mapping
\begin{equation}\label{varphi4}
\varphi_4(z)=\pm\frac{1}{8}iz\left(z^2-\frac{8}{3}\right)^{\frac{3}{2}},\quad \pm\Im z>0,
\end{equation}
which behaves like
\begin{equation}\label{varphi4beha}
\varphi_4(z)=2^{\frac{3}{2}}3^{-\frac{3}{2}}z(1+o(1)),\quad \mathrm{as} \quad z\rightarrow0.
\end{equation}
The solution of the above RH problem can be built explicitly in term of the modified Bessel functions. To see this, let $\Phi^{(\mathrm{Bes})}$ be the Bessel model parametrix with parameter $\alpha$ as given in  Appendix \ref{Bessel}. The parametrix $\mathbf{P}^{(0)}(z)$ is then constructed as follows:
\begin{equation}\label{P0}
\mathbf{P}^{(0)}(z)=\mathbf{E}^{(0)}(z)\Phi^{(\mathrm{Bes})}
\left(x^2\varphi_4(z)\right)\mathbf{K}(z)
\left[s_1(e^{-2\pi i\alpha}+s_*)\right]^{-\frac{\sigma_3}{2}}e^{-x^2g(z)\sigma_3},
\end{equation}
where $\mathbf{K}(z)$ is a piecewise constant matrix defined in regions $\Lambda_k$ described in Figure \ref{Bes}
\begin{equation}\label{K(z)}
\mathbf{K}(z)=\left\{\begin{aligned}
&\mathbf{I}, \quad &z&\in \Lambda_1\cup\Lambda_8,\\
&\begin{pmatrix}1  & 0\\ -e^{-2\pi i\alpha} &1\end{pmatrix},\quad &z&\in \Lambda_2,\\
&\begin{pmatrix}e^{-\pi i\alpha}  & 0 \\ -e^{\pi i\alpha} & e^{\pi i\alpha}\end{pmatrix},\quad &z&\in \Lambda_3,\\
&e^{-\pi i\alpha\sigma_3},\quad &z&\in \Lambda_4,\\
&e^{\pi i\alpha\sigma_3},\quad &z&\in \Lambda_5,\\
&\begin{pmatrix}e^{\pi i\alpha}  & 0 \\ e^{-\pi i\alpha} & e^{-\pi i\alpha}\end{pmatrix},\quad &z&\in\Lambda_6,\\
&\begin{pmatrix}1  & 0\\ e^{2\pi i\alpha} & 1\end{pmatrix},\quad &z&\in \Lambda_7,
\end{aligned}\right.
\end{equation}
and $\mathbf{E}^{(0)}(z)$ is given by
\begin{equation}\label{E(0)}
\mathbf{E}^{(0)}(z)=\mathbf{P}^{(\infty)}(z)\left[s_1(e^{-2\pi i\alpha}+s_*)\right]^{\frac{\sigma_3}{2}}\mathbf{Q}(z)e^{-\frac{1}{4}\pi i\sigma_3}\frac{1}{\sqrt{2}}\begin{pmatrix}1 & i\\ i & 1\end{pmatrix}
\end{equation}
with
\begin{equation}\label{Q}
\mathbf{Q}(z)=\left\{\begin{aligned}
&e^{\frac{1}{2}\pi i\alpha\sigma_3}, \quad &\Im z>0,~~ |z|<\delta,\\
&e^{-\frac{1}{2}\pi i\alpha\sigma_3}\sigma_3\sigma_1 ,\quad &\Im z<0,~~ |z|<\delta.
\end{aligned}\right.
\end{equation}

First, using the jump conditions \eqref{Besseljump},  it can be easily checked that the function $\mathbf{P}^{(0)}(z)$ constructed by \eqref{P0}-\eqref{Q} satisfies the same jump relations as $\mathbf{T}(z)$ on $\Sigma_\mathbf{T}\cap U(0,\delta) $. Applying \eqref{Pinftyjump}, \eqref{E(0)} and \eqref{Q}, one can check that  $\mathbf{E}^{(0)}(z)$ is also analytic in the neighborhood $U(0,\delta)$.   Moreover, the matching condition \eqref{P0matching} follows from \eqref{P0}, \eqref{E(0)} and \eqref{BesInfty}. Next, we verify the behavior of $\mathbf{P}^{(0)}(z)$ as $z\to 0$.
Recalling the definition of the connection matrix $E_0$ as given in \eqref{E0}, we may rewrite the asymptotic behavior \eqref{T0} in the form
\begin{equation}\label{T1}
\mathbf{T}(z)=\mathbf{T}_1(z)z^{\alpha\sigma_{3}}  \begin{pmatrix}1  & \frac{1}{1+e^{-2\pi i\alpha}}\\ 0&1\end{pmatrix} \begin{pmatrix}1  & 0\\ -e^{-2\pi i\alpha}&1\end{pmatrix}
\left[s_1(e^{-2\pi i\alpha}+s_*)\right]^{-\frac{\sigma_3}{2}}e^{-x^2g(z)\sigma_3},
\end{equation}
where $\mathbf{T}_1(z)$ is analytic in a neighborhood of $z=0$.
Comparing \eqref{T1} with \eqref{P0}  and \eqref{BesParaExpand}, one  can see that $\mathbf{P}^{(0)}(z)$  satisfies the asymptotic behavior \eqref{T0} as $z\to0$.

\subsection{Final transformation}
The final transformation is defined as
\begin{equation}\label{R}
\mathbf{R}(z)=\left\{\begin{aligned}
&\mathbf{T}(z)\left[\mathbf{P}^{(1,\pm)}(z)\right]^{-1},\quad &z&\in U(z_{1,\pm},\delta)\setminus\Sigma_\mathbf{T},\\
&\mathbf{T}(z)\left[\mathbf{P}^{(2,\pm)}(z)\right]^{-1},\quad &z&\in U(z_{2,\pm},\delta)\setminus\Sigma_\mathbf{T},\\
&\mathbf{T}(z)\left[\mathbf{P}^{(0)}(z)\right]^{-1},\quad &z&\in U(0,\delta)\setminus\Sigma_\mathbf{T},\\
&\mathbf{T}(z)\left[\mathbf{P}^{(\infty)}(z)\right]^{-1},\quad  &\mathrm{e}&\mathrm{lsewhere}.\\ \end{aligned}
\right.
\end{equation}
Immediately, $\mathbf{R}(z)$ fulfills the following RH problem.
\subsection*{RH problem for $\mathbf{R}(z)$}

\begin{figure}[t]
  \centering
  \includegraphics[width=13cm,height=8cm]{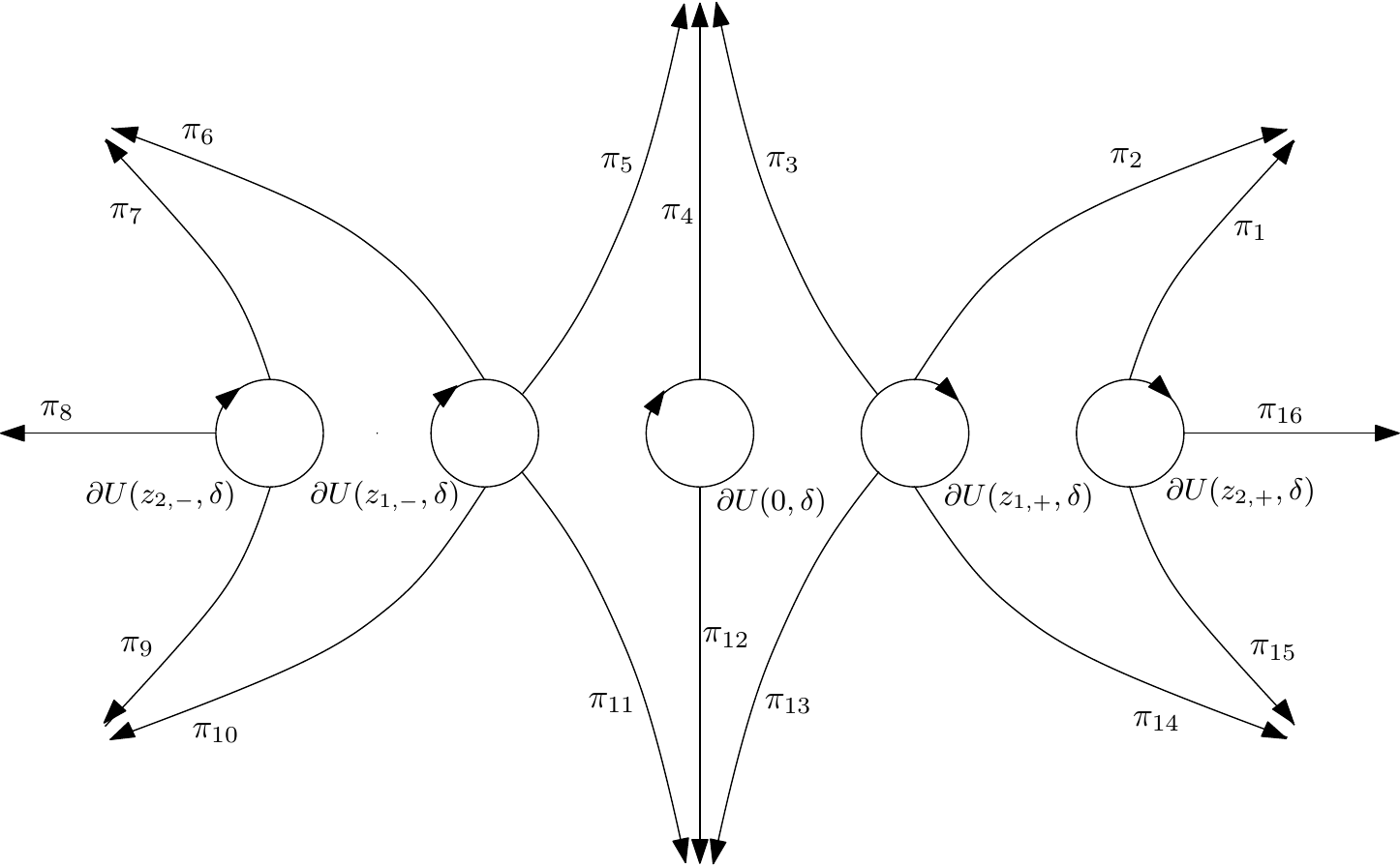}\\
  \caption{The jump contour $\Sigma_\mathbf{R}$}\label{Rjump}
\end{figure}

\begin{description}
\item{(1)} $\mathbf{R}(z)$ is analytic for $z\in \mathbb{C}\setminus\Sigma_\mathbf{R}$, where the contour $\Sigma_\mathbf{R}$ is illustrated in Figure \ref{Rjump}.
\item{(2)} On the contour $\Sigma_\mathbf{R}$, we have $\mathbf{R}_+(z)=\mathbf{R}_-(z)J_{\mathbf{R}}(z)$, where
 \begin{equation}\label{JumpR}
 J_{\mathbf{R}}(z)=\left\{\begin{aligned}
&\mathbf{P}^{(1,\pm)}(z)\mathbf{P}^{(\infty)}(z)^{-1},\quad &z&\in \partial U(z_{1,\pm},\delta),\\
&\mathbf{P}^{(2,\pm)}(z)\mathbf{P}^{(\infty)}(z)^{-1},\quad &z&\in \partial U(z_{2,\pm},\delta),\\
&\mathbf{P}^{(0)}(z)\mathbf{P}^{(\infty)}(z)^{-1},\quad &z&\in \partial U(0,\delta),\\
&\mathbf{P}^{(\infty)}_-(z)J_\mathbf{T}(z)\mathbf{P}^{(\infty)}_+(z)^{-1},\quad &\mathrm{e}&\mathrm{lsewhere}.\\ \end{aligned}
\right.
 \end{equation}
\item{(3)} As $z\rightarrow\infty$, we have
\begin{equation}\label{Rexpan}
\mathbf{R}(z)=\mathbf{I}+\frac{\mathbf{R}_1}{z}+\frac{\mathbf{R}_2}{z^2}+O(z^{-3}).
\end{equation}
\end{description}

In view of the matching conditions \eqref{P2+matching}, \eqref{P2-matching}, \eqref{P1+matching},  \eqref{P1-matching} and \eqref{P0matching}, it is now readily seen that as $x\rightarrow -\infty$,
\begin{equation}\label{JRestimation}
J_\mathbf{R}(z)=\left\{\begin{aligned}
&\mathbf{I}+O(|x|^{-1}), &z&\in\partial U(z_{1,\pm},\delta),\\
&\mathbf{I}+O(|x|^{-2}), &z&\in\partial U(0,\delta)\cup\partial U(z_{2,\pm},\delta),\\
&\mathbf{I}+O(e^{-c_1|x|^2}),&\mathrm{e}&\mathrm{lsewhere},
\end{aligned}\right.
\end{equation}
where $c_1$ is some positive constant. We thus have
\begin{equation}\label{Restimation}
\mathbf{R}(z)=\mathbf{I}+O(|x|^{-1})\quad \mathrm{as} \quad x\rightarrow -\infty,
\end{equation}
where the error bound is uniform for $z\in \mathbb{C}\setminus\Sigma_{\mathbf{R}}$.

\subsection{Case $|s_*|=0$}\label{sec:szero}
In this section, we consider  the special case $|s_*|=0$. It follows from the definitions \eqref{sstar}, \eqref{sstarbar} of $s_*$, $\bar s_*$  that
\begin{equation*}
\left\{\begin{aligned}
&1+s_0s_1=0,\\
&1+s_0s_1e^{2\pi i\alpha}=0.
\end{aligned}\right.
\end{equation*}
Substituting the first equation into the second one gives $e^{2\pi i\alpha}=1$. Therefore, $\alpha$ is an integer. Recalling the jump matrices on the contour $\Sigma_\mathbf{T}$ (see Figure \ref{Deforma3}), it is readily seen that in this case, there are no jumps on the infinite branches $\pi_k$, $k=2,\cdots,6,10,\cdots,14$. As a result,  the RH problem for $\mathbf{T}(z)$ given in Section \ref{sec:RHT} is reduced to the following RH problem for $\widehat{\mathbf{T}}(z)$.
\subsection*{RH problem for $\widehat{\mathbf{T}}(z)$}
\begin{description}
  \item{(1)} $\widehat{\mathbf{T}}(z)$ is analytic for $z\in \mathbb{C}\setminus (\Sigma_{\mathbf{T}}\setminus\bigcup\pi_k)$, where $k=2,\cdots,6,10,\cdots,14$.
  \item{(2)} $\widehat{\mathbf{T}}(z)$ satisfies the jump relations
  \begin{equation*}
  \widehat{\mathbf{T}}_+(z)=\widehat{\mathbf{T}}_-(z)\left\{\begin{aligned}
  &\begin{pmatrix}1 & -s_0^{-1}e^{2x^2g(z)}\\ 0 & 1 \end{pmatrix},\quad &z&\in \pi_1,\pi_{15}\\
  &\begin{pmatrix}1 & s_0^{-1}e^{2x^2g(z)} \\ 0 & 1 \end{pmatrix},\quad &z&\in \pi_7,\pi_9,\\
  &\begin{pmatrix}1 & 0\\ -s_0e^{-2x^2g(z)} & 1 \end{pmatrix},\quad &z&\in \pi_8,\pi^-_{16},\\
  &\begin{pmatrix}0 & -s_0^{-1}\\ s_0 & 0 \end{pmatrix},\quad &z&\in (z_{2,-},z_{2,+}),
  \end{aligned}\right.
  \end{equation*}
  \item{(3)} $\widehat{\mathbf{T}}(z)=\left(\mathbf{I}+O(z^{-1})\right)z^{\alpha\sigma_3}$ as $z\rightarrow \infty$.
  \item{(4)} $\widehat{\mathbf{T}}(z)$ has the following behavior at the origin
\begin{equation}\label{That0}
 \widehat{\mathbf{T}}(z)=\widehat{\mathbf{T}}_0(z)z^{\alpha\sigma_3}
 E_0S_1e^{-x^2g(z)\sigma_3}, \qquad \arg z\in(\pi/4,\pi/2),
\end{equation}
where $\widehat{\mathbf{T}}_0(z)$ is analytic  in a neighborhood of the origin. The asymptotic behaviors of $\widehat{\mathbf{T}}(z)$ in other regions are determined by \eqref{That0} and the jump relations satisfied by $\widehat{\mathbf{T}}(z)$.
\end{description}


The following RH analysis is similar to what we have done in the case $0<|s_*|<1$.
Firstly, neglecting the exponentially small entries in the jump matrices for $\widehat{\mathbf{T}}(z)$, we arrive at an approximate RH problem with constant jump on $(z_{2,-},z_{2,+})$. The solution to this RH problem, namely the global paramatrix,  can be constructed as \eqref{Pinfty} by taking the parameter $\beta=0$ therein. Next, we build two local parametrices near $z_{2,\pm}=\pm\sqrt{8/3}$, satisfying exactly the same RH problems for
 $\mathbf{P}^{(2,\pm)}(z)$, of which the solutions have been constructed in Section \ref{sec:AiParametrix}. Finally, by defining $\mathbf{R}(z)$ as
  the ratio of the solution  of the  RH problem for $\widehat{\mathbf{T}}(z)$ and the parametrices
 similarly to \eqref{R}, we show that $\mathbf{R}(z)$ satisfies the asymptotic $\mathbf{R}(z)=\mathbf{I}+O(|x|^{-2})$ as $x\to-\infty$ and thus the parametrices are good approximations of $\widehat{\mathbf{T}}(z)$ as $x\to-\infty$.

\section{RH analysis as $x\to-\infty$ with $|s_*|>1$}\label{Asymptotic-infty2}
When  $|s_*|>1$ and as $x\to-\infty$, a  detailed RH analysis has been carried out by the current authors in \cite{XXZ} for $\alpha\in\mathbb{R}$ but $\alpha-\frac{1}{2}\notin\mathbb{Z}$. In this section, we will briefly review the RH analysis therein, and discuss   the exceptional  case $\frac{1}{2}-\alpha\in\mathbb{N}$.

In the case $|s_*|>1$, we have $1-|s_*|^2<0$. From the definition \eqref{nu} of ${\beta}$, it follows that
\begin{equation}\label{nu0}
{\beta}=-\frac{1}{2\pi i}\ln(1-|s_*|^2)=-\frac{1}{2\pi i}\ln(|s_*|^2-1)-\frac{1}{2}:={\beta}_0-\frac{1}{2},\quad {\beta}_0\in i\mathbb{R}.
\end{equation}
This, together with \eqref{G+} and \eqref{G-},  implies that
$\mathbf{P}^{(1,\pm)}(z)\mathbf{P}^{\infty}(z)^{-1}\nrightarrow \mathbf{I}$ as $x\rightarrow-\infty$. Therefore,  the matching conditions \eqref{P1+matching} and \eqref{P1-matching} are no longer fulfilled.

Similar problems have also occurred  in   deriving singular asymptotics for the PII transcendents \cite{BI,Hu}. In \cite{BI}, Bothner and Its develop a certain ``dressing'' technique to tackle this problem. However, we find their method does not apply to our case. Instead, our idea is to modify the global parametrix $\mathbf{P}^{(\infty)}(z)$, so as to satisfy the matching condition. Similar technique was first used in \cite{ZZ} to derive the uniform asymptotics of the Pollaczek polynomials and subsequently to derive the uniform asymptotics of a system of Szeg\H{o} class polynomials \cite{ZXZ}.


\subsection{Modified global parametrix }
We define
\begin{equation}\label{mPinfty}
\widetilde{\mathbf{P}}^{(\infty)}(z)=\mathbf{H}(z)\mathbf{P}^{(\infty)}(z),
\end{equation}
where $\mathbf{P}^{(\infty)}(z)$ is introduced in \eqref{Pinfty}. The meromorphic function $\mathbf{H}(z)$ with $\det{\mathbf{H}}(z)=1$ is given by
\begin{equation}\label{H}
\mathbf{H}(z)=\mathbf{I}+\frac{\widetilde{A}}{z-\sqrt{\frac{2}{3}}}
+\frac{\widetilde{B}}{z+\sqrt{\frac{2}{3}}}.
\end{equation}
The explicit expressions of $\widetilde{A}$ and $\widetilde{B}$  are
\begin{equation}\label{reptildeAB}
\widetilde{A}=\sqrt{\frac 2 3} \begin{pmatrix} \frac{l}{l+\sqrt 2}   &
e^{\frac{5\pi i} 6} \frac{ l\, s_0^{-1}f^{-2}_{\infty}}
{ \sqrt 2+l }   \\[.3cm]
 e^{ \frac{\pi i} 6} \frac{ l\, s_0 f^{ 2}_{\infty}}
{ \sqrt 2-l }   &  \frac{l}{  l-\sqrt{2} }
\end{pmatrix},\quad
\widetilde{B}=\sqrt{\frac 2 3} \begin{pmatrix}    -\frac{l}{l+\sqrt 2}  &
e^{\frac{5\pi i} 6} \frac{ l\, s_0^{-1}f^{-2}_{\infty}}
{ \sqrt 2+l }   \\[.3cm]
e^{ \frac{\pi i} 6} \frac{ l\, s_0 f^{ 2}_{\infty}}
{ \sqrt 2-l }  &  -\frac{l}{  l-\sqrt{2} }
\end{pmatrix},
\end{equation}
where $f_{\infty}$ is given by \eqref{finfty} and $l=l(x)$ is defined as
\begin{equation}\label{c}
l=-\frac{i\sqrt{6}e^{i\phi}}{2+e^{i\phi}},\quad \phi=-\frac{\sqrt{3}}{3}x^{2}+i{\beta}_0\ln \left(2 \sqrt{3} x^{2}\right)
+\frac{2\pi\alpha}{3}+\arg\Gamma\left({\beta}_0+ \frac{1}{2}\right)+\arg s_{*},
\end{equation}
with ${\beta}_0$ and $s_*$  given in \eqref{nu0} and \eqref{sstar}, respectively.

We shall assume in \eqref{reptildeAB} that $x$ does not take   the zeros of the functions $\sqrt{2}\pm l(x)$. The set of zeros consists of two sequences   $\{x_n\}$ and  $\{y_n\}$ for  $n\in \mathbb{N}$, defined respectively by the equations
\begin{equation}\label{pole1}
\frac{\sqrt{3}}{3}x_n^{2}-i{\beta}_0\ln \left(2 \sqrt{3} x_n^{2}\right)
-\frac{2\pi\alpha}{3}-\arg \Gamma\left({\beta}_0-\frac{1}{2}\right)-\arg s_{*}-\frac{2\pi}{3}-2n\pi=0,
\end{equation}
and
\begin{equation}\label{pole2}
\frac{\sqrt{3}}{3}y_n^{2}-i{\beta}_0\ln \left(2 \sqrt{3} y_n^{2}\right)
-\frac{2\pi\alpha}{3}-\arg \Gamma\left({\beta}_0+ \frac{1}{2}\right)-\arg s_{*}+\frac{2\pi}{3}-2n\pi=0.
\end{equation}
For  detailed derivations of $\widetilde{A}$ and $\widetilde{B}$, we refer to \cite[Section 3.3]{XXZ}.

\subsection{Modified local parametrices at $z_{1,\pm}$, $z_{2,\pm}$ }
We first consider the local parametrix at $z_{1,+}=\sqrt{\frac{2}{3}}$. Instead of \eqref{P1+}, we define
\begin{align*}
\widetilde{\mathbf{P}}^{(1,+)}(z)=\widetilde {\mathbf{E}}^{(1,+)}(z)\Phi^{\mathrm{(PC)}}\left(|x|\varphi_2(z)\right)
\left(\frac{s_*}{h_0}\right)^{\frac{\sigma_3}{2}}
\left\{\begin{aligned}&s_0^{-\frac{\sigma_3}{2}}(-\sigma_{1})
e^{-x^2g(z)\sigma_3},& \Im z>0,\\
&s_0^{\frac{\sigma_3}{2}}\sigma_3e^{2\pi i(\alpha+{\beta})\sigma_3}
e^{-x^2g(z)\sigma_3},& \Im z<0,
\end{aligned}\right.
\end{align*}
where $\Phi^{\mathrm{(PC)}}$ is the parabolic cylinder model parametrix, $\varphi_2(z)$ is the conformal mapping \eqref{varphi2},  ${\beta}$ is given by \eqref{nu} and $\widetilde{\mathbf{E}}^{(1,+)}(z)$ is defined as
\begin{equation*}
\widetilde{\mathbf{E}}^{(1,+)}(z)=\mathbf{H}(z)\mathbf{W}^{(+)}(z)
\left(\frac{s_*}{h_0}\right)^{-\frac{\sigma_3}{2}}
|x|^{{\beta}\sigma_3}e^{i\frac{\sqrt{3}}{6}x^2\sigma_{3}}
\begin{pmatrix}1 & 0 \\ -\frac{1}{|x|\varphi_2(z)} & 1\end{pmatrix}2^{-\frac{\sigma_3}{2}}
\begin{pmatrix}|x|\varphi_2(z) & 1 \\ 1 & 0\end{pmatrix}
\end{equation*}
with $\mathbf{H}(z)$ and $\mathbf{W}^{(+)}(z)$ given by \eqref{H} and \eqref{W+}, respectively.

It is straightforward to check that $\widetilde{\mathbf{P}}^{(1,+)}(z)$ solves the following RH problem.

\subsection*{RH problem for $\widetilde{\mathbf{P}}^{(1,+)}(z)$}
\begin{description}
\item{(1)} $\widetilde{\mathbf{P}}^{(1,+)}(z)$ is analytic for all $z\in U(z_{1,+},\delta)\setminus\Sigma_\mathbf{T}$.
\item{(2)} $\widetilde{\mathbf{P}}^{(1,+)}(z)$ fulfills the same jump relations as $\mathbf{T}(z)$ on $U(z_{1,+},\delta)\cap\Sigma_\mathbf{T}$.
\item{(3)} On the circle $\partial U(z_{1,+},\delta)$, we have the matching condition
\begin{equation}\label{P1+matchingm}
\widetilde{\mathbf{P}}^{(1,+)}(z)=\left(\mathbf{I}+O(|x|^{-2})\right)\widetilde {\mathbf{P}}^{(\infty)}(z),\quad \mathrm{as}\quad x\rightarrow -\infty.
\end{equation}
\end{description}

Similarly, the local parametrix at $z_{1,-}=-\sqrt{\frac{2}{3}}$ is constructed by
\begin{align*}
\widetilde{\mathbf{P}}^{(1,-)}(z)=\widetilde {\mathbf{E}}^{(1,-)}(z)\Phi^{\mathrm{(PC)}}\left(|x|\varphi_3(z)\right)
\left(\frac{s_*}{h_0}\right)^{\frac{\sigma_3}{2}}\left\{\begin{aligned}&e^{2\pi i{\beta}\sigma_3}\left(s_0e^{2\pi i\alpha}\right)^{\frac{\sigma_3}{2}}
e^{-x^2g(z)\sigma_3},& \Im z>0,\\
&\left(s_0e^{2\pi i\alpha}\right)^{-\frac{\sigma_3}{2}}
\sigma_{3}\sigma_{1}e^{-x^2g(z)\sigma_3},
& \Im z<0,\end{aligned}\right.
\end{align*}
where $\varphi_3(z)$ is the conformal mapping \eqref{varphi3} and $\widetilde{\mathbf{E}}^{(1,-)}(z)$ is given by
\begin{equation*}
\widetilde{\mathbf{E}}^{(1,-)}(z)=\mathbf{H}(z)\mathbf{W}^{(-)}(z)
\left(\frac{s_*}{h_0}\right)^{-\frac{\sigma_3}{2}}
|x|^{{\beta}\sigma_3}e^{\frac{i\sqrt{3}x^2}{6}
\sigma_{3}}\begin{pmatrix}1 & 0 \\ -\frac{1}{|x|\varphi_3(z)} & 1\end{pmatrix}2^{-\frac{\sigma_3}{2}}\begin{pmatrix}|x|\varphi_3(z) & 1 \\ 1 & 0\end{pmatrix}
\end{equation*}
with $\mathbf{H}(z)$ and $\mathbf{W}^{(-)}(z)$ given by \eqref{H} and \eqref{W-}, respectively.

Then, $\widetilde{\mathbf{P}}^{(1,-)}(z)$ satisfies the following RH problem.

\subsection*{RH problem for $\widetilde{\mathbf{P}}^{(1,-)}(z)$}
\begin{description}
\item{(1)} $\widetilde{\mathbf{P}}^{(1,-)}(z)$ is analytic for all $z\in U(z_{1,-},\delta)\setminus\Sigma_\mathbf{T}$.
\item{(2)} $\widetilde{\mathbf{P}}^{(1,-)}(z)$ satisfies the same jump relations as $\mathbf{T}(z)$ on $U(z_{1,-},\delta)\cap\Sigma_\mathbf{T}$.
\item{(3)} On the boundary $\partial U(z_{1,-},\delta)$, we have
\begin{equation}\label{P1-matchingm}
\widetilde{\mathbf{P}}^{(1,-)}(z)=\left(\mathbf{I}+O(|x|^{-2})\right)\widetilde {\mathbf{P}}^{(\infty)}(z),\quad \mathrm{as}\quad x\rightarrow -\infty.
\end{equation}
\end{description}

Since the function $\mathbf{H}(z)$ defined by \eqref{H} is analytic at $z_{2,\pm}=\pm\sqrt{8/3}$, in order to construct the local parametrices $\widetilde {\mathbf{P}}^{(2,\pm)}(z)$ at these points, we only need to replace $\mathbf{P}^{(\infty)}(z)$ by $\widetilde{\mathbf{P}}^{(\infty)}(z)$ in the constructions of local parametrices at these points in the case $0<|s_*|<1$; see \eqref{P2+} and \eqref{P2-}. Moreover, we have the following matching conditions on the circular boundaries:
\begin{align}\label{P2+matchingm}
\widetilde{\mathbf{P}}^{(2,+)}(z)=\left(\mathbf{I}+O(|x|^{-2})\right)\widetilde {\mathbf{P}}^{(\infty)}(z),\quad \mathrm{as} \quad x\to-\infty
\end{align}
and
\begin{equation}\label{P2-matchingm}
\widetilde{\mathbf{P}}^{(2,-)}(z)=\left(\mathbf{I}+O(|x|^{-2})\right)\widetilde {\mathbf{P}}^{(\infty)}(z),\quad \mathrm{as} \quad x\to-\infty.
\end{equation}

\subsection{Modified local parametrix near the origin}
Accordingly, the parametrix $\widetilde{\mathbf{P}}^{(0)}(z)$ is defined as
\begin{equation}\label{Ptilde0}
\widetilde{\mathbf{P}}^{(0)}(z)=\widetilde{\mathbf{E}}^{(0)}(z)\Phi^{(\mathrm{Bes})}
\left(x^2\varphi_4(z)\right)\mathbf{K}(z)
\left[s_1\left(e^{-2\pi i\alpha}+s_*\right)\right]^{-\frac{\sigma_3}{2}}e^{-x^2g(z)\sigma_3},
\end{equation}
where $\varphi_4(z)$ is the conformal mapping \eqref{varphi4}, $\mathbf{K}(z)$ is defined by \eqref{K(z)} and $\widetilde{\mathbf{E}}^{(0)}(z)$ is given by
\begin{equation}\label{Etilde(0)}
\widetilde{\mathbf{E}}^{(0)}(z)=\widetilde{\mathbf{P}}^{(\infty)}(z)
\left[s_1(e^{-2\pi i\alpha}+s_*)\right]^{\frac{\sigma_3}{2}}\mathbf{Q}(z)e^{-\frac{1}{4}\pi i\sigma_3}\frac{1}{\sqrt{2}}\begin{pmatrix}1 & i\\ i & 1\end{pmatrix}
\end{equation}
with $\mathbf{Q}(z)$ defined by \eqref{Q}.

It is direct to verify that $\widetilde{\mathbf{P}}^{(0)}(z)$ satisfies the same jump conditions as $\mathbf{T}(z)$ and $\widetilde{\mathbf{P}}^{(0)}(z)$ possesses the same asymptotic behavior  near the origin as $\mathbf{T}(z)$ for the case $\alpha-\frac{1}{2}\notin\mathbb{Z}$. Moreover, $\widetilde{\mathbf{P}}^{(0)}(z)$ fulfills the following matching condition
\begin{equation}\label{P0matchingm}
\widetilde{\mathbf{P}}^{(0)}(z)=\left(\mathbf{I}+O(|x|^{-2})\right)
\widetilde{\mathbf{P}}^{(\infty)}(z)\quad \mathrm{as} \quad x\to-\infty.
\end{equation}

While, in the case $\frac{1}{2}-\alpha\in\mathbb{N}$ and  $|s_*|>1$, we need to check that $\widetilde{\mathbf{P}}^{(0)}(z)$ shares the same asymptotic behaviors as $\mathbf{T}(z)$ near the origin. From \eqref{E0} and \eqref{T(z)}, it follows that as $z\to0$ with $\arg z\in (0,\frac{\pi}{2})$
\begin{equation}\label{Ttilde0}
\mathbf{T}(z)=\widetilde{\mathbf{T}}_0(z)z^{\alpha\sigma_{3}}S_1 e^{-x^2g(z)\sigma_{3}}
\begin{pmatrix}1 & 0\\ -\frac{\overline{s}_*e^{-2\pi i\alpha}e^{-2x^2g(z)}}
  {s_1(e^{-2\pi i\alpha}+s_*)} & 1 \end{pmatrix},
\end{equation}
where $\widetilde{\mathbf{T}}_0(z)$ is analytic in a neighborhood of $z=0$. Using the following factorization
\begin{equation}
\begin{pmatrix}1 & s_1 \\ 0 & 1 \end{pmatrix}=\begin{pmatrix}1 & 0 \\ s_1^{-1} & 1 \end{pmatrix}\begin{pmatrix}0 & s_1 \\ -s_1^{-1} & 0 \end{pmatrix}\begin{pmatrix}1 & 0 \\ s_1^{-1} & 1 \end{pmatrix},
\end{equation}
we can rewrite \eqref{Ttilde0} as
\begin{equation}\label{Ttilde1}
\mathbf{T}(z)=\widetilde{\mathbf{T}}_1(z)z^{-\alpha\sigma_{3}}\begin{pmatrix}1 & 0\\ 1 & 1 \end{pmatrix}\left[s_1(e^{-2\pi i\alpha}+s_*)\right]^{-\frac{\sigma_3}{2}}
e^{-x^2g(z)\sigma_3},
\end{equation}
where $\widetilde{\mathbf{T}}_1(z)$ is analytic in a neighborhood of $z=0$.  Comparing \eqref{Ttilde1} with \eqref{Ptilde0}  and \eqref{BesParaExpand1}, it is readily seen that $\widetilde{\mathbf{P}}^{(0)}(z)$  satisfies the asymptotic behavior \eqref{Ttilde0} as $z\to0$ for $\frac{1}{2}-\alpha\in\mathbb{N}$.

\subsection{Final transformation}
The final transformation is defined as
\begin{equation}\label{Rtilde}
\widetilde{\mathbf{R}}(z)=\left\{\begin{aligned}
&\mathbf{T}(z)\left[\widetilde{\mathbf{P}}^{(1,\pm)}(z)\right]^{-1},\quad &z&\in U(z_{1,\pm},\delta)\setminus\Sigma_\mathbf{T},\\
&\mathbf{T}(z)\left[\widetilde{\mathbf{P}}^{(2,\pm)}(z)\right]^{-1},\quad &z&\in U(z_{2,\pm},\delta)\setminus\Sigma_\mathbf{T},\\
&\mathbf{T}(z)\left[\widetilde{\mathbf{P}}^{(0)}(z)\right]^{-1},\quad &z&\in U(0,\delta)\setminus\Sigma_\mathbf{T},\\
&\mathbf{T}(z)\left[\widetilde{\mathbf{P}}^{(\infty)}(z)\right]^{-1},\quad  &\mathrm{e}&\mathrm{lsewhere}.\\ \end{aligned}
\right.
\end{equation}
As a consequence, $\widetilde{\mathbf{R}}(z)$ solves the following RH problem.
\subsection*{RH problem for $\widetilde{\mathbf{R}}(z)$}

\begin{description}
\item{(1)} $\widetilde{\mathbf{R}}(z)$ is analytic for $z\in \mathbb{C}\setminus\Sigma_{\mathbf{R}}$, where $\Sigma_\mathbf{R}$ is illustrated in Figure \ref{Rjump}.
\item{(2)} On $\Sigma_{R}$, we have $\widetilde{\mathbf{R}}_+(z)=\widetilde{\mathbf{R}}_-(z)J_{\widetilde{\mathbf{R}}}(z)$, where
 \begin{equation}\label{Rtildejump}
 J_{\widetilde{\mathbf{R}}}(z)=\left\{\begin{aligned}
&\widetilde{\mathbf{P}}^{(1,\pm)}(z)\widetilde{\mathbf{P}}^{(\infty)}(z)^{-1},\quad &z&\in \partial U(z_{1,\pm},\delta),\\
&\widetilde{\mathbf{P}}^{(2,\pm)}(z)\widetilde{\mathbf{P}}^{(\infty)}(z)^{-1},\quad &z&\in \partial U(z_{2,\pm},\delta),\\
&\widetilde{\mathbf{P}}^{(0)}(z)\widetilde{\mathbf{P}}^{(\infty)}(z)^{-1},\quad &z&\in \partial U(0,\delta),\\
&\widetilde{\mathbf{P}}^{(\infty)}_-(z)J_{\mathbf{T}}(z)
\widetilde{\mathbf{P}}^{(\infty)}_+(z)^{-1},\quad &\mathrm{e}&\mathrm{lsewhere}.\\ \end{aligned}
\right.
 \end{equation}
\item{(3)} As $z\rightarrow\infty$, we have
\begin{equation}\label{Rtildeexp}
\widetilde{\mathbf{R}}(z)=\mathbf{I}+\frac{\widetilde{\mathbf{R}}_1}{z}+\frac{\widetilde {\mathbf{R}}_2}{z^2}+O(z^{-3}).
\end{equation}
\end{description}

Using the matching conditions  \eqref{P1+matchingm}, \eqref{P1-matchingm}, \eqref{P2+matchingm}, \eqref{P2-matchingm} and \eqref{P0matchingm},
we have the following estimates
for   the jump matrices \eqref{Rtildejump}    as $x\rightarrow-\infty$
\begin{equation}\label{Rtildejumpestima}
J_{\widetilde{\mathbf{R}}}(z)=\left\{\begin{aligned}
&\mathbf{I}+O(|x|^{-2}), &z&\in\partial U(0,\delta)\cup\partial U(z_{1,\pm},\delta)\cup\partial U(z_{2,\pm},\delta),\\
&\mathbf{I}+O(e^{-c_{2}|x|^2}),&z&\in\pi_k,\ k=1,\cdots,16.
\end{aligned}\right.
\end{equation}
where $c_2$ is some positive constant. Consequently, we have that for any $z\in \mathbb{C}\setminus\Sigma_\mathbf{R}$
\begin{equation}\label{Rtildeexpx}
\widetilde{\mathbf{R}}(z)=\mathbf{I}+O(|x|^{-2}),\qquad \mathrm{as}\qquad x\rightarrow-\infty.
\end{equation}

\section{RH analysis as $x\to-\infty$ with $|s_*|=1$ but $s_*\neq1$}\label{Asymptotic-infty3}
In the case when $|s_*|=1$ with $s_*\neq1$, it follows from \eqref{sstar} and \eqref{sstarbar} that
\begin{equation}\label{sstarvalue}
s_*=-e^{-2\pi i\alpha}.
\end{equation}
We need another  $g$-function
\begin{equation}\label{gbar}
\widehat{g}(z)=\frac{1}{8}(z^2-2)^2=\frac{1}{8}z^4-\frac{1}{2}z^2+\frac{1}{2}.
\end{equation}
It is direct to see that $\widehat{g}(z)$ has three saddle points $z_0=0$, $z_{\pm}=\pm\sqrt{2}$.
As shown   in Figure \ref{Ubarjump}, the topology of the anti-Stokes curves of $\widehat {g}(z)$ is quite different from $g(z)$ depicted in Figure \ref{ASC}.

As before, we begin with  the re-scaling transformation \eqref{rescaling}. To proceed, we introduce
\begin{equation}\label{Ubar}
\widehat{\mathbf{U}}(z)=e^{\frac{x^2}{2}\sigma_3}|x|^{\alpha\sigma_3}\Phi(z)
|x|^{-\alpha\sigma_3}e^{-x^2\widehat{g}(z)\sigma_3}.
\end{equation}
Then, $\widehat{\mathbf{U}}(z)$ satisfies a RH problem with jumps along $\gamma_k$, $k=1,3,4,5,7,8$ as shown in Figure \ref{PIVj}.
Subsequently, we  deform the jump contours to the anti-Stokes curves of $\widehat {g}(z)$ illustrated in Figure \ref{Ubarjump}. We rewrite the RH problem as the following problem formulated on the anti-Stokes curves of $\widehat {g}(z)$.

\begin{figure}[H]
  \centering
  \includegraphics[width=13cm,height=6cm]{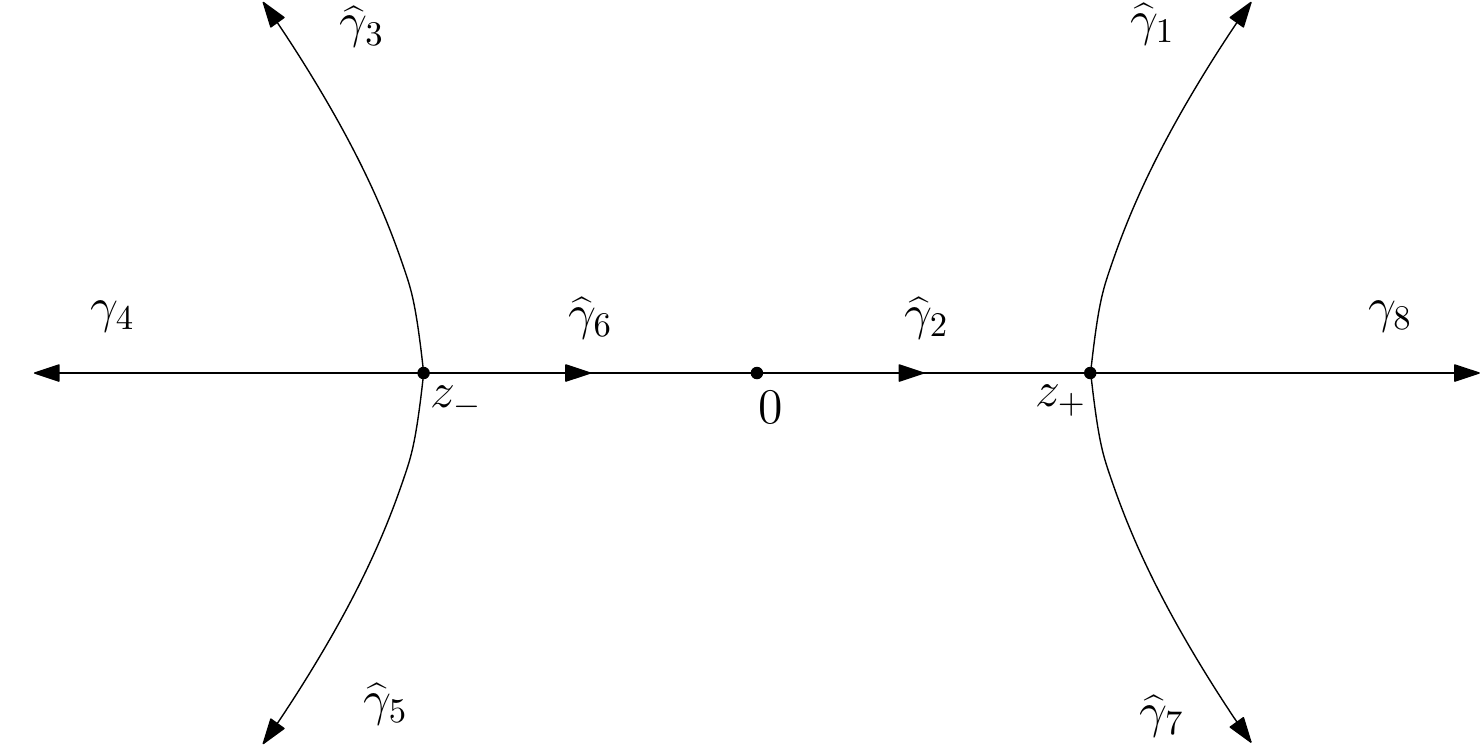}\\
  \caption{The jump contour $\Sigma_{\widehat{\mathbf{U}}}$}\label{Ubarjump}
\end{figure}

\subsection*{RH problem for $\widehat{\mathbf{U}}(z)$}
\begin{description}
\item{(1)} $\widehat{\mathbf{U}}(z)$ is analytic for $z\in\mathbb{C}\setminus \Sigma_{\widehat{\mathbf{U}}}$, with $\Sigma_{\widehat{\mathbf{U}}}$ is shown in Figure \ref{Ubarjump}. Note that we have deformed $\gamma_k$ to the anti-Stokes curves of $\widehat{g}(z)$, namely $\widehat{\gamma}_{k}$, $k=1,3,5,7$.

\item{(2)} On the contour $\Sigma_{\widehat{\mathbf{U}}}$, we have $\widehat{\mathbf{U}}_{+}(z)=\widehat{\mathbf{U}}_{-}(z)J_{\widehat{\mathbf{U}}}(z)$, where
$$J_{\widehat{\mathbf{U}}}(z)=\left\{
\begin{aligned}
&\begin{pmatrix}1 & s_k|x|^{2\alpha}e^{2x^2\widehat{g}(z)}\\ 0 & 1 \end{pmatrix},\quad &z&\in\widehat{\gamma}_{k},\ k=1,3,5,7,\\
&\begin{pmatrix}-1 & 0 \\ s_0e^{2\pi i\alpha}|x|^{-2\alpha}e^{-2x^2\widehat{g}(z)} & -1 \end{pmatrix},\quad &z&\in\widehat{\gamma}_{2},\\
&\begin{pmatrix}-e^{-2\pi i\alpha} & 0 \\ s_0e^{2\pi i\alpha}|x|^{-2\alpha}e^{-2x^2\widehat{g}(z)} & -e^{2\pi i\alpha}\end{pmatrix},\quad &z&\in\widehat{\gamma}_{6},\\
&\begin{pmatrix}1 & 0\\ -s_0e^{2\pi i\alpha}|x|^{-2\alpha}e^{-2x^2\widehat{g}(z)} & 1 \end{pmatrix},\quad &z&\in\gamma_{4},\\
&\begin{pmatrix}e^{-2\pi i\alpha} & 0\\ s_0e^{2\pi i\alpha}|x|^{-2\alpha}e^{-2x^2\widehat{g}(z)} & e^{2\pi i\alpha} \end{pmatrix},\quad &z&\in\gamma_{8}.
\end{aligned}
\right.
$$
\item{(3)} $\widehat{\mathbf{U}}(z)=\left(\mathbf{I}+O(z^{-1})\right)z^{\alpha\sigma_3}$ as $z\rightarrow\infty$, where $\arg z\in(0,2\pi)$.
\item{(4)} $\widehat{\mathbf{U}}(z)$ possesses the following asymptotic behavior near the origin
\begin{equation}\label{Ubaratzero}
\widehat{\mathbf{U}}(z)=\widehat{\mathbf{U}}_{0}(z)z^{\alpha\sigma_3}E_0S_1
|x|^{-\alpha\sigma_3}e^{-x^2\widehat{g}(z)\sigma_3},\qquad \arg z\in(\pi/4,\pi/2),
\end{equation}
where $\widehat{\mathbf{U}}_{0}(z)$ is analytic  in the neighborhood of $z_0=0$. The asymptotic behaviors of $\widehat{\mathbf{U}}(z)$ in other regions are determined by \eqref{Ubaratzero} and the jump relations satisfied by $\widehat{\mathbf{U}}(z)$.
\end{description}

From the properties of $g$-function $\widehat{g}(z)$, it is readily seen that the jump matrices on $\gamma_4$ and the anti-Stokes  curves $\widehat{\gamma}_k$, $k=1,3,5,7$ approach to identity matrix   exponentially fast as $x\rightarrow-\infty$. Thus,  the task is to construct a global parametrix satisfying the remaining jump along $(z_-,+\infty)$ and three local parametrices in the neighborhoods of the saddle points $z_0$, $z_{\pm}$.

\subsection{Global parametrix }
We need to solve the following RH problem for a $2\times 2$ matrix-valued function $\widehat{\mathbf{P}}^{(\infty)}(z)$.
\subsection*{RH problem for $\widehat{\mathbf{P}}^{(\infty)}(z)$}
\begin{description}
  \item{(1)} $\widehat{\mathbf{P}}^{(\infty)}(z)$ is analytic for $z\in \mathbb{C}\setminus (z_{-},+\infty)$.

  \item{(2)} $\widehat{\mathbf{P}}^{(\infty)}(z)$ satisfies the following jump relations
  \begin{equation}\label{Phatjump}
  \widehat{\mathbf{P}}^{(\infty)}_{+}(z)=\widehat{\mathbf{P}}^{(\infty)}_{-}(z)
  \left\{\begin{aligned}
  &\begin{pmatrix}-e^{-2\pi i\alpha} & 0 \\ 0 & -e^{2\pi i\alpha}\end{pmatrix},\quad& z&\in(z_{-},0), \\
  &-\mathbf{I},\quad & z&\in(0,z_{+}),\\
  &\begin{pmatrix}e^{-2\pi i\alpha} & 0 \\ 0 & e^{2\pi i\alpha}\end{pmatrix},\quad & z&\in(z_{+},+\infty)\\
  \end{aligned}\right.
  \end{equation}
  \item{(3)} $\widehat{\mathbf{P}}^{(\infty)}(z)$ has at most singularities of order less than $|\alpha|+\frac{3}{2}$ at $z_{\pm}=\pm\sqrt{2}$.
  \item{(4)}  As $z\rightarrow \infty$, we have
  $\widehat{\mathbf{P}}^{(\infty)}(z)=\left(\mathbf{I}+O(z^{-1})\right)z^{\alpha\sigma_3}$.
\end{description}

A solution to the above RH problem is given by
\begin{equation}\label{Pbarinfty}
\widehat{\mathbf{P}}^{(\infty)}(z)=\widehat{\mathbf{H}}(z)z^{-\alpha\sigma_3}
\left(z+\sqrt{2}\right)^{(\alpha+\frac{1}{2})\sigma_3}
\left(z-\sqrt{2}\right)^{(\alpha-\frac{1}{2})\sigma_3},
\end{equation}
where the branches of the powers are chosen such that $\arg z\in(0,2\pi)$ and $\arg(z\pm\sqrt{2})\in(0,2\pi)$. The function $\widehat{\mathbf{H}}(z)$, similar to \eqref{H}, takes the rational form
\begin{equation}\label{H(z)bar}
\widehat{\mathbf{H}}(z)=\mathbf{I}+\frac{\widehat{A}}{z+\sqrt{2}}+\frac{\widehat {B}}{z-\sqrt{2}},\quad \det\widehat{\mathbf{H}}(z)=1,
\end{equation}
where constant matrices $\widehat{A}$ and $\widehat{B}$ are to be determined. We point out that the factor $\widehat{\mathbf{H}}(z)$ in \eqref{Pbarinfty} is introduced to meet the matching conditions \eqref{matchingPr} and \eqref{matchingPl} below.

\subsection{Local parametrices at $z_{\pm}=\pm\sqrt{2}$}
We   seek two functions $\widehat{\mathbf{P}}^{(r)}(z)$ and $\widehat{\mathbf{P}}^{(l)}(z)$, satisfying the same jumps as $\widehat{\mathbf{\mathbf{U}}}(z)$ respectively in the neighborhoods $U(z_{\pm},\delta)=\{z\in\mathbb{C}\mid|z-z_{\pm}|<\delta\}$ of the saddle points $z_{\pm}$, and matching with $\widehat{\mathbf{P}}^{(\infty)}(z)$ on the boundaries $\partial U(z_{\pm},\delta)=\{z\in \mathbb{C}\mid|z-z_{\pm}|=\delta\}$.

\subsection*{RH problem for $\widehat{\mathbf{P}}^{(r)}(z)$}
\begin{description}
\item{(1)} $\widehat{\mathbf{P}}^{(r)}(z)$ is analytic for $z\in U(z_{+},\delta)\setminus \Sigma_{\widehat{\mathbf{U}}}$.
\item{(2)} $\widehat{\mathbf{P}}^{(r)}(z)$ shares the same jumps as $\widehat{\mathbf{U}}(z)$ on $\Sigma_{\widehat{\mathbf{U}}}\cap U(z_{+},\delta)$.
\item{(3)} On the boundary $\partial U(z_+,\delta)$, we have
\begin{equation}\label{matchingPr}
\widehat{\mathbf{P}}^{(r)}(z)=\left(\mathbf{I}+O(|x|^{-2})\right)
\widehat{\mathbf{P}}^{(\infty)}(z)\quad \mathrm{as}\quad x\rightarrow-\infty.
\end{equation}
\end{description}

Firstly, we introduce a conformal mapping
\begin{equation}\label{varphi5}
\varphi_5(z)=2\widehat{g}(z)^{\frac{1}{2}}=2(z-z_{\pm})(1+o(1)), \quad z\rightarrow z_{\pm}.
\end{equation}
We will make use of the parabolic cylinder function to construct the solution. Let $\Phi^{(\mathrm{PC})}$ be the parabolic cylinder parametrix with parameter $\beta=\frac{1}{2}-\alpha$ as given in Appendix \ref{PCP}. Then, the parametrix $\widehat{\mathbf{P}}^{(r)}(z)$ is defined as
\begin{equation}\label{Pr}
\widehat{\mathbf{P}}^{(r)}(z)=\mathbf{E}^{(r)}(z)\Phi^{(\mathrm{PC})}(|x|\varphi_5(z))
\left(\frac{h_1^{(r)}}{s_1}\right)^{\frac{\sigma_3}{2}}
\mathbf{D}^{(r)}(z)\sigma_3|x|^{-\alpha\sigma_3}e^{-x^2\widehat{g}(z)\sigma_3},
\end{equation}
where $h_1^{(r)}=\frac{\sqrt{2\pi}e^{-(\alpha-\frac{1}{2})\pi i}}{\Gamma(\alpha-\frac{1}{2})}$ is the Stokes multiplier defined in \eqref{h0}, \begin{equation}\label{Dr}
\mathbf{D}^{(r)}(z)=\left\{\begin{aligned}
&e^{2\pi(\alpha+\frac{1}{4})\sigma_3},
&\arg z&\in(-\frac{\pi}{4},0),\\
&e^{\frac{1}{2}\pi i\sigma_3}, &\arg z&\in(0,\pi),\\
&e^{-\frac{1}{2}\pi i\sigma_3}, &\arg z&\in(\pi,\frac{7\pi}{4}).
\end{aligned}\right.
\end{equation}
and $\mathbf{E}^{(r)}(z)$ is given by
\begin{align}\label{Er}
\mathbf{E}^{(r)}(z)=\widehat{\mathbf{P}}^{(\infty)}(z)\sigma_3
&\left(\frac{s_1}{h_1^{(r)}}\right)^{\frac{\sigma_3}{2}}|x|^{\frac{\sigma_3}{2}}
e^{\mp\frac{1}{2}\pi i\sigma_3}\nonumber\\
&\times\varphi_5(z)^{(\frac{1}{2}-\alpha)\sigma_3}
\begin{pmatrix}1 & \frac{\alpha-\frac{1}{2}}{|x|\varphi_5(z)} \\ 0 & 1
\end{pmatrix}2^{-\frac{\sigma_3}{2}}\begin{pmatrix} |x|\varphi_5(z) & 1\\
1 & 0 \end{pmatrix},~~\pm\Im z>0
\end{align}
with $\arg\varphi_5(z)\in(0,2\pi)$.

It is readily seen from \eqref{Phatjump}, \eqref{Er} and \eqref{PCAsyatinfty} that $\mathbf{E}^{(r)}(z)$ is analytic in the deleted neighborhood $U(z_+,\delta)\setminus\{z_+\}$ and the matching condition \eqref{matchingPr} is satisfied. To guarantee that $\mathbf{E}^{(r)}(z)$ is also analytic at the isolated point $z_+=\sqrt{2}$, we find, by computing the Laurent expansion at $z_+=\sqrt{2}$ using \eqref{Pbarinfty}, \eqref{H(z)bar}, \eqref{varphi5}, \eqref{Dr} and \eqref{Er},  that the constant matrices $\widehat{A}$ and $\widehat{B}$ in \eqref{H(z)bar} should satisfy the following algebraic equation
\begin{equation}\label{AbarBbar1}
\left(I+\frac{\widehat{ A}}{2\sqrt{2}}\right)\begin{pmatrix} 0 & c_r \\ 0 & 0\end{pmatrix}=-\widehat{B},
\end{equation}
where the constant $c_r$ is given by
\begin{equation}\label{cr}
c_r=\frac{2s_1\Gamma(\alpha+\frac{1}{2})e^{i\pi(\alpha-\frac{1}{2})}}{\sqrt{\pi}}.
\end{equation}

\subsection*{RH problem for $\widehat{\mathbf{P}}^{(l)}(z)$}
\begin{description}
\item{(1)} $\widehat{\mathbf{P}}^{(l)}(z)$ is analytic for $z\in U(z_{-},\delta)\setminus \Sigma_{\widehat{\mathbf{U}}}$.
\item{(2)} $\widehat{\mathbf{P}}^{(l)}(z)$ satisfies the same jumps as $\widehat{\mathbf{U}}(z)$ on $\Sigma_{\widehat{\mathbf{U}}}\cap U(z_-,\delta)$.
\item{(3)} On the circle $\partial U(z_-,\delta)$, it holds that
\begin{equation}\label{matchingPl}
\widehat{\mathbf{P}}^{(l)}(z)=\left(\mathbf{I}+O(|x|^{-2})\right)
\widehat{\mathbf{P}}^{(\infty)}(z)\quad \mathrm{as}\quad x\rightarrow-\infty.
\end{equation}
\end{description}
Similarly, the solution $\widehat{\mathbf{P}}^{(l)}(z)$ can  also be built in terms of the parabolic cylinder function. We choose the parameter $\beta=-\frac{1}{2}-\alpha$ in Appendix \ref{PCP}. More precisely, we define
\begin{equation}\label{Pl}
\widehat{\mathbf{P}}^{(l)}(z)=\mathbf{E}^{(l)}(z)\Phi^{(\mathrm{PC})}(|x|\varphi_5(z))
\left(\frac{h_1^{(l)}}{s_1}\right)^{\frac{\sigma_3}{2}}
\mathbf{D}^{(l)}(z)\sigma_3|x|^{-\alpha\sigma_3}e^{-x^2\widehat{g}(z)\sigma_3},
\end{equation}
where $h_1^{(l)}=\sqrt{2\pi}e^{-(\alpha+\frac{1}{2})\pi i}/\Gamma(\alpha+\frac{1}{2})$ is the Stokes multiplier given in \eqref{h0},
\begin{equation}\label{Dl}
\mathbf{D}^{(l)}(z)=\left\{\begin{aligned}
&e^{2\pi(\alpha+\frac{1}{2})\sigma_3}, &\arg z&\in(-\frac{\pi}{4},0),\\
&\mathbf{I},
&\arg z&\in(0,\frac{7\pi}{4}).
\end{aligned}\right.
\end{equation}
and $\mathbf{E}^{(l)}(z)$ is given by
\begin{align}\label{El}
\mathbf{E}^{(l)}(z)=\widehat{\mathbf{P}}^{(\infty)}(z)\sigma_3
&\left(\frac{s_1}{h_1^{(l)}}\right)^{\frac{\sigma_3}{2}}
|x|^{-\frac{\sigma_3}{2}}
\varphi_5(z)^{-(\alpha+\frac{1}{2})\sigma_3}\begin{pmatrix}
1 & 0 \\ -\frac{1}{|x|\varphi_5(z)} & 1
\end{pmatrix}2^{-\frac{\sigma_3}{2}}\begin{pmatrix} |x|\varphi_5(z) & 1\\
1 & 0 \end{pmatrix}
\end{align}
with $\arg\varphi_5(z)\in(0,2\pi)$.

Using the jumps \eqref{Phatjump}, it is straightforward to check that $\mathbf{E}^{(l)}(z)$ is analytic in the deleted neighborhood $U(z_-,\delta)\setminus\{z_-\}$ and the matching condition \eqref{matchingPl} is also satisfied. To ensure that $\mathbf{E}^{(l)}(z)$ is also analytic  at the isolated point $z_-=-\sqrt{2}$, by calculating the Laurent expansion at $z_-$ using \eqref{Pbarinfty}, \eqref{H(z)bar}, \eqref{varphi5}, \eqref{El} and \eqref{Dl},  we obtain another algebraic equation
\begin{equation}\label{AB2}
\left(I-\frac{\widehat{B}}{2\sqrt{2}}\right)\begin{pmatrix} 0 & 0 \\ c_l & 0\end{pmatrix}=-\widehat{A},
\end{equation}
where the constant $c_l$ is given by
\begin{equation}\label{cl}
c_l=-\frac{4\sqrt{\pi}}
{s_1\Gamma(\alpha+\frac{1}{2})e^{i\pi(\alpha+\frac{1}{2})}}.
\end{equation}
Combining \eqref{AbarBbar1}, \eqref{AB2} with the fact that $c_rc_l=8$, we  can now derive explicit expressions of $\widehat{A}$ and $\widehat{B}$ as follows:
\begin{equation}\label{ABexpres}
\widehat{A}=\begin{pmatrix} -\sqrt{2} & 0\\ -\frac{1}{2}c_l & 0 \end{pmatrix},\quad
\widehat{B}=\begin{pmatrix} 0 & -\frac{1}{2}c_r\\ 0 & \sqrt{2} \end{pmatrix},
\end{equation}
where $c_r$ and $c_l$ are given by \eqref{cr} and \eqref{cl}, respectively.

After determining $\widehat{A}$ and $\widehat{B}$ as given in \eqref{ABexpres}, it is straightforward to verify that the determinant condition $\det\widehat{\mathbf{H}}(z)=1$ is also satisfied.

\subsection{Local parametrix near the origin}
Now, we   look for a function $\widehat{\mathbf{P}}^{(0)}(z)$ satisfying the same jumps as $\widehat{\mathbf{U}}(z)$ in the neighborhood $U(z_0,\delta)=\{z\in\mathbb{C}\mid|z-z_0|<\delta\}$ of $z_0=0$ with some constant $0<\delta<\sqrt{2}$, and matching with $\widehat{\mathbf{P}}^{(\infty)}(z)$ on the boundary $\partial U(z_0,\delta)=\{z\in \mathbb{C}\mid|z-z_0|=\delta\}$.
\subsection*{RH problem for $\widehat{\mathbf{P}}^{(0)}(z)$}
\begin{description}
  \item{(1)} $\widehat{\mathbf{P}}^{(0)}(z)$ is analytic for $z\in U(z_0,\delta)\setminus [z_{-},z_{+}]$.

  \item{(2)} $\widehat{\mathbf{P}}^{(0)}(z)$ satisfies the following jump relations
  \begin{equation*}
  \widehat{\mathbf{P}}^{(0)}_{+}(z)=\widehat{\mathbf{P}}^{(0)}_{-}(z)
  \left\{\begin{aligned}
  &\begin{pmatrix}-e^{-2\pi i\alpha} & 0 \\ s_0e^{2\pi i\alpha}|x|^{-2\alpha}e^{-2x^2\widehat{g}(z)} & -e^{2\pi i\alpha}\end{pmatrix},\quad& z&\in(-\delta,0), \\
  &\begin{pmatrix}-1 & 0 \\ s_0e^{2\pi i\alpha}|x|^{-2\alpha}e^{-2x^2\widehat{g}(z)} & -1\end{pmatrix},\quad & z&\in(0,\delta).
  \end{aligned}\right.
  \end{equation*}
  \item{(3)} On the boundary $\partial U(z_0,\delta)$, we have
  \begin{equation}\label{matchingcondition0}
   \widehat{\mathbf{P}}^{(0)}(z)=(\mathbf{I}+O(|x|^{-2\alpha}e^{-x^2}))\widehat {\mathbf{P}}^{(\infty)}(z)\quad \mathrm{as}\quad x\to-\infty.
  \end{equation}
  \item{(4)} $\widehat{\mathbf{P}}^{(0)}(z)$ has the same behavior  as $\widehat{\mathbf{U}}(z)$ near the origin; see \eqref{Ubaratzero}.
\end{description}

A solution to the above RH problem can be constructed explicitly as follows:
\begin{equation}\label{Pbar0}
\widehat{\mathbf{P}}^{(0)}(z)=\widehat{\mathbf{H}}(z)\begin{pmatrix}1 & 0\\ r(z) & 1 \end{pmatrix}
z^{-\alpha\sigma_3}
\left(z+\sqrt{2}\right)^{(\alpha+\frac{1}{2})\sigma_3}
\left(z-\sqrt{2}\right)^{(\alpha-\frac{1}{2})\sigma_3},
\end{equation}
where $\widehat{\mathbf{H}}(z)$ is given by \eqref{H(z)bar} and \eqref{ABexpres}. The function $r(z)$ is defined by
\begin{equation}\label{rexpre}
r(z)=\left\{\begin{aligned} &
-\frac{s_0|x|^{-2\alpha}e^{-2x^2\widehat{g}(z)}z^{2\alpha}}{(1+e^{-2\pi i\alpha})(z+\sqrt{2})^{2\alpha+1}(z-\sqrt{2})^{2\alpha-1}}, && |z|<\delta,~\Im z>0,\\
&
\frac{s_0|x|^{-2\alpha}e^{-2x^2\widehat{g}(z)}e^{2\pi i\alpha}z^{2\alpha}}{(1+e^{-2\pi i\alpha})(z+\sqrt{2})^{2\alpha+1}(z-\sqrt{2})^{2\alpha-1}}, && |z|<\delta,~\Im z<0,\end{aligned}\right.
\end{equation}
where the branches  are chosen such that $\arg z\in(0,2\pi)$ and $\arg(z\pm\sqrt{2})\in(0,2\pi)$.

\subsection{Final transformation}
The final transformation is now defined as
\begin{equation}\label{Rbar}
\widehat{\mathbf{R}}(z)=\left\{\begin{aligned}
&\widehat{\mathbf{U}}(z)\left[\widehat{\mathbf{P}}^{(l)}(z)\right]^{-1},\quad &z&\in U(z_-,\delta)\setminus\Sigma_{\widehat{\mathbf{U}}},\\
&\widehat{\mathbf{U}}(z)\left[\widehat{\mathbf{P}}^{(r)}(z)\right]^{-1},\quad &z&\in U(z_+,\delta)\setminus\Sigma_{\widehat{\mathbf{U}}},\\
&\widehat{\mathbf{U}}(z)\left[\widehat{\mathbf{P}}^{(0)}(z)\right]^{-1},\quad &z&\in U(z_0,\delta)\setminus\Sigma_{\widehat{\mathbf{U}}},\\
&\widehat{\mathbf{U}}(z)\left[\widehat{\mathbf{P}}^{(\infty)}(z)\right]^{-1},\quad  &\mathrm{e}&\mathrm{lsewhere}.\\ \end{aligned}
\right.
\end{equation}
Then, $\widehat{\mathbf{R}}(z)$ satisfies the following RH problem.

\subsection*{RH problem for $\widehat{\mathbf{R}}(z)$}
\begin{description}
\item{(1)} $\widehat{\mathbf{R}}(z)$ is analytic for $z\in \mathbb{C}\setminus\Sigma_{\widehat{\mathbf{R}}}$, where $\Sigma_{\widehat{\mathbf{R}}}$ is depicted in Figure \ref{Rbarjumppicture}.
\item{(2)} On the contour $\Sigma_{\widehat{\mathbf{R}}}$, we have $\widehat{\mathbf{R}}_+(z)=\widehat{\mathbf{R}}_-(z)J_{\widehat{\mathbf{R}}}(z)$, where
 \begin{equation}\label{Rbarjump}
 J_{\widehat{\mathbf{R}}}(z)=\left\{\begin{aligned}
&\widehat{\mathbf{P}}^{(l)}(z)\widehat{\mathbf{P}}^{(\infty)}(z)^{-1},\quad &z&\in \partial U(z_-,\delta),\\
&\widehat{\mathbf{P}}^{(r)}(z)\widehat{\mathbf{P}}^{(\infty)}(z)^{-1},\quad &z&\in \partial U(z_+,\delta),\\
&\widehat{\mathbf{P}}^{(0)}(z)\widehat{\mathbf{P}}^{(\infty)}(z)^{-1},\quad &z&\in \partial U(z_0,\delta)\\
&\widehat{\mathbf{P}}^{(\infty)}_-(z)J_{\widehat{\mathbf{U}}}(z)\widehat {\mathbf{P}}^{(\infty)}_+(z)^{-1},\quad &\mathrm{e}&\mathrm{lsewhere}.\end{aligned}
\right.
 \end{equation}
\item{(3)} As $z\rightarrow\infty$, $\widehat{\mathbf{R}}(z)$ admits the expansion
\begin{equation}\label{Rbarexp}
\widehat{\mathbf{R}}(z)=\mathbf{I}+\frac{\widehat{\mathbf{R}}_1}{z}
+\frac{\widehat{\mathbf{R}}_2}{z^2}+O(z^{-3}).
\end{equation}
\end{description}

\begin{figure}[t]
  \centering
  \includegraphics[width=13cm,height=7cm]{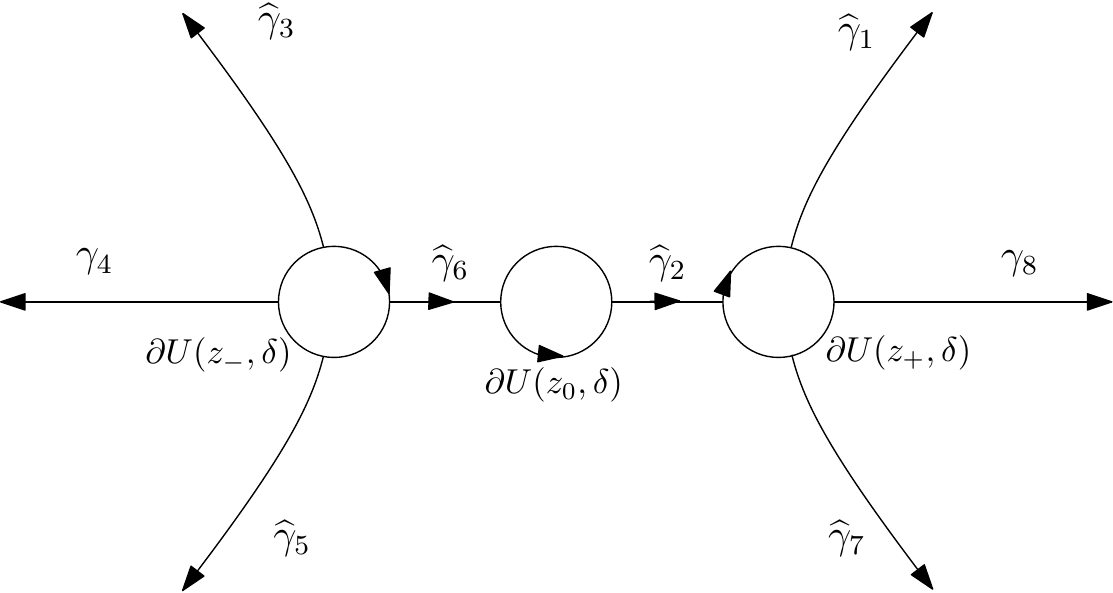}\\
  \caption{The jump contour $\Sigma_{\widehat{\mathbf{R}}}$}\label{Rbarjumppicture}
\end{figure}

By virtue of the matching conditions \eqref{matchingPr}, \eqref{matchingPl} and \eqref{matchingcondition0}, the jump matrix \eqref{Rbarjump} has the following estimations as $x\rightarrow-\infty$
\begin{equation}\label{Rbarjumpestima}
J_{\widehat{\mathbf{R}}}(z)=\left\{\begin{aligned}
&\mathbf{I}+O(|x|^{-2}), &z&\in\partial U(z_{\pm},\delta),\\
&\mathbf{I}+O(|x|^{-2\alpha}e^{-x^2}), &z&\in\partial U(z_0,\delta)\cup\widehat{\gamma}_2\cup\widehat{\gamma}_6,\\
&\mathbf{I}+O(e^{-c_{3}x^2}),&\mathrm{e}&\mathrm{lsewhere},
\end{aligned}\right.
\end{equation}
where $c_3$ is some positive constant. As a result, we get that for any $z\in\mathbb{C}\setminus\Sigma_{\widehat{\mathbf{R}}}$,
\begin{equation}\label{Rbarexpans}
\widehat{\mathbf{R}}(z)=\mathbf{I}+O(|x|^{-2}),\quad \mathrm{as}\quad x\rightarrow-\infty.
\end{equation}

\section{Proof of Theorems \ref{thm1} - \ref{thm3}}\label{proof1}
To derive  the asymptotics for the PIV solution $q(x;\alpha,\kappa)$ and the associated Hamiltonian $\mathcal{H}(x;\alpha,\kappa)$ as $x\rightarrow-\infty$, we use the identities \eqref{solu2} and \eqref{tau1} and the asymptotic analysis of the RH problem  for $\Psi(\xi,x)$
performed in Sections \ref{Asymptotic-infty1}-\ref{Asymptotic-infty3}.

\subsection{Derivations of \eqref{q1} and \eqref{H1}}\label{Derivation1}
Let us focus on the case $\kappa(\kappa-\kappa^*)<0$. Tracing back the series of invertible transformations $$\Phi\mapsto \mathbf{U}\mapsto \mathbf{T}\mapsto \mathbf{R}$$
as defined in \eqref{U(z)}, \eqref{T(z)} and \eqref{R}), respectively, we obtain that for large $z$
\begin{equation*}
e^{\frac{x^2}{3}\sigma_3}\Phi(z)e^{-x^2g(z)\sigma_3}
=\mathbf{R}(z)\mathbf{P}^{(\infty)}(z).
\end{equation*}
It follows from the asymptotic expansions \eqref{Asyatinfty1} and \eqref{gatinfty} that
\begin{equation}\label{RHanalysis}
e^{\frac{x^2}{3}\sigma_3}\Phi(z)e^{-x^2g(z)\sigma_3}z^{-\alpha\sigma_3}
=\mathbf{I}+\frac{e^{\frac{x^2}{3}\sigma_3}\Phi_1e^{-\frac{x^2}{3}\sigma_3}}{z}
+\frac{e^{\frac{x^2}{3}\sigma_3}\Phi_2e^{-\frac{x^2}{3}\sigma_3}-\frac{4x^2}{27}
\sigma_3}{z^2}+O(z^{-3}).
\end{equation}
To compute the coefficients $\Phi_1$ and $\Phi_2$, we need to know the asymptotics of $\mathbf{R}(z)$ and $\mathbf{P}^{(\infty)}(z)$ as $z\to\infty$. The large-$z$ asymptotics of $\mathbf{R}(z)$ has  already  been given in \eqref{Rexpan}. As for $\mathbf{P}^{(\infty)}(z)$, we obtain from \eqref{Pinfty}, \eqref{X(z)atinfty} and \eqref{fatinfty} that
\begin{equation}\label{Pinftyatinfty}
\mathbf{P}^{(\infty)}(z)z^{-\alpha\sigma_3}=\mathbf{I}+\frac{\mathbf{P}_1^{(\infty)}}{z}
+\frac{\mathbf{P}_2^{(\infty)}}{z^2}+O\left(\frac{1}{z^3}\right),~~~\mbox{as}~ ~ z\rightarrow\infty,
\end{equation}
where
\begin{equation}\label{P1P2}
\mathbf{P}_1^{(\infty)}=\sqrt{\frac{2}{3}}s_0^{-\frac{\sigma_3}{2}}f_{\infty}^{-\sigma_3}
\sigma_2f_{\infty}^{\sigma_3}s_0^{\frac{\sigma_3}{2}},\quad \mathbf{P}_2^{(\infty)}=\frac{\mathbf{I}}{3}
-\frac{2}{3}\left(\alpha+\sqrt{3}i{\beta}\right)\sigma_3.
\end{equation}
A combination of \eqref{Rexpan}, \eqref{RHanalysis} and \eqref{Pinftyatinfty} gives
\begin{equation}\label{Phi1}
\Phi_1=e^{-\frac{x^2}{3}\sigma_3}\left[\mathbf{R}_1+\mathbf{P}_1^{(\infty)}\right]
e^{\frac{x^2}{3}\sigma_3}
\end{equation}
and
\begin{equation}\label{Phi2}
\Phi_2=e^{-\frac{x^2}{3}\sigma_3}\left[\mathbf{R}_1\mathbf{P}_1^{(\infty)}
+\mathbf{R}_2+\mathbf{P}_2^{(\infty)}\right]e^{\frac{x^2}{3}\sigma_3}
+\frac{4x^2}{27}\sigma_3.
\end{equation}
Using formulas \eqref{solu2}, \eqref{tau1} and combining them with \eqref{Rexpan}, \eqref{RHanalysis}, \eqref{P1P2}, \eqref{Phi1} and \eqref{Phi2}, we get the following expressions for $q(x;\alpha,\kappa)$ and $\mathcal{H}(x;\alpha,\kappa)$
\begin{equation}\label{qexpre2}
q(x;\alpha,\kappa)=-\frac{2}{3}x-i\sqrt{\frac{2}{3}}x\Big[s_0f_{\infty}^2(\mathbf{R}_1)_{12}
-s_0^{-1}f_{\infty}^{-2}(\mathbf{R}_1)_{21}\Big]-x(\mathbf{R}_1)_{12}(\mathbf{R}_1)_{21}
\end{equation}
and
\begin{align}\label{Hexpres}
\mathcal{H}(x;\alpha,\kappa)=&-\frac{8}{27}x^3+\frac{4}{3}(\alpha
+\sqrt{3}i{\beta})x-x\sqrt{\frac{2}{3}}i\Big[s_0f_{\infty}^2(\mathbf{R}_1)_{12}
+ s_0^{-1}f_{\infty}^{-2}(\mathbf{R}_1)_{21}\Big]\nonumber \\
&-x\big[(\mathbf{R}_2)_{11}-(\mathbf{R}_2)_{22}\big].
\end{align}

Thus, the remaining task is to compute the asymptotics of $\mathbf{R}_1$ and $\mathbf{R}_2$. It is observed
from \eqref{JRestimation} and \eqref{Restimation} that, neglecting a uniform error term $O(|x|^{-2})$, the contribution to $\mathbf{R}(z)-\mathbf{I}$  comes from the jumps on $\partial U(z_{1,\pm},\delta)$. Therefore, it follows from \eqref{matchingcondition+}, \eqref{matchingcondition-} and \eqref{JumpR} that as $x\to-\infty$,
\begin{equation}\label{Rexpanrevisited}
\mathbf{R}(z)=\mathbf{I}+\frac{\mathbf{R}^{(1)}(z)}{|x|}+O\left(\frac{1}{|x|^2}\right),\quad \mathrm{for} \quad z\in\partial U(z_{1,\pm},\delta),
\end{equation}
where the coefficient $\mathbf{R}^{(1)}(z)=O(z^{-1})$ as $z\rightarrow\infty$, and satisfies the jump relation
\begin{equation}
\mathbf{R}^{(1)}_+(z)-\mathbf{R}^{(1)}_-(z)=\mathbf{G}(z),\quad z\in\partial U(z_{1,\pm},\delta),
\end{equation}
in which $\mathbf{G}(z)=\mathbf{G}^{(\pm)}(z)$ are given by \eqref{G+} and \eqref{G-}, respectively. Using the Sokhotskii-Plemelj formula and keeping in mind the clockwise orientations of the boundaries $\partial U(z_{1,\pm},\delta)$, the solution to the above RH problem is explicitly given by
\begin{equation}\label{R(1)(z)}
\mathbf{R}^{(1)}(z)=\left\{\begin{aligned}
&\frac{L_+}{z-\sqrt{\frac{2}{3}}}+\frac{L_-}{z+\sqrt{\frac{2}{3}}},&z&\notin \overline{U(z_{1,+},\delta)}\cup \overline{U(z_{1,-},\delta)},\\
&\frac{L_+}{z-\sqrt{\frac{2}{3}}}+\frac{L_-}{z+\sqrt{\frac{2}{3}}}-\mathbf{G}(z),
&z&\in U(z_{1,+},\delta)\cup U(z_{1,-},\delta),
\end{aligned}\right.
\end{equation}
where $L_{\pm}=\mathrm{Res}\,\left(\mathbf{G}^{(\pm)}(z),z=\pm\sqrt{\frac{2}{3}}\;\right)$ are given by
\begin{align}\label{res1}
L_+&=\mathbf{W}^{(+)}(\textstyle\sqrt{\frac{2}{3}})\begin{pmatrix}0 & -\frac{{\beta} h_0e^{i\frac{\sqrt{3}}{3}x^2}|x|^{2{\beta}}}
{2\cdot 3^{-\frac{1}{4}}e^{\frac{3\pi i}{4}}s_*}\\ -\frac{s_*|x|^{-2{\beta}}}
{2\cdot 3^{-\frac{1}{4}}e^{\frac{3\pi i}{4}}h_0e^{ i\frac{\sqrt{3}}{3}x^2}} & 0 \end{pmatrix}\mathbf{W}^{(+)}(\sqrt{\frac{2}{3}})^{-1},\\
L_-&=\mathbf{W}^{(-)}(-\textstyle\sqrt{\frac{2}{3}})\begin{pmatrix}0 & \frac{{\beta} h_0e^{i\frac{\sqrt{3}}{3}x^2}|x|^{2{\beta}}}
{2\cdot 3^{-\frac{1}{4}}e^{\frac{3\pi i}{4}}s_*}\\ \frac{s_*|x|^{-2{\beta}}}{2\cdot 3^{-\frac{1}{4}}e^{\frac{3\pi i}{4}}h_0e^{i\frac{\sqrt{3}}{3}x^2}} & 0 \end{pmatrix}\mathbf{W}^{(-)}(-\textstyle\sqrt{\frac{2}{3}})^{-1},\label{res2}
\end{align}
with $\mathbf{W}^{(\pm)}(\pm\sqrt{\frac{2}{3}})$ given by \eqref{W+atz+}) and \eqref{W-atz-}, respectively.

Using \eqref{R(1)(z)} to expand $\mathbf{R}^{(1)}(z)$ in \eqref{Rexpanrevisited}  into the Taylor series at infinity, we obtain the asymptotics for the coefficients $\mathbf{R}_1$ and $\mathbf{R}_2$ in the expansion \eqref{Rexpan}:
\begin{equation}\label{R1R2}
\mathbf{R}_1=\frac{\mathbf{R}_1^{(1)}}{|x|}+O(|x|^{-2}),\quad \mathbf{R}_2=\frac{\mathbf{R}_2^{(1)}}{|x|}+O(|x|^{-2}),\quad  \mathrm{as} \quad x\rightarrow-\infty.
\end{equation}
 The coefficients $\mathbf{R}_1^{(1)}=L_++L_-$ and $\mathbf{R}_2^{(1)}=\sqrt{\frac{2}{3}} \left(L_+-L_-\right)$, after a direct computation from \eqref{res1} and \eqref{res2}, are explicitly given by
\begin{equation}\label{R1(1)}
\mathbf{R}_1^{(1)}=\begin{pmatrix}
0 & \left(C_1-C_2+\frac{i}{\sqrt{3}}(C_1+C_2)\right)s_0^{-1}f_{\infty}^{-2}\\
\left(C_2-C_1+\frac{i}{\sqrt{3}}(C_1+C_2)\right)s_0f_{\infty}^{2} & 0
\end{pmatrix},
\end{equation}
and
\begin{equation}\label{R2(1)}
\mathbf{R}_2^{(1)}=\begin{pmatrix}
\frac{2\sqrt{2}}{3}(C_1+C_2) & 0\\
0 & -\frac{2\sqrt{2}}{3}(C_1+C_2)
\end{pmatrix},
\end{equation}
with
\begin{equation}\label{C1C2}
C_1=\frac{{\beta} h_0e^{\frac{i\sqrt{3}}{3}x^2}2^{{\beta}-1}3^{\frac{2{\beta}+1}{4}}}{s_*
e^{\frac{3\pi i}{4}+\frac{2\pi i\alpha}{3}+\frac{\pi i{\beta}}{2}}}|x|^{2{\beta}},\qquad
C_2=\frac{s_*e^{-\frac{3\pi i}{4}+\frac{2\pi i\alpha}{3}+\frac{\pi i{\beta}}{2}}}{h_0e^{\frac{i\sqrt{3}}{3}x^2}
2^{{\beta}+1}3^{\frac{2{\beta}-1}{4}}}|x|^{-2{\beta}}.
\end{equation}

Now, we are ready to derive the asymptotics for $q(x;\alpha,\kappa)$ and $\mathcal{H}(x;\alpha,\kappa)$. Substituting \eqref{R1R2}, \eqref{R1(1)}, \eqref{R2(1)} and \eqref{C1C2} into \eqref{qexpre2} and \eqref{Hexpres} yields
\begin{align}\label{qexpre3}
q(x;\alpha,\kappa)&=-\frac{2}{3}x+2^{\frac{1}{2}}3^{-\frac{1}{4}}e^{-\frac{\pi i}{4}}\Big[{\beta} h_0s_*^{-1}e^{\frac{\sqrt{3}x^2i}{3}-\frac{2\pi i\alpha}{3}-\frac{\pi i{\beta}}{2}+{\beta}\ln\left(2\sqrt{3}x^2\right)}\nonumber\\
&\qquad\qquad\ -h_0^{-1}s_*e^{-\frac{\sqrt{3}x^2i}{3}+\frac{2\pi i\alpha}{3}+\frac{\pi i{\beta}}{2}-{\beta}\ln\left(2\sqrt{3}x^2\right)}\Big]+O(|x|^{-1}),\quad \mathrm{as}\quad x\rightarrow-\infty,
\end{align}
and
\begin{align}\label{Hexpres1}
\mathcal{H}(x;\alpha,\kappa)&=-\frac{8}{27}x^3
+\frac{4}{3}\left(\alpha+\sqrt{3}i{\beta}\right)x
+2^{\frac{1}{2}}3^{-\frac{3}{4}}e^{-\frac{3\pi i}{4}}\Big[{\beta} h_0s_*^{-1}e^{\frac{\sqrt{3}x^2i}{3}-\frac{2\pi i\alpha}{3}-\frac{\pi i{\beta}}{2}+{\beta}\ln\left(2\sqrt{3}x^2\right)}\nonumber\\
&\qquad\qquad\ +h_0^{-1}s_*e^{-\frac{\sqrt{3}x^2i}{3}+\frac{2\pi i\alpha}{3}+\frac{\pi i{\beta}}{2}-{\beta}\ln\left(2\sqrt{3}x^2\right)}\Big]+O(|x|^{-1}),\quad \mathrm{as}\quad x\rightarrow-\infty.
\end{align}
Recalling the definition \eqref{h0} of the Stokes multiplier $h_0$ and using the reflection formula
\begin{equation}\label{refleformu}
\Gamma(1-{\beta})\Gamma({\beta})=\frac{\pi}{\sin(\pi{\beta})},
\end{equation}
we can rewrite \eqref{qexpre3} and \eqref{Hexpres1} in the following symmetric form
\begin{align}\label{qexpre4}
q(x;\alpha,\kappa)=&-\frac{2}{3}x+3^{-\frac{1}{4}}\pi^{-\frac{1}{2}}
e^{\frac{\pi i{\beta}}{2}}{\beta}|s_*|\Big[\Gamma(-{\beta})
e^{\frac{\sqrt{3}x^2i}{3}-\frac{2\pi i\alpha}{3}-\frac{\pi i}{4}-i\arg s_*+{\beta}\ln\left(2\sqrt{3}x^2\right)}\nonumber\\
&- \Gamma({\beta})e^{-\frac{\sqrt{3}x^2i}{3}+\frac{2\pi i\alpha}{3}+\frac{\pi i}{4}+i\arg s_*-{\beta}\ln\left(2\sqrt{3}x^2\right)}\Big]
+O(|x|^{-1}),\quad \mathrm{as}\quad x\rightarrow-\infty,
\end{align}
and
\begin{align}\label{Hexpres2}
\mathcal{H}(x;\alpha,\kappa)=&-\frac{8}{27}x^3
+\frac{4}{3}\left(\alpha+\sqrt{3}i{\beta}\right)x\nonumber\\
&-3^{-\frac{3}{4}}\pi^{-\frac{1}{2}}e^{\frac{\pi i{\beta}}{2}}{\beta}|s_*|\Big[\Gamma(-{\beta})
e^{\frac{\sqrt{3}x^2i}{3}-\frac{2\pi i\alpha}{3}+\frac{\pi i}{4}-i\arg s_*+{\beta}\ln\left(2\sqrt{3}x^2\right)}\nonumber\\
& - \Gamma({\beta})e^{-\frac{\sqrt{3}x^2i}{3}+\frac{2\pi i\alpha}{3}-\frac{\pi i}{4}+i\arg s_*-{\beta}\ln\left(2\sqrt{3}x^2\right)}\Big]
+O(|x|^{-1}),\quad \mathrm{as}\quad x\rightarrow-\infty.
\end{align}
Remember that ${\beta}$ in the case under consideration is purely imaginary (see \eqref{nuimag}), we thus have the following complex conjugate relation
\begin{equation*}
\overline{\Gamma({\beta})}=\Gamma(-{\beta}).
\end{equation*}
Denote by
\begin{equation}\label{a}
a^2=2i{\beta}=-\frac{1}{\pi}\ln(1-|s_*|^2),\quad  a>0.
\end{equation}
It is seen from \eqref{nu}, \eqref{refleformu} and \eqref{a} that
\begin{equation}\label{gammanu}
|\Gamma({\beta})|^2=\Gamma({\beta})\Gamma(-{\beta})
=-\frac{\pi}{{\beta}\sin\pi{\beta}}
=-\frac{2\pi i}{{\beta}e^{\pi i\beta}(1-e^{-2\pi i{\beta}})}
=4\pi |s_*|^{-2}a^{-2}e^{-\frac{1}{2}\pi a^2}.
\end{equation}
By substituting \eqref{a}, \eqref{gammanu} into \eqref{qexpre4} and \eqref{Hexpres2}, we arrive at the final formulas \eqref{q1} and \eqref{H1} for the case $0<|s_*|<1$.

Similarly, by using  the asymptotic analysis of the RH problem  for $\Psi(\xi,x)$
performed in Section \ref{sec:szero}, we obtain in the special case $|s_{*}|=0$, the following asymptotic formulas
\begin{align}\label{q0}
q(x;\alpha,\kappa)&=-\frac{2}{3}x+O(x^{-1}), \\
\mathcal{H}(x;\alpha,\kappa)&=-\frac{8}{27}x^3+\frac{4\alpha}{3}x+O(x^{-1}),\label{H0}
\end{align}
as $x\rightarrow-\infty$. Notice that it follows from \eqref{nu} that ${\beta}=0$ if $|s_*|=0$. Therefore, \eqref{q0} and \eqref{H0} can be respectively regarded as the limits of \eqref{q1} and \eqref{H1} as $b_1\rightarrow0$.

\subsection{Derivations of \eqref{q3} and \eqref{H3}}
Tracing back the above modified RH analysis for the case $\kappa(\kappa-\kappa^*)<0$, it follows from \eqref{U(z)}, \eqref{T(z)} and \eqref{Rtilde} that
\begin{equation}\label{RHanalysis1}
e^{\frac{x^2}{3}\sigma_3}\Phi(z)
e^{-x^2g(z)\sigma_3}=\widetilde{\mathbf{R}}(z)\widetilde{\mathbf{P}}^{(\infty)}(z),
\end{equation}
for large $z$. Accordingly, to compute the coefficients $\Phi_1$ and $\Phi_2$ in expansion \eqref{Asyatinfty1}, we need the asymptotic approximation of $\widetilde{\mathbf{P}}^{(\infty)}(z)$ as $z\rightarrow\infty$. By \eqref{H} and \eqref{Pinftyatinfty}, we have
\begin{equation}\label{Ptildeatinfty}
\widetilde{\mathbf{P}}^{(\infty)}(z)z^{-\alpha\sigma_3}=\mathbf{I}
+\frac{\widetilde{\mathbf{P}}_1^{(\infty)}}{z}
+\frac{\widetilde{\mathbf{P}}_2^{(\infty)}}{z^2}
+O\left(\frac{1}{z^3}\right),\quad \mathrm{as} \quad z\rightarrow\infty,
\end{equation}
where
\begin{equation}\label{Ptilde1Ptilde2}
\widetilde{\mathbf{P}}_1^{(\infty)}=\mathbf{P}_1^{(\infty)}+\widetilde{A}+\widetilde{B},
\qquad \widetilde{\mathbf{P}}_2^{(\infty)}=(\widetilde{A}+\widetilde{B})\mathbf{P}_1^{(\infty)}
+\sqrt{\frac{2}{3}}(\widetilde{A}-\widetilde{B})+\mathbf{P}_2^{(\infty)},
\end{equation}
with $\widetilde{A}$, $\widetilde{B}$ and $\mathbf{P}_1^{(\infty)}$, $\mathbf{P}_2^{(\infty)}$ given by \eqref{reptildeAB} and \eqref{P1P2}, respectively. Using \eqref{Rtildeexp}, \eqref{RHanalysis1} and \eqref{Ptildeatinfty}, we now obtain
\begin{equation}\label{Phi11}
e^{\frac{x^2}{3}\sigma_3}\Phi_1e^{-\frac{x^2}{3}\sigma_3}
=\mathbf{P}_1^{(\infty)}+\widetilde{\mathbf{R}}_1+\widetilde{A}+\widetilde{B},
\end{equation}
and
\begin{equation}\label{Phi22}
e^{\frac{x^2}{3}\sigma_3}\Phi_2e^{-\frac{x^2}{3}\sigma_3}=\frac{4x^2}{27}\sigma_3
+\widetilde{\mathbf{R}}_1(\widetilde{A}+\widetilde{B}+\mathbf{P}_1^{(\infty)})
+\widetilde{\mathbf{R}}_2+(\widetilde{A}+\widetilde{B})\mathbf{P}_1^{(\infty)}
+\sqrt{\frac{2}{3}}(\widetilde{A}-\widetilde{B})+\mathbf{P}_2^{(\infty)}.
\end{equation}
Inserting \eqref{Phi11}, \eqref{Phi22} into the formulas \eqref{solu2} and \eqref{tau1}, we obtain the following expressions for $q(x;\alpha,\kappa)$ and $\mathcal{H}(x;\alpha,\kappa)$
\begin{align}\label{qexpre6}
q(x;\alpha,\kappa)&=-x\,\Big[(\mathbf{P}^{(\infty)}_1)_{12}+\widetilde{A}_{12}
+\widetilde{B}_{12}+(\widetilde{\mathbf{R}}_1)_{12}\Big]\nonumber\\
&\qquad\qquad\qquad\qquad
\times\Big[(\mathbf{P}^{(\infty)}_1)_{21}+\widetilde{A}_{21}+\widetilde{B}_{21}
+(\widetilde{\mathbf{R}}_1)_{21}\Big],\\
\mathcal{H}(x;\alpha,\kappa)&=-x\,\bigg\{\frac{8}{27}x^2+\left[\widetilde {\mathbf{R}}_1(\widetilde{A}+\widetilde{B}+\mathbf{P}_1^{(\infty)})\right]_{11}
-\left[\widetilde {\mathbf{R}}_1(\widetilde{A}+\widetilde{B}+\mathbf{P}_1^{(\infty)})\right]_{22}\nonumber\\
&\qquad\qquad +(\widetilde{\mathbf{R}}_2)_{11}-(\widetilde {\mathbf{R}}_2)_{22}+\left[(\widetilde{A}+\widetilde{B})\mathbf{P}_1^{(\infty)}\right]_{11}
-\left[(\widetilde{A}+\widetilde{B})\mathbf{P}_1^{(\infty)}\right]_{22}\nonumber\\
&\qquad \qquad\qquad +\sqrt{\frac{2}{3}}(\widetilde{A}_{11}-\widetilde{A}_{22}
+\widetilde{B}_{22}-\widetilde{B}_{11})
+(\mathbf{P}_2^{(\infty)})_{11}-(\mathbf{P}_2^{(\infty)})_{22}\bigg\}.\label{Hexpres3}
\end{align}

Next, we compute the asymptotics of $\widetilde{\mathbf{R}}_1$ and $\widetilde{\mathbf{R}}_2$ as $x\rightarrow-\infty$. In view of the error estimate \eqref{Rtildeexpx} and following similar analysis as we have done in Section \ref{Derivation1}, we easily get that
\begin{equation}\label{tildeR1asym}
\widetilde{\mathbf{R}}_1=O(|x|^{-2}),\qquad
\widetilde{\mathbf{R}}_2=O(|x|^{-2}),\quad\mathrm{as}\quad x\rightarrow-\infty.
\end{equation}
Now, substituting \eqref{reptildeAB}, \eqref{P1P2} and \eqref{tildeR1asym} into \eqref{qexpre6} and \eqref{Hexpres3} gives
\begin{align}\label{qexpre7}
q(x;\alpha,\kappa)&=-2x-4\sqrt{\frac{2}{3}}i\cdot\frac{l-\sqrt{\frac{2}{3}}i}{l^2-2}x
+O(x^{-1}),\\
\mathcal{H}(x;\alpha,\kappa)&=-\frac{8}{27}x^3
+\frac{4}{3}\left(\alpha+\sqrt{3}i{\beta}\right)x
+i\frac{4l}{\sqrt{3}}x\cdot\frac{l-\sqrt{\frac{2}{3}}i}{l^2-2}+O(x^{-1}),\label{Hexpres4}
\end{align}
where the error terms are uniform for $x$ bounded away from the zeros of $l^2-2$. From the definition of the function $l$  in \eqref{c}, we may rewrite
\begin{equation*}
\frac{l-\sqrt{\frac {2}{3}}i}{l^2-2}
=\frac{i}{4}\sqrt{\frac{2}{3}}\frac{(2+e^{i\phi})(1+2e^{i\phi})}
{1+e^{i\phi}+e^{2i\phi}}=\frac{i}{4}\sqrt{\frac{2}{3}}
\left(2+\frac{3}{2\cos\phi+1}\right).
\end{equation*}
By inserting the above equation into the expressions \eqref{qexpre7} and \eqref{Hexpres4},  we arrive at the  formulas \eqref{q3} and \eqref{H3}.

\subsection{Derivations of \eqref{q2} and \eqref{H2}}
We now consider the case $\kappa=\kappa^*$. By tracing back the transformations \eqref{Ubar} and \eqref{Rbar}, it follows that for large $z$
\begin{equation}\label{largez}
e^{\frac{x^2}{2}\sigma_3}|x|^{\alpha\sigma_3}\Phi(z)|x|^{-\alpha\sigma_3}
e^{-x^2\widehat{g}(z)\sigma_3}=\widehat{\mathbf{R}}(z)\widehat{\mathbf{P}}^{(\infty)}(z).
\end{equation}
To compute the coefficients $\Phi_1$ and $\Phi_2$ in expansion \eqref{Asyatinfty1}, we need to write down the asymptotic  approximation  of $\widehat{\mathbf{P}}^{(\infty)}(z)$ as $z\rightarrow\infty$. By \eqref{Pbarinfty}, \eqref{H(z)bar} and \eqref{ABexpres}, we get that as $z\rightarrow\infty$
\begin{equation}\label{Pbaratinfty}
\widehat{\mathbf{P}}^{(\infty)}(z)z^{-\alpha\sigma_3}=
\mathbf{I}+\frac{\widehat{\mathbf{P}}_1^{(\infty)}}{z}
 +\frac{\widehat{\mathbf{P}}_2^{(\infty)}}{z^2}+O\left(\frac{1}{z^3}\right),
\end{equation}
where
\begin{equation}\label{Pbar1Pbar2}
\widehat{\mathbf{P}}_1^{(\infty)}=\begin{pmatrix} 0 & -\frac{1}{2}c_r \\ -\frac{1}{2}c_l & 0\end{pmatrix},\qquad
\widehat{\mathbf{P}}_2^{(\infty)}=\begin{pmatrix} 1-2\alpha & 0 \\ 0 & 1+2\alpha\end{pmatrix}
\end{equation}
with $c_r$, $c_l$ given by \eqref{cr} and \eqref{cl}, respectively.

From \eqref{largez} and the large-$z$  expansions \eqref{Rbarexp} and \eqref{Pbaratinfty}, we have
\begin{equation}\label{Phi111}
\Phi_1=e^{-\frac{x^2}{2}\sigma_3}|x|^{-\alpha\sigma_3}
\left(\widehat{\mathbf{P}}_1^{(\infty)}
+\widehat{\mathbf{R}}_1\right)|x|^{\alpha\sigma_3}e^{\frac{x^2}{2}\sigma_3}
\end{equation}
and
\begin{equation}\label{Phi222}
\Phi_2=e^{-\frac{x^2}{2}\sigma_3}|x|^{-\alpha\sigma_3}
\left(\widehat{\mathbf{P}}_2^{(\infty)}
+\widehat{\mathbf{P}}_1^{(\infty)}\widehat{\mathbf{R}}_1
+\widehat{\mathbf{R}}_2\right)|x|^{\alpha\sigma_3}e^{\frac{x^2}{2}\sigma_3}.
\end{equation}
Now, recalling \eqref{solu2} and \eqref{tau1}, we  obtain from \eqref{Pbar1Pbar2}, \eqref{Phi111} and \eqref{Phi222} that
\begin{equation}\label{qexpre5}
q(x;\alpha,\kappa)=-x\left[-\frac{c_r}{2}+(\widehat{\mathbf{R}}_1)_{12}\right]
\left[-\frac{c_l}{2}+(\widehat{\mathbf{R}}_1)_{21}\right]
\end{equation}
and
\begin{equation}\label{Hexpres5}
\mathcal{H}(x;\alpha,\kappa)=-x\bigg\{-4\alpha-\frac{1}{2}
\left[c_r(\widehat{\mathbf{R}}_1)_{21}-c_l(\widehat{\mathbf{R}}_1)_{12}\right]
+(\widehat{\mathbf{R}}_2)_{11}-(\widehat{\mathbf{R}}_2)_{22}\bigg\}.
\end{equation}
In view of the estimation \eqref{Rbarexpans}, we easily obtain the following asymptotics for $\widehat{\mathbf{R}}_1$ and $\widehat{\mathbf{R}}_2$ as $x\rightarrow-\infty$
\begin{equation}\label{Rbar1Rbar2asymp}
\widehat{\mathbf{R}}_1=O(|x|^{-2}),\quad \widehat{\mathbf{R}}_2=O(|x|^{-2}).
\end{equation}
Substituting the asymptotics \eqref{Rbar1Rbar2asymp} into the expressions \eqref{qexpre5} and \eqref{Hexpres5} and using the fact that $c_rc_l=8$, we obtain \eqref{q2} and \eqref{H2}.

\subsection{Derivation of \eqref{Hasymp+infty} }
It is seen from \eqref{qAsy} and \eqref{eq:dH} that
\begin{equation}\label{eq:dHAsy}
\frac{d}{dx}\mathcal{H}(x;\alpha,\kappa)=\kappa\, 2^{\alpha+\frac{1}{2}}x^{2\alpha}e^{-x^2}\left(1+O\left(x^{-2}\right)\right), \quad \mathrm{as}~~x\to+\infty.
\end{equation}
Moreover, it follows from \eqref{qAsy} and \eqref{eq:Hq} that
\begin{equation}\label{eq:Hlimit} \mathcal{H}(x;\alpha,\kappa)\to0, \quad \mathrm{as}~~x\to+\infty. \end{equation}
Thus,  integrating    both sides of \eqref{eq:dHAsy} and using the boundary condition
\eqref{eq:Hlimit}, we obtain \eqref{Hasymp+infty}.

\subsection{Derivations of \eqref{integral-1} and \eqref{integral-2}}\label{proof2}

Recalling from \eqref{Fequation} and using the fact that $F(x)\to \mathbf{I}$ as $x\to+\infty$,  as can be seen from the asymptotic analysis of $\Psi(\xi,x)$ as $x\to+\infty$ performed in \cite{IK},  we have
\begin{equation}\label{eq:qInt}
F(x):=\Psi^{(0)}(0,x)=\exp\left(\int_{+\infty}^xq(t)dt\sigma_3\right),
\end{equation}
with  $\Psi^{(0)}(\xi,x)$ given in  \eqref{Asyatzero}.
Hence,
the evaluation of the total integral of $q(x)$ is  boiled down  to computing the asymptotic of $F(x)$ as $x\to-\infty$.
In the remaining part of this section, we derive the  asymptotics of $F(x)$ as $x\to-\infty$ and prove Theorem \ref{thm3} based on the asymptotic analysis of the RH problem for $\Psi(\xi,x)$
performed in Sections \ref{Asymptotic-infty1} and \ref{Asymptotic-infty3}.

\subsubsection{Derivation of \eqref{integral-1}}
In the case $\kappa(\kappa-\kappa^{*})<0$, by inverting the transformations \eqref{rescaling}, \eqref{U(z)}, \eqref{T(z)} and \eqref{R}, we obtain that for $z$ small and $\arg z\in(\frac{\pi}{4},\frac{\pi}{2})$,
\begin{align}\label{Psiexpress2}
\Psi(|x|^{\frac{1}{2}}z,x)
&=e^{-\frac{x^2}{3}\sigma_3}|x|^{\frac{\alpha}{2}\sigma_3}
\mathbf{R}(z)\mathbf{P}^{(0)}(z)\begin{pmatrix}1 & 0\\ \frac{\overline{s}_*e^{-2\pi i\alpha}e^{-2x^2g(z)}}
{s_1(e^{-2\pi i\alpha}+s_*)} & 1 \end{pmatrix}e^{x^2g(z)\sigma_3}\nonumber\\
&=e^{-\frac{x^2}{3}\sigma_3}|x|^{\frac{\alpha}{2}\sigma_3}
\mathbf{R}(z)\mathbf{E}^{(0)}(z)\Phi^{(\mathrm{Bes})}(x^2\varphi_4(z))\nonumber\\
&\qquad\qquad\qquad \times\begin{pmatrix}1 & 0\\ -e^{-2\pi i\alpha} & 1 \end{pmatrix}\left[s_1(e^{-2\pi i\alpha}+s_*)\right]^{-\frac{\sigma_3}{2}}\begin{pmatrix}1 & 0\\ \frac{\bar{s}_*e^{-2\pi i\alpha}}
{s_1(e^{-2\pi i\alpha}+s_*)} & 1 \end{pmatrix}.
\end{align}
It is readily seen from \eqref{Asyatzero} and \eqref{Psiexpress2} that
\begin{align}\label{Fexpress6}
F(x)&=e^{-\frac{x^2}{3}\sigma_3}|x|^{\frac{\alpha}{2}\sigma_3}
\mathbf{R}(0)\mathbf{E}^{(0)}(0)\lim\limits_{z\to0}\Phi^{(\mathrm{Bes})}(x^2\varphi_4(z))
\begin{pmatrix}1 & 0\\ -e^{-2\pi i\alpha} & 1 \end{pmatrix}\nonumber\\
&\qquad\qquad\qquad \times\left[s_1(e^{-2\pi i\alpha}+s_*)\right]^{-\frac{\sigma_3}{2}}\begin{pmatrix}1 & 0\\ \frac{\bar{s}_*e^{-2\pi i\alpha}}{s_1(e^{-2\pi i\alpha}+s_*)} & 1 \end{pmatrix}S_1^{-1}E_0^{-1}z^{-\alpha\sigma_3}|x|^{-\frac{\alpha}{2}\sigma_3}.
\end{align}
By \eqref{E0}, \eqref{varphi4beha} and \eqref{BesParaExpand}, we have
\begin{align}\label{limit3}
&\lim\limits_{z\to0}\Phi^{(\mathrm{Bes})}(x^2\varphi_4(z))\begin{pmatrix}1 & 0\\ -e^{-2\pi i\alpha} & 1 \end{pmatrix}[s_1(e^{-2\pi i\alpha}+s_*)]^{-\frac{\sigma_3}{2}}\begin{pmatrix}1 & 0\\ \frac{\bar{s}_*e^{-2\pi i\alpha}}{s_1(e^{-2\pi i\alpha}+s_*)} & 1 \end{pmatrix}S_1^{-1}E_0^{-1}z^{-\alpha\sigma_3}\nonumber\\
&\qquad=C_{\alpha}^{\sigma_3}\begin{pmatrix}0 & -1\\ 1 & 0 \end{pmatrix}2^{\frac{3}{2}\alpha\sigma_3}3^{-\frac{3}{2}\alpha\sigma_3}
s_1^{-\frac{\sigma_3}{2}}(e^{-2\pi i\alpha}+s_*)^{\frac{\sigma_3}{2}}(1+e^{-2\pi i\alpha})^{-\sigma_3}|x|^{2\alpha\sigma_3},
\end{align}
where the constant $C_{\alpha}$ is given by \eqref{Calpha}. Recalling the definition \eqref{E(0)} of $\mathbf{E}^{(0)}(z)$, it follows from \eqref{Pinfty}, \eqref{finfty} and \eqref{nu} that
\begin{align}\label{E(0)at0}
\mathbf{E}^{(0)}(0)&=\frac{1}{2}s_0^{-\frac{\sigma_3}{2}}2^{\frac{\alpha}{2}\sigma_3}
3^{-\frac{\alpha}{2}\sigma_3}e^{-\frac{\pi i{\beta}}{3}\sigma_3}
\begin{pmatrix}1 & -1\\ 1 & 1 \end{pmatrix}e^{-\frac{3\pi i}{4}\sigma_3}
\begin{pmatrix}1 & i\\ i & 1 \end{pmatrix}.
\end{align}
Substituting \eqref{limit3} and \eqref{E(0)at0} into \eqref{Fexpress6} yields
\begin{align}\label{Fexpress++}
e^{\frac{x^2}{3}\sigma_3}|x|^{-\frac{\alpha}{2}\sigma_3}F(x)|x|^{-\frac{3\alpha}{2}\sigma_3}
=\mathbf{R}(0)2^{(\alpha-\frac{1}{2})\sigma_3}3^{-2\alpha\sigma_3}\pi^{-\frac{\sigma_3}{2}}
\Gamma(\textstyle\frac{1}{2}-\alpha)^{\sigma_3}[(1-s_*)e^{\pi i\alpha}]^{-\sigma_3}e^{-\frac{4}{3}\pi i{\beta}\sigma_3}.
\end{align}

We choose some constants  $c<0$ and $d>0$ such that all real poles of $q(x;\alpha,\kappa)$ lie in the $x$-interval $(c,d)$. Using the estimation \eqref{Restimation}, the expression \eqref{eq:qInt} and letting $x\to-\infty$ in \eqref{Fexpress++}, we obtain
\begin{align}\label{explimit1}
&\exp\Bigg\{\int^{d}_{+\infty}q(t)dt+\int_{\Upsilon}q(t)dt
+\int^{-\infty}_{c}\left(q(t)+\frac{2t}{3}-\frac{2\alpha}{t}\right)dt
+\frac{c^2}{3}-2\alpha\ln|c|\Bigg\}\nonumber\\
&\qquad\qquad=2^{\alpha-\frac{1}{2}}
3^{-2\alpha}\pi^{-\frac{1}{2}}\Gamma\left(\textstyle\frac{1}{2}-\alpha\right)[(1-s_*)
e^{\pi i\alpha}]^{-1}e^{-\frac{4}{3}\pi i{\beta}}.
\end{align}
where $\Upsilon$ is any contour in the complex plane from $d$ to $c$ that avoids the real poles of $q(x;\alpha,\kappa)$. Indeed, we may   take for $\Upsilon$ a path along the real axis  with infinitesimal semicircular indentations centered at each real pole of $q(x;\alpha,\kappa)$ in the lower half complex plane. Using the fact that all poles of $q(x;\alpha,\kappa)$ are simple with residues $\pm1$, we have
\begin{equation}\label{qintegral1}
\int_{\Upsilon}q(t)dt=\mathrm{P.V.}\int^{c}_{d}q(t)dt+\pi i(N_+-N_-),
\end{equation}
where $\mathrm{P.V.}$ denotes the Cauchy principal value and $N_{\pm}$ are the numbers of real poles of $q(x;\alpha,\kappa)$ of residue $\pm1$, respectively.

By inserting \eqref{qintegral1} into \eqref{explimit1} and using the definitions of $\kappa$, ${\beta}$
as given by \eqref{kapparep}, \eqref{nu}, respectively, we arrive at the total integral \eqref{integral-1}.


\subsubsection{Derivation of \eqref{integral-2}}
Now, we concentrate on the case $\kappa=\kappa^{*}$. Tracing back the transformations \eqref{rescaling}, \eqref{Ubar} and \eqref{Rbar} gives
\begin{align}\label{Psiexpress3}
\Psi(|x|^{\frac{1}{2}}z,x)
&=e^{-\frac{x^2}{2}\sigma_3}|x|^{-\frac{\alpha}{2}\sigma_3}
\widehat{\mathbf{R}}(z)\widehat{\mathbf{P}}^{(0)}(z)|x|^{\alpha\sigma_3}
e^{x^2\widehat{g}(z)\sigma_3}.
\end{align}
It follows from \eqref{Asyatzero}, \eqref{Fequation}, \eqref{Pbar0} and \eqref{Psiexpress3} that
\begin{align}\label{Fexpress7}
F(x)=e^{-\frac{x^2}{2}\sigma_3}|x|^{-\frac{\alpha}{2}\sigma_3}
&\widehat{\mathbf{R}}(0)\widehat{\mathbf{H}}(0)\lim\limits_{z\to0}
\begin{pmatrix}1 & 0\\ r(z) & 1 \end{pmatrix}
z^{-\alpha\sigma_3}(z+\sqrt{2})^{(\alpha+\frac{1}{2})\sigma_3}\nonumber\\
&\qquad\qquad \times (z-\sqrt{2})^{(\alpha-\frac{1}{2})\sigma_3}|x|^{\alpha\sigma_3}
e^{x^2\widehat{g}(z)\sigma_3}S_1^{-1}E_0^{-1}z^{-\alpha\sigma_3}
|x|^{-\frac{\alpha}{2}\sigma_3},
\end{align}
for $z\to0$ with $\arg{z}\in(\frac{\pi}{4},\frac{\pi}{2})$.

Notice that by \eqref{H(z)bar}, \eqref{ABexpres} and \eqref{rexpre},
\begin{equation}\label{H(z)barat0}
\widehat{\mathbf{H}}(0)=(2\pi)^{-\frac{\sigma_3}{2}}e^{\pi i(\alpha-\frac{1}{2})\sigma_3}\Gamma(\alpha+\textstyle\frac{1}{2})^{\sigma_3}
\begin{pmatrix}0 & s_1 \\ -s_1^{-1} & 0 \end{pmatrix}
\end{equation}
and
\begin{align}\label{limit4}
&e^{-x^2\widehat{g}(z)\sigma_3}|x|^{-\alpha\sigma_3}(z+\sqrt{2})^{-(\alpha+\frac{1}{2})\sigma_3}
(z-\sqrt{2})^{-(\alpha-\frac{1}{2})\sigma_3}
z^{\alpha\sigma_3}\begin{pmatrix}1 & 0\\ r(z) & 1 \end{pmatrix}\nonumber\\
&\qquad\qquad\qquad\times z^{-\alpha\sigma_3}(z+\sqrt{2})^{(\alpha+\frac{1}{2})\sigma_3}
(z-\sqrt{2})^{(\alpha-\frac{1}{2})\sigma_3}|x|^{\alpha\sigma_3}
e^{x^2\widehat{g}(z)\sigma_3}=\begin{pmatrix}1 & 0\\ -\frac{s_0e^{2\pi i\alpha}}{1+e^{2\pi i\alpha}} & 1 \end{pmatrix}.
\end{align}
Substituting \eqref{limit4} and \eqref{H(z)barat0} into \eqref{Fexpress7} and combining with the facts
\begin{equation*}
\begin{pmatrix}0 & s_1 \\ -s_1^{-1} & 0 \end{pmatrix}=\begin{pmatrix}1 & 0 \\ -s_1^{-1} & 1 \end{pmatrix}\begin{pmatrix}1 & s_1 \\ 0 & 1 \end{pmatrix}\begin{pmatrix}1 & 0 \\ -s_1^{-1} & 1 \end{pmatrix}
\end{equation*}
and
\begin{equation*}
E_0=\begin{pmatrix}1 & 0\\ \frac{s_0e^{2\pi i\alpha}}{1+e^{2\pi i\alpha}} & 1 \end{pmatrix}=\begin{pmatrix}1 & 0 \\ -s_1^{-1} & 1 \end{pmatrix},
\end{equation*}
which follows from \eqref{E0} and \eqref{sstarvalue}, we have
\begin{align}\label{Fexpress+++}
e^{\frac{x^2}{2}\sigma_3}&|x|^{\frac{\alpha}{2}\sigma_3}F(x)|x|^{\frac{3\alpha}{2}\sigma_3}
e^{\frac{x^2}{2}\sigma_3}=
\widehat{\mathbf{R}}(0)(2\pi)^{-\frac{\sigma_3}{2}}e^{\pi i(\alpha-\frac{1}{2})\sigma_3}
\Gamma(\alpha+\textstyle\frac{1}{2})^{\sigma_3}
2^{-\alpha\sigma_3}e^{\pi i(\frac{1}{2}-\alpha)\sigma_3}.
\end{align}

Similarly, take $c<0$ and $d>0$ such that all real poles of $q(x;\alpha,\kappa)$ lie in the interval $(c,d)$.
Using the estimation \eqref{Rbarexpans} and letting $x\to-\infty$ in \eqref{Fexpress+++}, we
obtain  \eqref{integral-2}.

\section{Proof of Theorem \ref{thm:IntRep}}\label{sec:proof of det}

Denote
\begin{align}\label{def:f}
\mathbf{f}(\lambda)&=\begin{pmatrix}f_1(\lambda)\\ f_2(\lambda) \end{pmatrix}=\begin{pmatrix}D_{\nu}\left(\sqrt{2}(\lambda+x)\right) \\
-D_{\nu-1}\left(\sqrt{2}(\lambda+x)\right) \end{pmatrix}, \\
\mathbf{h}(\lambda)&=\begin{pmatrix}h_1(\lambda)\\ h_2(\lambda) \end{pmatrix}=\gamma\begin{pmatrix} D_{\nu-1}\left(\sqrt{2}(\lambda+x)\right)  \\
D_{\nu}\left(\sqrt{2}(\lambda+x)\right) \end{pmatrix}, \label{def:h}
\end{align}
where $D_{\nu}$ is the parabolic cylinder function.
Then, the parabolic cylinder kernel \eqref{eq:PCKernel} can be written as
 \begin{equation}\label{eq:PCKernel1}
\gamma  K_{\nu,x}(\lambda,\mu)=\frac{\mathbf{ f}(\lambda)^{T}\mathbf{ h}(\mu)}{\lambda-\mu}.\end{equation}

For general parameter $\nu\in \mathbb{R}$, the logarithmic derivative of  $\gamma  K_{\nu,x}(\lambda,\mu)$ can  be  expressed in terms of the solution of the following RH problem
for $Y(\lambda):=Y(\lambda, x)$.
\subsection*{RH problem for $Y(\lambda)$}
\begin{itemize}
\item[\rm (1)]
$Y(\lambda)$ is analytic in
  $\mathbb{C}\setminus [0,\infty)$.

\item[\rm (2)]  $Y(\lambda)$  satisfies the jump condition
 \begin{equation}\label{eq:YJump}
 Y_+(\lambda)=Y_-(\lambda) \left(\mathbf{I}-2\pi i \mathbf{ f}(\lambda) \mathbf{ h}(\lambda)^{T}\right),
\end{equation}
where $\mathbf{ f}$ and $\mathbf{ h}$ are defined by \eqref{def:f} and \eqref{def:h}, respectively.

\item[\rm (3)]  The behavior of $Y(\lambda)$ at infinity is
  \begin{equation}\label{eq:YInfinity}Y(\lambda)=\mathbf{I}+\frac {Y_{1}}{\lambda}+O\left (\frac 1 {\lambda^2}\right).\end{equation}
  \item[\rm (4)]  The behavior of $Y(\lambda)$ at the origin is
   \begin{equation}\label{eq:Y0}Y(\lambda)=O\left (\ln\lambda\right).\end{equation}
\end{itemize}

Actually, the solution to the RH problem for $Y(\lambda)$ can be expressed as follows \cite{IIKS}
\begin{equation}\label{eq:Ysolution}Y(\lambda)=\mathbf{I}-\int_{0}^{+\infty} \frac{\mathbf{ V}(\mu)\mathbf{ h}(\mu)^{T}}{\mu-\lambda}d\mu,\end{equation}
where
  \begin{equation}\label{eq:YFH}
  \mathbf{ V}(\mu)=(\mathbf{I}-\gamma K_{\nu,x})^{-1} \mathbf{ f}(\mu) .\end{equation}
The equation \eqref{eq:PCKernel1}  implies
 \begin{equation}\label{eq:DF1}
\frac{d}{dx} \ln \det(\mathbf{I}-\gamma K_{\nu,x})=-\mathrm{ tr} \left(\left(\mathbf{I}-\gamma K_{\nu,x}\right)^{-1} \gamma  \frac{d}{dx}
K_{\nu,x}\right).\end{equation}
Using the recurrence relations satisfied by
the parabolic cylinder function  (cf. \cite[Equation (12.8)]{NIST})
\begin{equation}\label{eq:PCRelation1}
D_{\nu}'(\lambda)=\frac{\lambda}{2}D_{\nu}(\lambda)-D_{\nu+1}(\lambda),\end{equation}
and
\begin{equation}\label{eq:PCRelation2}
D_{\nu}'(\lambda)=-\frac{\lambda}{2}D_{\nu}(\lambda)+\nu D_{\nu-1}(\lambda), \end{equation}
 we find  after some direct calculation that
 \begin{equation}\label{eq:dK}
\gamma\frac{d}{dx} K_{\nu,x}(\lambda,\mu) =-\mathbf{ h}(\mu)^{T}\sigma_3 \mathbf{ f}(\lambda)=-\left(f_1(\lambda)h_1(\mu)
-f_2(\lambda)h_2(\mu)\right).\end{equation}
Substituting \eqref{eq:dK} into \eqref{eq:DF1} and applying \eqref{eq:Ysolution}-\eqref{eq:YFH}, we arrive at
the differential identity
\begin{equation}\label{eq:FY}
\frac{d}{dx} \ln \det(\mathbf{I}-\gamma K_{\nu,x})=(Y_1)_{11}-(Y_1)_{22}=2(Y_1)_{11},
\end{equation}
with $Y_1$ given in \eqref{eq:YInfinity}.

To proceed, we define
   \begin{equation}\label{eq:PCM0}
P(z)=\begin{pmatrix}e^{i\frac{\pi}{2}\nu} D_{-\nu}(e^{i\frac{\pi}{2}}z)& D_{\nu-1}(z) \\
\nu e^{i\frac{\pi}{2}(\nu+1)} D_{-\nu-1}(e^{i\frac{\pi}{2}}z) & D_{\nu}(z) \end{pmatrix}
\end{equation}
for  $\arg z\in(-\frac{1}{4}\pi,0)$.
In view of the asymptotic behavior \eqref{eq:DAsy} for $D_{\nu}$, we see that $P(z)$ satisfies the following asymptotic behavior at infinity
  \begin{equation}\label{eq:PAsy}
P(z)=\left(\mathbf{I}+\frac{P_1}{z}+O\left(\frac 1{z^{2}}\right)
\right)z^{-\nu \sigma_3}\exp\left(\frac{1}{4}z^2\sigma_3\right),
\end{equation}
 where the diagonal entries of $P_1$ are zero.  Define
 \begin{equation}\label{eq:Psi}
\widehat{\Psi}(\lambda)=e^{-\frac{1}{2}x^2\sigma_3}2^{\frac{\nu}{2}\sigma_3} Y(\lambda)^{-T} P\left(\sqrt{2}(\lambda+x)\right)\prod_{i=0}^{k-1}J_i,
\end{equation}
for $\arg\lambda\in (\frac{k-1}{2}\pi,\frac{k}{2}\pi)$, $k=1,2,3$.
Here, the constant matrices
\begin{equation}\label{eq:J_k}
 J_0=  \left(
                               \begin{array}{cc}
                                 1& 0\\
                             j_0 & 1 \\
                                 \end{array}
                             \right), \quad  J_1= \left(
                               \begin{array}{cc}
                                 1& j_1\\
                               0 & 1 \\
                                 \end{array}
                             \right), \quad
                               J_2= \left(
                               \begin{array}{cc}
                                 1& 0\\
                              -j_0e^{-2\pi i\nu} & 1 \\
                                 \end{array}
                             \right), \end{equation}
with $j_0=-i\frac{\sqrt{2\pi}}{\Gamma(\nu)} $ and $j_1=-\frac{\sqrt{2\pi}}{\Gamma(1-\nu)}e^{i\pi \nu}$.
We also denote
\begin{equation}\label{eq:J3}
 J_3=  (J_0J_1J_2)^{-1}=\left(
                               \begin{array}{cc}
                                 1& -j_1e^{2\pi i\nu}\\
                               0 & 1 \\
                                 \end{array}
                             \right)e^{2\pi i\nu\sigma_3}.
                             \end{equation}

Then, it is direct to see that $\widehat\Psi(\lambda)$ satisfies the following RH problem.
\subsection*{RH problem for $\widehat{\Psi}(\lambda)$}
\begin{itemize}
\item[\rm (1)]
$\widehat\Psi(\lambda)$ is analytic in
  $\mathbb{C}\setminus \Xi_k$,
  where
  $\Xi_k=e^{i\frac{k}{2}\pi}\mathbb{R}_+$, $k=0,1,2,3$.

\item[\rm (2)]  $\widehat\Psi(\lambda)$  satisfies the jump conditions
 \begin{equation}\label{eq:PsiJump1}\widehat\Psi_+(\lambda)=\widehat\Psi_-(\lambda)
 \left(\begin{array}{cc}1& 0\\ 2\pi i \gamma-i\frac{\sqrt{2\pi}}{\Gamma(\nu)} & 1 \\
 \end{array}\right), \quad \lambda\in\Xi_0,\end{equation}
  and
  \begin{equation}\label{eq:PsiJump2}\widehat\Psi_+(\lambda)=\widehat\Psi_-(\lambda)
J_k, \quad \lambda\in\Xi_k,\end{equation}
  with $J_k$, $ k=1,2,3$ given in \eqref{eq:J_k} and \eqref{eq:J3}.
\item[\rm (3)]  $\widehat\Psi(\lambda)$ has the asymptotic behavior as $\lambda\to\infty$
  \begin{equation}\label{eq:PsiInfinity}\widehat\Psi(\lambda)=\left(\mathbf{I}+\frac {\widehat\Psi_{1}}{\lambda}+O\left (\frac 1 {\lambda^2}\right)\right)\lambda^{-\nu \sigma_3}\exp\left((\lambda^2/2+x\lambda)\sigma_3\right).\end{equation}

  \item[\rm (4)]  $\widehat\Psi(\lambda)$ possesses the behavior near the origin
   \begin{equation}\label{eq:Psi0}\widehat\Psi(\lambda)=O\left (\ln\lambda\right).\end{equation}
\end{itemize}

It comes out that the RH problem for $\widehat\Psi(\lambda)$ is the same as the RH problem for  the Jimbo-Miwa Lax pair of the PIV equation with the parameters $\theta_0=0$ and $\theta_{\infty}=\nu$; cf. \cite[(C.30)-(C.31)]{JM}, see also
\cite[Chapter 5.1]{FIKN}.
 It is shown in \cite[(C.34)-(C.37)]{JM} that
\begin{equation}\label{eq:PsiExpand1}
\sigma_{\nu}(x)=-2(\widehat\Psi_1)_{11}-2\nu x
\end{equation}
satisfies the equation \eqref{eq:sPIV}.
Using \eqref{eq:PAsy} and \eqref{eq:Psi}, we find
\begin{equation}\label{eq:PsiExpand0}
2(Y_1)_{11} =-2(\widehat\Psi_1)_{11}-2\nu x.
\end{equation}
Thus, we have
\begin{equation}\label{eq:PsiExpand}
2(Y_1)_{11} =\sigma_{\nu}(x).
\end{equation}
Recalling \eqref{eq:FY}, we have derived \eqref{eq:IntRep}.

Next, we prove \eqref{eq:SigmaAsy}.
It is seen from \eqref{eq:DAsy} that the parabolic cylinder function $D_{\nu}(x)$ decay exponentially fast   as $x$ tends to positive infinity. We get from the series expansion of
the Fredholm determinant  that
\begin{equation}\label{eq:DExpand}
\det(\mathbf{I}-\gamma K_{\nu,x})\sim 1-\gamma\int_0^{+\infty}K_{\nu,x}(\lambda,\lambda)d\lambda. \end{equation}
Thus, we have
\begin{equation}\label{eq:FAsy1}
\frac{d}{dx}  \ln \det(\mathbf{I}-\gamma K_{\nu,x})\sim \gamma K_{\nu,x}(0,0),~~~x\to +\infty.
\end{equation}
The asymptotic \eqref{eq:SigmaAsy}  then follows from \eqref{eq:DAsy}, \eqref{eq:PCKernel} and \eqref{eq:FAsy1}.
We complete the proof of Theorem \ref{thm:IntRep}.


\section*{Acknowledgements}
The authors are grateful to the editor and the referee for  their valuable suggestions and comments.
The work of Shuai-Xia Xu was supported in part by the National Natural Science Foundation of China under grant numbers 11571376 and 11971492, and by  Guangdong Basic and Applied Basic Research Foundation (Grant No. 2022B1515020063). Yu-Qiu Zhao was supported in part by the National Natural Science Foundation of China under grant numbers 11571375 and 11971489.

\begin{appendices}

\section{Local parametrix models}
\subsection{Airy parametrix}\label{AP}
Introduce
\begin{equation}
\Phi^{(\mathrm{Ai})}(\z)=\mathbf{N}\left\{
\begin{aligned}
&\begin{pmatrix}
\mathrm{Ai}(\z) & \mathrm{Ai}(w^{2}\z) \\
\mathrm{Ai}^{\prime}(\z) & w^{2} \mathrm{Ai}^{\prime}(w^{2}\z)
\end{pmatrix} e^{-i \frac{\pi}{6} \sigma_{3}}, &\z&\in \mathrm{I},\\
&\begin{pmatrix}
\mathrm{Ai}(\z) & \mathrm{Ai}(w^{2}\z) \\
\mathrm{Ai}^{\prime}(\z) & w^{2} \mathrm{Ai}^{\prime}(w^{2}\z)
\end{pmatrix} e^{-i \frac{\pi}{6} \sigma_{3}}\begin{pmatrix}
1 & 0 \\
-1 & 1
\end{pmatrix},&\z&\in \mathrm{II},\\
&\begin{pmatrix}
\mathrm{Ai}(\z) & -w^{2}\mathrm{Ai}(w\z) \\
\mathrm{Ai}^{\prime}(\z) & -\mathrm{Ai}^{\prime}(w\z)
\end{pmatrix} e^{-i \frac{\pi}{6} \sigma_{3}}\begin{pmatrix}
1 & 0 \\
1 & 1
\end{pmatrix},&\z&\in \mathrm{III},\\
&\begin{pmatrix}
\mathrm{Ai}(\z) & -w^{2} \mathrm{Ai}(w\z) \\
\mathrm{Ai}^{\prime}(\z) & -\mathrm{Ai}^{\prime}(w\z)
\end{pmatrix} e^{-i \frac{\pi}{6} \sigma_{3}},&\z&\in \mathrm{IV},
\end{aligned}
\right.
\end{equation}
where $\mathrm{Ai}(\z)$ denotes the Airy function (cf. \cite[Chapter 9]{NIST}), $w=e^{2\pi i/3}$ and
$$
\mathbf{N}=\sqrt{2\pi}\,e^{\frac{1}{6}\pi i}\begin{pmatrix}1 & 0\\ 0 & -i \end{pmatrix}.
$$
The regions I-IV are illustrated in Figure \ref{Airy}. Then, $\Phi^{(\mathrm{Ai})}(\z)$ solves the following RH problem (cf. \cite{Deft}).

\subsection*{RH problem for $\Phi^{(\mathrm{Ai})}(\z)$}

\begin{description}
\item{(1)} $\Phi^\mathrm{(Ai)}(\z)$ is analytic for $\z\in \mathbb{C}\setminus \bigcup^4_{k=1}\Sigma_{k}$, where $\Sigma_1=\mathbb{R}_+$, $\Sigma_2=e^{\frac{2\pi i}{3}}\mathbb{R}_+$, $\Sigma_3=\mathbb{R}_-$ and $\Sigma_4=e^{-\frac{2\pi i}{3}}\mathbb{R}_+$ with orientations indicated in Figure \ref{Airy}.
\item{(2)} We have the jump relation
\begin{equation*}
\Phi^\mathrm{(Ai)}_+(\z)=\Phi^\mathrm{(Ai)}_-(\z)\left\{
\begin{aligned}
&\begin{pmatrix}1 & 1 \\ 0 & 1 \end{pmatrix}, && \z\in\Sigma_{1},\\
&\begin{pmatrix}1 & 0 \\ 1 & 1 \end{pmatrix}, && \z\in\Sigma_{2}\cup\Sigma_4,\\
&\begin{pmatrix}0 & 1 \\ -1 & 0 \end{pmatrix}, && \z\in\Sigma_{3}.
\end{aligned}\right.
\end{equation*}
\item{(3)} $\Phi^\mathrm{(Ai)}(\z)$ possesses the following asymptotic behavior as $\z\to \infty$
\begin{equation}\label{AiryAsyatinfty}
\Phi^{(\mathrm{Ai})}(\z)=\z^{-\frac{\sigma_{3}}{4} }\frac{1}{\sqrt{2}} \begin{pmatrix}1 & i\\ i &1\end{pmatrix}\left(\mathbf{I}+O\left(\z^{-\frac{3}{2}}\right)\right) e^{-\frac{2}{3} \z^{\frac{3}{2}} \sigma_{3}}.
\end{equation}
\end{description}

\begin{figure}[t]
  \centering
  \includegraphics[width=8cm,height=6cm]{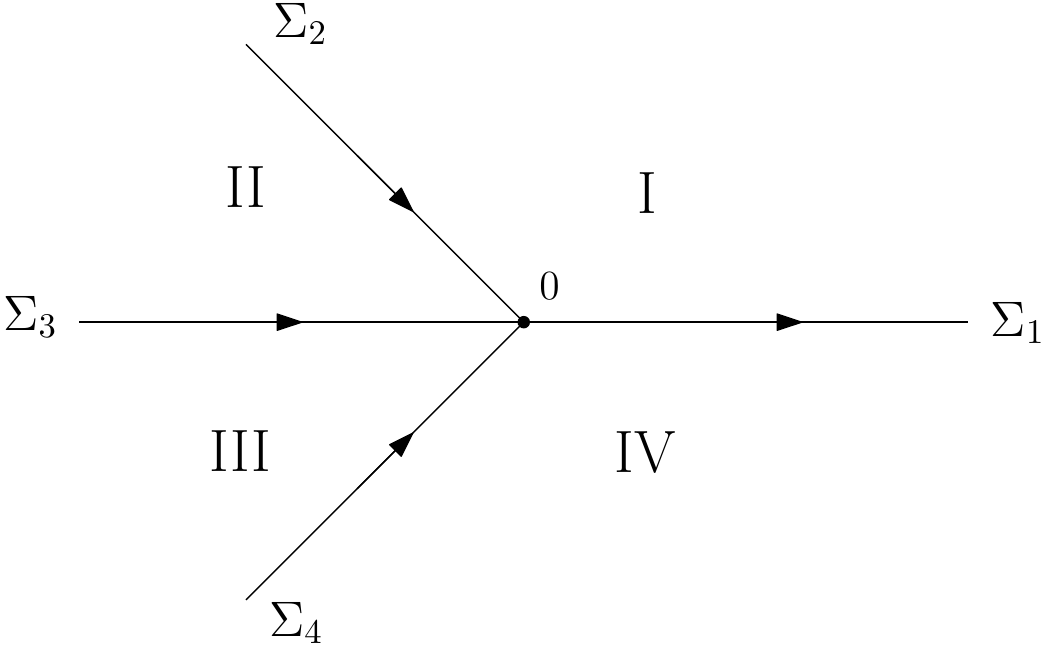}\\
  \caption{The jump contour and regions for $\Phi^{(\mathrm{Ai})}$}\label{Airy}
\end{figure}

\subsection{Parabolic cylinder  parametrix}\label{PCP}
Let ${\beta}$ be a fixed real or complex number. Define
\begin{equation}\label{eq:Dpra}
\mathbf{D}(\z)=2^{-\frac{\sigma_{3}}{2}}\begin{pmatrix}
D_{-{\beta}-1}(i\z) & D_{{\beta}}(\z) \\iD_{-{\beta}-1}'(i\z) & D_{{\beta}}'(\z)\end{pmatrix}\begin{pmatrix}
e^{i \frac{\pi}{2}({\beta}+1)} & 0 \\ 0 & 1 \end{pmatrix},
\end{equation}
where $D_{{\beta}}(\z)$ denotes the standard parabolic cylinder function (cf. \cite[Chapter 12]{NIST}).
Denote
$$
H_{0}=\begin{pmatrix}1 & 0 \\ h_{0} & 1\end{pmatrix}, \quad
H_{1}=\begin{pmatrix}1 & h_{1} \\ 0 & 1\end{pmatrix}, \quad
H_{k+2}=e^{i \pi\left({\beta}+\frac{1}{2}\right) \sigma_{3}} H_{k} e^{-i \pi\left({\beta}+\frac{1}{2}\right) \sigma_{3}}, \ k=0,1,
$$
where
\begin{equation}\label{h0}
h_{0}=-i \frac{\sqrt{2 \pi}}{\Gamma({\beta}+1)}, \quad h_{1}=\frac{\sqrt{2 \pi}}{\Gamma(-{\beta})} e^{i \pi {\beta}}, \quad 1+h_{0} h_{1}=e^{2 \pi i {\beta}}.
\end{equation}
We consider the function
$$
\Phi^{(\mathrm{PC})}(\z)=\left\{\begin{aligned}
&\mathbf{ D}(\z),   \quad & \arg \z \in&\left(-\frac{\pi}{4}, 0\right), \\
&\mathbf{ D}(\z)H_0,\quad & \arg \z \in&\left(0, \frac{\pi}{2}\right), \\
&\mathbf{ D}(\z)H_1,\quad & \arg \z \in&\left(\frac{\pi}{2}, \pi\right), \\
&\mathbf{ D}(\z)H_2,\quad & \arg \z \in&\left(\pi, \frac{3 \pi}{2}\right), \\
&\mathbf{ D}(\z)H_3,\quad & \arg \z \in&\left(\frac{3 \pi}{2}, \frac{7 \pi}{4}\right).
\end{aligned}\right.
$$
It is direct to check that $\Phi^{(\mathrm{PC})}(\z)$ solves the following RH problem (cf. \cite{BI,FIKN}).

\subsection*{RH problem for $\Phi^{(\mathrm{PC})}(\z)$}

\begin{description}
\item{(1)} $\Phi^\mathrm{(PC)}(\z)$ is analytic for all $\z\in \mathbb{C}\setminus \bigcup^4_{k=0}\Upsilon_{k}$, where
    $$
    \Upsilon_{k}=\left\{\zeta\in\mathbb{C}~\Big|~\arg\zeta=\frac{k\pi}{2}\right\},~k=0,1,2,3,
    \quad \Upsilon_{4}=\left\{\zeta\in\mathbb{C}~\Big|~\arg\zeta=-\frac{\pi}{4}\right\};
    $$ see Figure \ref{PC}.
\item{(2)} $\Phi^\mathrm{(PC)}(\z)$ satisfies the jump conditions
$$
\Phi^{(\mathrm{PC})}_+(\z)=\Phi^{(\mathrm{PC})}_-(\z)\left\{\begin{aligned}
&H_{k},\quad && \z \in\Upsilon_k,~k=0,1,2,3, \\
&e^{2\pi i\beta\sigma_3},\quad && \z \in\Upsilon_4.
\end{aligned}\right.
$$
\item{(3)} $\Phi^\mathrm{(PC)}(\z)$ satisfies the following asymptotic behavior as $\z\to\infty$
\begin{align}\label{PCAsyatinfty}
\Phi^\mathrm{(PC)}(\z)=\begin{pmatrix}0 &1 \\ 1 & -\z\end{pmatrix}2^{\frac{\sigma_3}{2}}\begin{pmatrix}
1+O\left(\z^{-2}\right) & \frac{{\beta}}{\z}+O\left(\z^{-3}\right) \\ \frac{1}{\z}+O\left(\z^{-3}\right) & 1+O\left(\z^{-2}\right)\end{pmatrix}
e^{\frac{\z^{2}}{4}\sigma_{3}}\z^{-{\beta}\sigma_3}.
\end{align}
\end{description}

\begin{figure}[t]
  \centering
  \includegraphics[width=7cm,height=7cm]{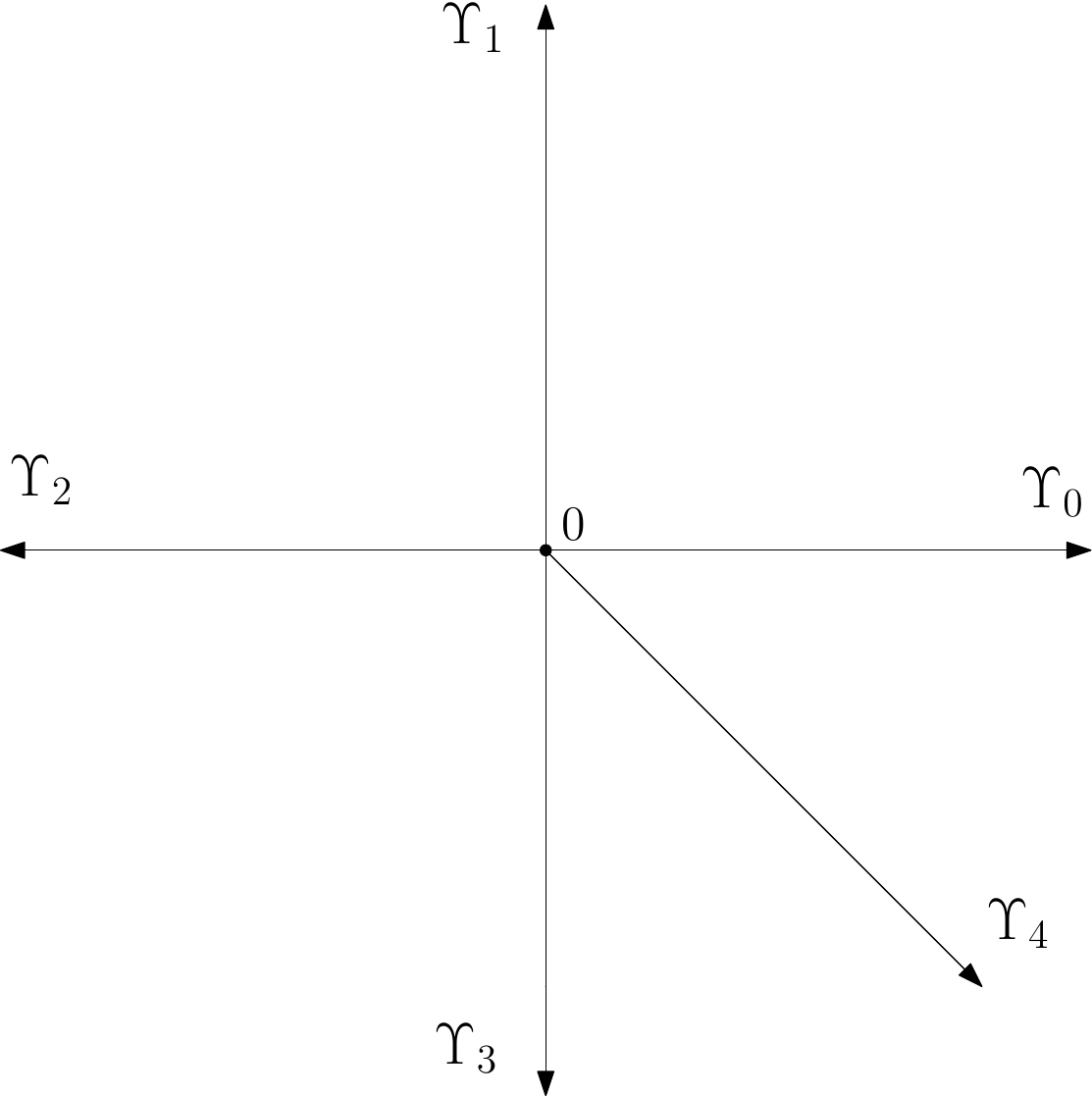}\\
  \caption{The jump contour and jump matrices for $\Phi^{(\mathrm{PC})}$}\label{PC}
\end{figure}

\subsection{Bessel parametrix}\label{Bessel}
Consider the following RH problem.
\subsection*{RH problem for $\Phi^{(\mathrm{Bes})}(\z)$}
\begin{description}
\item{(1)} $\Phi^\mathrm{(Bes)}(\z)$ is analytic for all $\z\in \mathbb{C}\setminus \bigcup^8_{k=1}\Gamma_{k}$, where $$
    \Gamma_{k}=\left\{\zeta\in\mathbb{C}~\Big|~\arg\zeta=\frac{(k-1)\pi}{4}\right\},~k=1,\cdots,8;
    $$ see Figure \ref{Bes}.
\item{(2)} We have the jump conditions
\begin{equation}\label{Besseljump}
\Phi_{+}^{(\mathrm{Bes})}(\z)=\Phi_{-}^{(\mathrm{Bes})}(\z)\left\{\begin{aligned} &\begin{pmatrix}0 & 1 \\ -1 & 0\end{pmatrix}, \quad &\z&\in\Gamma_1\cup\Gamma_5,\\
&\begin{pmatrix}1 & 0 \\ e^{-2\pi i\alpha} & 1\end{pmatrix}, \quad &\z&\in\Gamma_2\cup\Gamma_6,\\
&\begin{pmatrix}e^{\pi i\alpha} & 0 \\ 0 & e^{-\pi i\alpha}\end{pmatrix}, \quad  &\z&\in\Gamma_3\cup\Gamma_7,\\
&\begin{pmatrix}1 & 0 \\ e^{2\pi i\alpha} & 1\end{pmatrix}, \quad  &\z&\in\Gamma_4\cup\Gamma_8.
\end{aligned}\right.
\end{equation}
\item{(3)} As $\z\rightarrow\infty$,
\begin{equation}\label{BesInfty}
\Phi^\mathrm{(Bes)}(\z)=\frac{1}{\sqrt{2}}\begin{pmatrix}
1 & -i \\ -i & 1 \end{pmatrix}(\mathbf{I}+O(\z^{-1}))e^{\frac{\pi i\sigma_3}{4}}e^{-i\z\sigma_3}\left\{
\begin{aligned}
&e^{-\frac{\alpha\pi i\sigma_3}{2}}, & \z &\in \Lambda_1\cup\Lambda_2,\\
&e^{\frac{\alpha\pi i\sigma_3}{2}}, & \z &\in \Lambda_3\cup\Lambda_4,\\
&e^{\frac{\alpha\pi i\sigma_3}{2}}\sigma_1\sigma_3, & \z &\in \Lambda_5\cup\Lambda_6,\\
&e^{-\frac{\alpha\pi i\sigma_3}{2}}\sigma_1\sigma_3, & \z &\in\Lambda_7\cup\Lambda_8.
\end{aligned}\right.
\end{equation}
\end{description}
\begin{figure}[t]
  \centering
 \includegraphics[width=7cm,height=7cm]{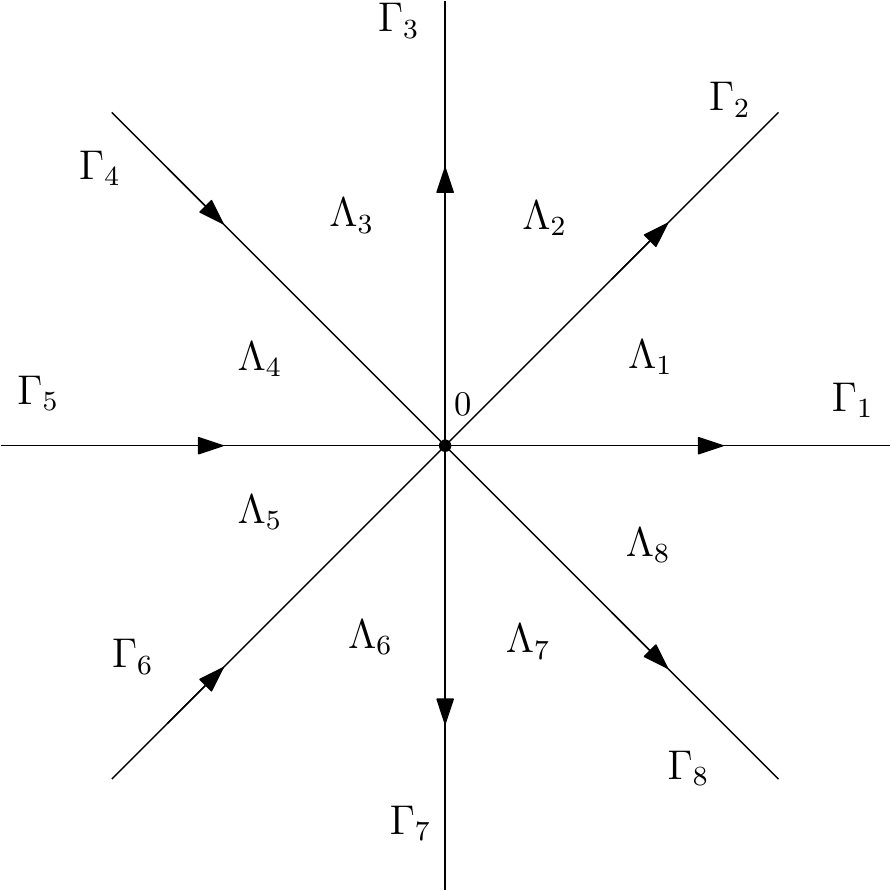}\\
 \caption{The jump contour and regions for $\Phi^{(\mathrm{Bes})}$}\label{Bes}
\end{figure}

From \cite{Van}, we see that the above RH problem can be constructed in terms of the modified Bessel function $I_{\alpha\pm\frac{1}{2}}(\z)$ and $K_{\alpha\pm\frac{1}{2}}(\z)$:
\begin{equation}\label{BesPara}
\Phi^\mathrm{(Bes)}(\z)=\begin{pmatrix}
\sqrt{\pi} \z^{\frac{1}{2}}I_{\alpha+\frac{1}{2}}(\z e^{-\frac{\pi i}{2}}) & -\frac{1}{\sqrt{\pi}}\z^{\frac{1}{2}}K_{\alpha+\frac{1}{2}}(\z e^{-\frac{\pi i}{2}})\\
-i\sqrt{\pi} \z^{\frac{1}{2}}I_{\alpha-\frac{1}{2}}(\z e^{-\frac{\pi i}{2}})&-\frac{i}{\sqrt{\pi}}\z^{\frac{1}{2}}K_{\alpha-\frac{1}{2}}(\z e^{-\frac{\pi i}{2}})
\end{pmatrix}e^{-\frac{1}{2}\alpha\pi i\sigma_3}\end{equation}
for $\z\in\Lambda_2$.
The explicit expressions of $\Phi^\mathrm{(Bes)}(\z)$ in other sectors are determined by \eqref{BesPara}  and the  jump relation \eqref{Besseljump}.

Using the  series expansion of the  modified Bessel function  \cite[(10.25.2)]{NIST}
 $$I_{\mu}(\z)=\left(\frac{\z}{2}\right )^{\mu}\sum_{k=0}^{\infty}\frac{(\frac{ \z^2}{4})^k}{k!\,\Gamma(\mu+k+1)},$$
and the connection formula \cite[(10.27.4)]{NIST}
$$K_{\mu}(\z)=\frac{\pi}{2}\frac{I_{-\mu}(\z)-I_{\mu}(\z) }{\sin(\pi \mu)}, \quad \mu \not\in \mathbb{Z},$$
it is seen  from  \eqref{BesPara} that
\begin{equation}\label{BesParaExpand}
\Phi^\mathrm{(Bes)}(\z)=(I+O(\z))C_\alpha^{\sigma_3}\begin{pmatrix}0 & -1 \\ 1 & 0\end{pmatrix} \z^{\alpha \sigma_3}\begin{pmatrix}
1 & \frac{1}{1+e^{-2\pi i \alpha}} \\ 0& 1
\end{pmatrix},\quad \z\to 0, ~~ \z\in \Lambda_2, \end{equation}
where $\alpha-\frac{1}{2} \not\in  \mathbb{Z} $ and the constant $C_{\alpha}$ is given by
\begin{equation}\label{Calpha}
C_{\alpha}=\pi^{-\frac{1}{2}}2^{\alpha-\frac{1}{2}}e^{\pi i(\alpha+\frac{1}{4})}\Gamma(\alpha+\frac{1}{2}).
\end{equation}

If $\frac{1}{2}-\alpha\in\mathbb{N}$, using the connection formulas \cite[(10.27.1)]{NIST} and \cite[(10.27.3)]{NIST}
\begin{equation*}
 I_{-n}(\z)=I_n(\z),\quad K_{-\mu}(\z)=K_{\mu}(\z),
\end{equation*}
and the small-$z$ asymptotics \cite[(10.30.1)-(10.30.3)]{NIST}
\begin{equation*}
I_{\mu}(\z)\sim \frac{\z^{\mu}}{2^{\mu}\Gamma(\mu+1)},\quad
K_{\mu}(\z)\sim \frac{2^{\mu-1}\Gamma(\mu+1)}{\z^{\mu}},\quad
K_0(\z)\sim -\ln \z,
\end{equation*}
we obtain
\begin{equation}\label{BesParaExpand1}
\Phi^\mathrm{(Bes)}(\z)=\Phi^\mathrm{(Bes)}_{0}(\z)\z^{-\alpha \sigma_3},\quad \z\to 0, \quad \z\in \Lambda_2,
\end{equation}
where $\Phi^\mathrm{(Bes)}_{0}(\z)$ is analytic near the origin.

The behaviors of $\Phi^\mathrm{(Bes)}(\z)$ near the origin in the other regions can be determined by \eqref{BesParaExpand}, \eqref{BesParaExpand1} and the jump relations \eqref{Besseljump}.

\end{appendices}


\end{document}